\documentclass[aps,pra,reprint,superscriptaddress,floatfix]{revtex4-2}
\usepackage{float}
\usepackage{amsthm}
\usepackage{amssymb}
\usepackage{color}
\usepackage{booktabs}
\usepackage{makecell}
\usepackage{tabularx}
\usepackage{array}
\usepackage{graphicx}
\usepackage{capt-of}
\usepackage{stfloats}
\usepackage{multirow}
\usepackage{mathtools}
\usepackage{dcolumn}
\usepackage{bm}
\usepackage{algorithm2e}
\usepackage[T1]{fontenc}
\usepackage{etoolbox}
\usepackage[colorlinks,
            linkcolor=blue,
            anchorcolor=blue,
            citecolor=blue,
            ]{hyperref}
\graphicspath{{figures/}}

\newtheorem{theorem}{Theorem}
\newtheorem{corollary}{Corollary}

\makeatletter
\def\@author@join@#1#2#3{%
  \def\@author{{#1}{\@join{\textsuperscript{,\,}}{#2}{\frontmatter@footnote{#3}}}}%
}
\patchcmd{\doauthor}{\frontmatter@footnote{#2}}{#2}{}%
  {\PackageError{pra-author-notes}{Unable to patch REVTeX author notes}{Check the REVTeX version.}}
\makeatother

\begin{document}

\title{Universal entanglement detection with geometrically optimal witnesses}
\author{Xuanran Zhu}
\thanks{These authors contributed equally to this work.}
\email[Contact author: ]{zhuxr@sustech.edu.cn}
\affiliation{Department of Physics, State Key Laboratory of Quantum Functional Materials, and Guangdong Basic Research Center of Excellence for Quantum Science, Southern University of Science and Technology, Shenzhen 518055, China}

\author{Keyi Huang}
\thanks{These authors contributed equally to this work.}
\affiliation{Department of Physics, State Key Laboratory of Quantum Functional Materials, and Guangdong Basic Research Center of Excellence for Quantum Science, Southern University of Science and Technology, Shenzhen 518055, China}

\author{Chao Zhang}
\affiliation{Department of Physics, The Hong Kong University of Science and Technology, Clear Water Bay, Kowloon, Hong Kong, China}

\author{Xinfang Nie}
\affiliation{Department of Physics, State Key Laboratory of Quantum Functional Materials, and Guangdong Basic Research Center of Excellence for Quantum Science, Southern University of Science and Technology, Shenzhen 518055, China}
\affiliation{Quantum Science Center of Guangdong-HongKong-Macao Greater Bay Area, Shenzhen 518045, China}

\author{Qi Zhao}
\email[Contact author: ]{zhaoqi@cs.hku.hk}
\affiliation{QICI Quantum Information and Computation Initiative, School of Computing and Data Science, The University of Hong Kong, Pokfulam Road, Hong Kong SAR, China}

\author{Dawei Lu}
\email[Contact author: ]{ludw@sustech.edu.cn}
\affiliation{Department of Physics, State Key Laboratory of Quantum Functional Materials, and Guangdong Basic Research Center of Excellence for Quantum Science, Southern University of Science and Technology, Shenzhen 518055, China}
\affiliation{Quantum Science Center of Guangdong-HongKong-Macao Greater Bay Area, Shenzhen 518045, China}

\date{\today}
\begin{abstract}
Constructing optimal entanglement witnesses for arbitrary quantum systems remains a formidable challenge, hindered by the structural complexity of high-dimensional quantum state spaces.
In this work, we overcome this barrier by introducing a unified framework for constructing geometrically optimal witnesses capable of identifying arbitrary entanglement structures. By leveraging the geometric properties of a target quantum state relative to a specific entanglement structure in Hilbert-Schmidt space, these witnesses are constructed in an efficient and effective manner. Among all nontrivial witnesses, they provide maximal distinguishability between the target state and the set of states without the desired entanglement structure. As a result, our approach enables the universal detection of bipartite, multipartite, and high-dimensional entanglement. To demonstrate its experimental feasibility, we apply the framework to a four-qubit nuclear spin quantum system. Compared to existing methods, our approach offers enhanced noise robustness, broader detection universality, and improved construction efficiency, providing a powerful and versatile tool for advancing the characterization and understanding of entanglement in quantum systems.
\end{abstract}

\maketitle

\section{Introduction}
Over the past decades, entanglement has been recognized as a fundamental resource for various quantum information processing tasks, including quantum teleportation \cite{PhysRevLett.70.1895}, quantum computing \cite{steane1998quantum,zhao2025entanglement}, and quantum metrology \cite{PhysRevLett.96.010401}.
A key challenge in both theoretical and experimental contexts lies in the detection of entanglement \cite{friis2019entanglement, guhne2009entanglement, RevModPhys.81.865}. Entanglement witnesses (EWs) \cite{terhal2000bell, PhysRevA.62.052310}, which circumvent the need for full state tomography \cite{cramer2010efficient}, are widely employed for entanglement detection. Their construction can be derived from separability criteria, such as the positive partial transpose (PPT) \cite{PhysRevLett.77.1413} and the computable cross norm or realignment (CCNR) \cite{rudolph2005further, chen2003matrix}. EWs have been further extended to detect more complex entanglement structures, including multipartite entanglement \cite{PhysRevLett.92.087902, PhysRevLett.94.060501, guhne2003experimental, PhysRevX.8.031047, PhysRevLett.111.110503} and high-dimensional entanglement \cite{PhysRevA.63.050301, PhysRevA.107.022431, PhysRevA.61.040301}. However, the construction of optimal EWs is fundamentally hindered, e.g., by the intricate geometry of the separable state set \cite{10.1145/780542.780545, ioannou2006computational}.

Geometric entanglement witnesses (GEWs) provide an intuitive framework for optimal entanglement detection by leveraging the geometry of separable states in a metric space, such as Hilbert-Schmidt space \cite{krammer2009characterizing, PhysRevA.72.052331, PhysRevA.77.024303, PhysRevA.66.032319}. GEWs offer maximal distinguishability between the target state and separable states, which yields intrinsic robustness to noise models such as depolarizing noise. The existence of corresponding GEWs is guaranteed by the Hahn-Banach separation theorem \cite{reed1980methods}. However, due to the intrinsic complexity of separability characterization, prior studies have only succeeded in determining GEWs for specific quantum states with particular symmetries. The construction of GEWs for general states remains a significant challenge. Moreover, the generalization of GEWs for detecting diverse entanglement structures has yet to be systematically explored.

\begin{figure}[H]
    \centering
    \includegraphics[width=0.4\textwidth]{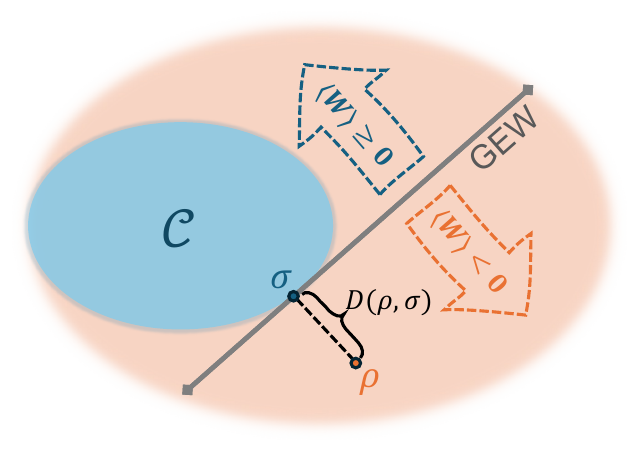}
    \caption{Schematic diagram of GEWs for detecting arbitrary entanglement structures. The convex, compact set $\mathcal{C}$ denotes the set of states without the desired entanglement structure. For example, for genuine multipartite entanglement detection, $\mathcal{C}$ can be chosen as the set of biseparable states. A GEW is constructed based on the nearest state $\sigma \in \mathcal{C}$ to the target entangled state $\rho$, as determined by the Hilbert-Schmidt distance $D$. The hyperplane defined by the GEW separates $\rho$ from $\mathcal{C}$, thereby certifying the entanglement structure of $\rho$.}
    \label{fig:geometric_witness}
\end{figure}

In this article, we present a universal framework for constructing GEWs capable of detecting arbitrary entanglement structures~\cite{PhysRevX.8.021072, zhou2019detecting, PhysRevLett.113.260502, PhysRevLett.110.030501}. The framework exploits the geometric properties of a target quantum state relative to the given entanglement structure in the Hilbert-Schmidt space, as illustrated in Fig.~\ref{fig:geometric_witness}. Compared with previous GEW studies, which focused mainly on special states or specific scenarios, our work formulates the closest-state construction for a general convex compact set $\mathcal{C}$ associated with the entanglement structure to be certified. We further develop the dynamic refinement of extremes approximation model (DREAM) to make this construction computationally accessible, and realize the resulting witnesses experimentally on a four-qubit nuclear spin system. We demonstrate this unified framework for bound entanglement~\cite{PhysRevLett.80.5239}, multipartite entanglement~\cite{walter2016multipartite}, and high-dimensional entanglement~\cite{PhysRevLett.127.220501, erhard2020advances, bavaresco2018measurements}. Experimental results show excellent agreement with theoretical predictions, underscoring the advantages of the framework in terms of detection universality, noise resilience, and construction efficiency when compared to existing methods.

\section{Geometric entanglement witnesses}
We begin with the basic definitions in bipartite systems. A bipartite state $\rho_{AB}$ in $\mathcal{H}_A \otimes \mathcal{H}_B$ is separable if it can be expressed as a convex combination of product pure states:
\begin{equation}
\rho_{AB} = \sum_i \lambda_i |\phi_A^{(i)}\rangle\langle \phi_A^{(i)}| \otimes |\phi_B^{(i)}\rangle\langle \phi_B^{(i)}|,
\label{eq:separable_state}
\end{equation}
where $0 < \lambda_i \leq 1$ and $\sum_i \lambda_i = 1$. Otherwise, $\rho_{AB}$ is entangled. An EW is a Hermitian operator $W$ satisfying $\langle W \rangle_\sigma = \mathrm{Tr}[W \sigma] \geq 0$ for all separable states $\sigma$, with at least one entangled state $\rho$ such that $\langle W \rangle_\rho < 0$. If there exists a separable state $\sigma_0$ such that $\langle W \rangle_{\sigma_0} = 0$, then $W$ is called an optimal EW. Geometrically, the hyperplane defined by such a witness is tangent to the set of separable states at $\sigma_0$.

For a given target state $\rho$, we call an optimal EW geometrically optimal if its tangent separable state is the exact nearest separable state $\sigma^\star$ to $\rho$ under the Hilbert--Schmidt (HS) distance, i.e., $\sigma^\star=\arg\min_{\sigma\in\mathrm{SEP}}\|\rho-\sigma\|$. This is the geometric entanglement witness (GEW) construction introduced in previous literature~\cite{krammer2009characterizing,PhysRevA.72.052331,PhysRevA.77.024303,PhysRevA.66.032319}. The corresponding witness is given by
\begin{equation}
    W(\sigma^{\star}) = \frac{(\sigma^{\star} - \rho) - \mathrm{Tr}[\sigma^{\star}(\sigma^{\star} - \rho)] I}{\|\sigma^{\star} - \rho\|}.
    \label{eq:optimal_witness}
\end{equation}
By construction, $\langle W(\sigma^\star)\rangle_{\sigma^\star}=0$, so it is optimal according to the definition above; because the tangent state $\sigma^\star$ is also the closest separable state to $\rho$ in HS distance, it is geometrically optimal.

Importantly, for a target state on a $d$-dimensional Hilbert space, GEWs provide maximal distinguishability among all nontrivial EWs: under a trace-free norm constraint, they maximize the expectation-value gap between the target state and the separable set, i.e.,
\begin{equation}
    W(\sigma^\star) = \mathop{\arg\max}_{\|W \|_{\text{tf}} \leq 1} \left(\min_{\sigma \in \mathrm{SEP}} \langle W \rangle_{\sigma} - \langle W \rangle_{\rho}\right),
\end{equation}
where the trace-free norm constraint, $\|W \|_{\text{tf}} = \|W - \mathrm{Tr}[W] I/d\| \leq 1$, prevents trivial scaling. Furthermore, it can be shown that the maximal distinguishability is directly related to the HS measure of the given state \cite{PhysRevA.66.032319}, i.e., $\langle W(\sigma^\star) \rangle_{\rho} = -D(\rho)$.

As a result, the maximal distinguishability of GEWs establishes their intrinsic robustness in detecting noisy entangled states. For example, when the effect of depolarizing noise is considered, the visibility threshold of GEWs is determined by the following theorem, whose proof is provided in Appendix~\ref{app:threshold-proof}:

\begin{theorem}[Robustness of GEWs]
    Consider a target state $\rho$ affected by depolarizing noise, described as $\mathcal{N}_p(\rho)=p\rho + (1-p)\frac{I}{d}$. The visibility threshold $p_{\text{th}}$ for a witness is the minimum visibility such that all noisy states with $p > p_{\text{th}}$ can be detected as entangled. For the GEW constructed for the target state $\rho$, the visibility threshold is given by
    \begin{equation}
        p_{\text{GEW}} = \frac{\mathrm{Tr}[\sigma^\star (\rho - \sigma^\star)]}{\mathrm{Tr}[\rho (\rho - \sigma^\star)]} \leq 1 - \frac{D(\rho)}{\sqrt{\gamma(\rho)}},
    \end{equation}
    where $\sigma^\star$ is the nearest separable state to $\rho$ under the HS distance, $D(\rho)$ is the HS measure of $\rho$, and $\gamma(\rho)$ denotes the purity of $\rho$.
    \label{theo:threshold}
\end{theorem}

This theorem establishes a direct link between the geometric measure and noise tolerance, providing a sufficient visibility threshold for entanglement detection under depolarizing noise. The following corollary further shows that, for bipartite pure states, the GEW is strictly more robust than the fidelity-based witness, the proof is provided in Appendix~\ref{app:corollary-proof}.

\begin{corollary}
For detecting entanglement in any bipartite pure state, the GEW is strictly more robust than the fidelity-based witness under depolarizing noise.
\label{coro:compare-fidelity-witness}
\end{corollary}

The above results extend to the detection of general entanglement structures, as illustrated in Fig.~\ref{fig:geometric_witness} and established in Appendix~\ref{app:general-proof}.

\section{Numerical construction}
\label{sec:construction}
In this work, we introduce a numerical algorithm for approximating the HS measure and a nearest separable state, denoted by \(\hat{\sigma}\), leveraging a dynamic characterization of separable states. Because the parametrized optimization is nonconvex, \(\hat{\sigma}\) is not accompanied by an independent certificate that it is the exact global minimizer \(\sigma^{\star}\).

For the set of separable states described in Eq.~\eqref{eq:separable_state}, we decompose it into three components: (i) the probability distribution $\{\lambda_i\}$, (ii) the local states $\{|\phi_A^{(i)}\rangle\}$, and (iii) the local states $\{|\phi_B^{(i)}\rangle\}$. A natural approach is to reformulate the constrained optimization problem over separable states into an unconstrained optimization problem by introducing a set of free parameters $\vec{\theta}$ to parametrize the separable states. This allows the distance to be expressed as
$
D(\rho) = \min_{\vec{\theta}} \|\rho - \sigma(\vec{\theta})\|,
$
where $\sigma(\vec{\theta})$ represents the parametrized separable state.

The probability distribution $\{\lambda_i\}$ resides on a probability simplex, which can be conveniently parametrized using the softmax function \cite{goodfellow2016deep}:
\begin{equation}
    \lambda_i = \frac{e^{\theta_i}}{\sum_j e^{\theta_j}}, \quad i=1,2,\ldots,n,
\end{equation}
where $\{\theta_i\}$ are the free real parameters.

The local states $\{|\phi_A^{(i)}\rangle\}$ and $\{|\phi_B^{(i)}\rangle\}$ are constrained to the unit spheres in their respective Hilbert spaces. These states can be parametrized by complex vectors normalized to unit length:
\begin{equation}
    |\phi\rangle = \frac{|\tilde{\phi}\rangle}{\sqrt{\langle \tilde{\phi}|\tilde{\phi}\rangle}}, \quad |\tilde{\phi}\rangle = \sum_{k=0}^{d-1} (\theta_{2k} + i \theta_{2k+1}) |k\rangle,
\end{equation}
where $d$ is the local dimension of the subsystem, and $\{\theta_{2k}, \theta_{2k+1}\}$ represent the free parameters of the unnormalized state vector $|\tilde{\phi}\rangle$.

For an $n$-entry separable state in a $d_A \otimes d_B$ composite system, where $d_A$ and $d_B$ denote the local dimensions, the total number of free parameters required is $(2(d_A + d_B) + 1)n$. Here, the number $n$ for any separable state can be upper bounded by $d_A^2 d_B^2$ due to Carathéodory's theorem \cite{caratheodory1911variabilitatsbereich}, ensuring that the optimization process involves polynomial parameter complexity.

Once the parametrization is established, a dynamic convex hull approximation (CHA) of the separable states is constructed. This approach provides a substantial improvement in efficiency compared to previous CHA methods, which relied on the random sampling of a large number of product states \cite{PhysRevA.98.012315, hou2020upper}. To further refine the approximation, gradient-based optimization techniques are employed to minimize the distance and obtain a numerical candidate $\hat{\sigma}$ for the closest separable state. Notably, our approach can be seamlessly generalized to locate the nearest state within any other convex compact sets $\mathcal{C}$, enabling the detection of arbitrary entanglement structures, which we refer to as the DREAM model. Optimization details for the different scenarios are provided in Appendix~\ref{app:dream}.

\section{Experimental realization}
With the construction of GEWs, we employ them to experimentally detect target entangled states on a four-qubit nuclear magnetic resonance (NMR) platform \cite{PhysRevLett.129.070502,PhysRevLett.132.210403,PhysRevLett.133.040401,y8wr-yml4,huang2025direct}. The sample used is \(^{13}\)C-labeled trans-crotonic acid dissolved in \(d_6\)-acetone, where four \(^{13}\)C nuclear spins constitute a four-qubit quantum system. The internal Hamiltonian of the system is described as
\begin{equation}
    \mathcal{H}_{\text{NMR}} = -\sum_{i=1}^4 \pi \nu_i \sigma_z^i + \sum_{i<j}^4 \pi \frac{J_{ij}}{2} \sigma_z^i \sigma_z^j,
\end{equation}
where \(\nu_i\) denotes the chemical shift of each spin, and \(J_{ij}\) represents the coupling strength between spins \(i\) and \(j\). Single-qubit gates are implemented using radio-frequency (rf) pulses with tailored shapes, while two-qubit gates are achieved via the natural coupling between qubits. To prepare the target quantum states in our experiments, shaped rf pulses were optimized using gradient ascent pulse engineering \cite{KHANEJA2005296}. The GEWs are computed from the intended target states specified by these preparation protocols, rather than reconstructed from tomography of unknown experimental states. The expectation values of the GEW on the target state were then measured through local Pauli measurements. All experiments were performed on a Bruker 300-MHz spectrometer at room temperature. Additional experimental details are given in Appendix~\ref{app:experiment}.

\section{Bound entanglement}
\label{sec:bound}
We first demonstrate the detection of bound entanglement, a special type of entanglement that cannot be distilled into Bell states via local operations and classical communication. It is well established that PPT entangled states are bound entangled \cite{PhysRevLett.80.5239,horodecki1997separability} and cannot be detected by fidelity-based witnesses \cite{PhysRevLett.126.140503}. Here, we consider a one-parameter $3 \otimes 3$ quantum state $\rho_b$ introduced by the Horodecki family \cite{PhysRevLett.82.1056}. It is known that $\rho_b$ is separable for $2 \leq b \leq 3$, bound entangled for $1 \leq b < 2$ and $3 < b \leq 4$, and free entangled for $0 \leq b < 1$ and $4 < b \leq 5$ \cite{note1}.

\begin{figure}[h]
    \centering
    \includegraphics[width=0.45\textwidth]{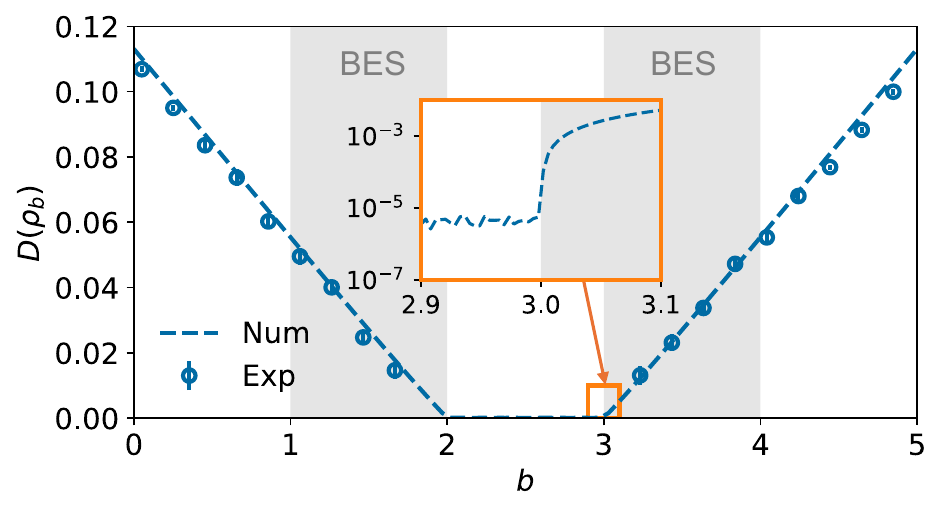}
    \caption{Numerical and experimental results for the Horodecki states \(\rho_b\). A positive value indicates the presence of entanglement. The dashed line represents the numerical results of the HS measure. Dots denote the mean expectation values of the GEWs, multiplied by a negative sign, from repeated experiments, while the error bars represent the standard deviations. The bound entangled state (BES) regions are shaded in gray.}
    \label{fig:bound_entanglement}
\end{figure}

We apply our method to compute the HS measure \(D(\rho_b)\) numerically and obtain the expectation values of the GEWs from the experimental data. As shown in Fig.~\ref{fig:bound_entanglement}, our method accurately identifies the boundary between the bound entangled region and the separable states, exhibiting strong alignment between the numerical and experimental results. This demonstrates the capability of the method to detect such subtle entanglement structures.

We can further compare our method with the Gilbert algorithm in detecting bound entangled states constructed from the unextendible product basis \cite{PhysRevLett.82.5385}. For the Tiles bound entangled state $\rho_{\text{Tiles}}$, the HS measure computed by our method is around $0.043455$, whereas Ref.~\cite{PhysRevA.102.012409} reports $D^2\simeq0.002177$ for the Gilbert algorithm, equivalently $D\simeq0.046658$. Under the same distance convention, our result gives a smaller HS distance. This implies that our method can identify a closer separable state to the target state, thereby constructing a more effective EW. Additionally, our numerical results reveal some novel symmetries of the nearest separable state to $\rho_{\text{Tiles}}$ in Appendix~\ref{app:tiles}.

We also compare the noise tolerance of our GEW with that of the operator-Schmidt-decomposition (OSD) witness~\cite{PhysRevLett.133.040203} by introducing depolarizing noise to $\rho_{\text{Tiles}}$. According to Theorem~\ref{theo:threshold}, the GEW demonstrates a higher robustness to noise, tolerating a lower visibility ($p_{\text{GEW}} \approx 0.8687$) compared to the OSD witness ($p_{\text{OSD}} \approx 0.8908$). Notably, this robustness is achieved without incorporating the noise model into the optimization process, unlike the OSD witness, which explicitly accounts for it.

\section{High-dimensional entanglement}
The entanglement dimensionality of a bipartite state $\rho$ is quantified by its Schmidt number \cite{terhal2000bell}:
\begin{equation}
    \mathrm{SN}(\rho) = \min_{\{\lambda_i, |\psi_i\rangle\}} \max_i \mathrm{SR}(|\psi_i\rangle),
\end{equation}
where $\mathrm{SR}(|\psi_i\rangle)$ denotes the Schmidt rank of the pure state $|\psi_i\rangle$, and $\rho = \sum_i \lambda_i |\psi_i\rangle\langle \psi_i|$ represents the possible convex decomposition of $\rho$.

To identify $k$-entangled states whose Schmidt numbers are at least $k$, one can construct a Schmidt number witness $W_k$, which generalizes the concept of an EW. Specifically, $W_k$ satisfies the condition $\langle W_k \rangle_{\sigma} \geq 0$ for all states $\sigma$ with Schmidt number less than $k$, while there exists at least one $k$-entangled state $\rho$ such that $\langle W_k \rangle_{\rho} < 0$.

In this scenario, the convex set $\mathcal{C}$ can be chosen as the set of all states with a Schmidt number less than $k$ to construct GEWs. The parametrization can be performed analogously by replacing the product pure states with states that have Schmidt ranks less than $k$ (see Appendix~\ref{app:dream-parameterization}).

Here, we first consider the $d \otimes d$ isotropic states, which are defined as follows \cite{PhysRevA.59.4206}:
\begin{equation}
    \rho_{I}(F) = F |\Psi_d^+\rangle\langle \Psi_d^+| + (1-F) \frac{I - |\Psi_d^+\rangle\langle \Psi_d^+|}{d^2 - 1},
    \label{eq:isotropic_state}
\end{equation}
where $|\Psi_d^+\rangle = \frac{1}{\sqrt{d}} \sum_{i=0}^{d-1} |ii\rangle$ is the maximally entangled state in $d \otimes d$ systems. We find that GEWs for isotropic states are in the form of fidelity-based witnesses, with the proof given in Appendix~\ref{app:isotropic-proof}:
\begin{theorem}[GEW for isotropic state]
For the isotropic state $\rho_I(F)$ in $d \otimes d$ systems, defined in Eq.~\eqref{eq:isotropic_state}, its Schmidt number $\mathrm{SN}(\rho_I(F)) = k$ for $F \in \left(\frac{k-1}{d}, \frac{k}{d}\right]$, and can be certified by the corresponding GEW:
\begin{equation}
    W_k = c \cdot \left(\frac{k-1}{d}I - |\Psi_d^+\rangle\langle \Psi_d^+| \right),
\end{equation}
where the coefficient $c=\frac{d}{\sqrt{d^2-1}}$.
\label{theo:isotropic}
\end{theorem}

We leverage the above analytical results to evaluate the performance of our method across various dimensions. The findings indicate that, for quantum systems of moderate size, the computation time scales mildly with increasing local dimension and remains exceptionally low. Furthermore, the computational error is consistently maintained below $10^{-9}$. These results underscore the efficiency and accuracy of our method in handling quantum systems of moderate size. Detailed results are provided in Appendix~\ref{app:dream-benchmark}.

\begin{figure}[H]
    \centering    \includegraphics[width=0.45\textwidth]{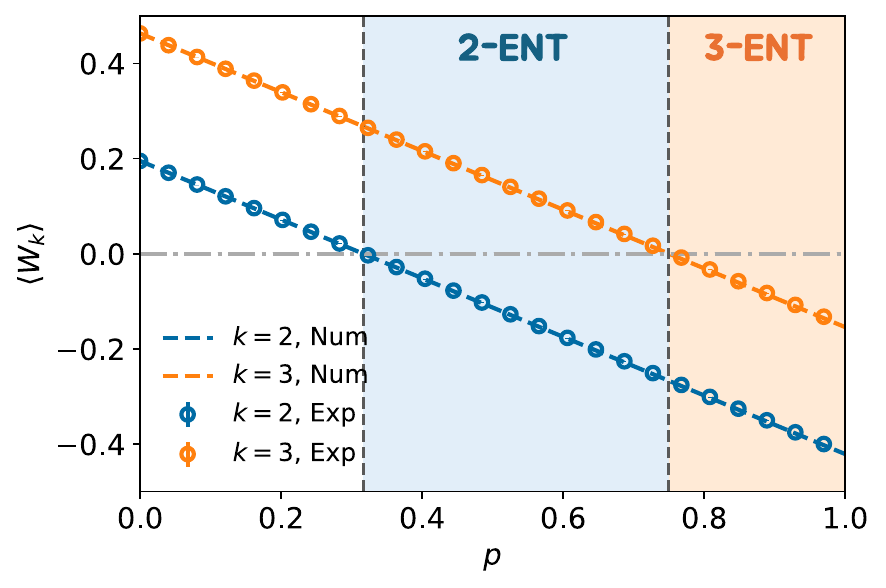}
    \caption{Numerical and experimental results for the state $\varrho$ in Eq.~(\ref{eq:varrho}) under depolarizing noise with different visibility $p$. Here, $W_k$ represents the constructed GEW for the target state $\varrho$, and a negative value of \(\langle W_k \rangle\) indicates the presence of \(k\)-entangled states. The dashed lines correspond to the numerical results, while the dots denote the experimental results obtained from repeated experiments. The blue and orange regions indicate the areas of different \(k\)-entangled states detected by the GEW.}
    \label{fig:high_dimensional_entanglement}
\end{figure}

It is worth noting that for general states lacking such high symmetry, fidelity-based witnesses fail to detect their Schmidt numbers. For example, the following state in a $4 \otimes 4$ system \cite{PhysRevLett.124.200502}:
\begin{equation}
    \varrho = \frac{1}{2} \left|\Psi_{3}^+\right\rangle \left\langle \Psi_{3}^+\right| + \frac{1}{2} \left(\frac{|23\rangle + |32\rangle}{\sqrt{2}}\right) \left(\frac{\langle 23| + \langle 32|}{\sqrt{2}}\right),
\label{eq:varrho}
\end{equation}
whose Schmidt number equals three but cannot be certified by fidelity-based witnesses. Although OSD witnesses can detect its Schmidt number, the efficient construction of OSD witnesses for arbitrary Schmidt numbers remains unsolved. In contrast, our method provides a more intuitive and straightforward approach to constructing GEWs for Schmidt number witnesses, making it significantly easier to detect states with higher Schmidt numbers.

To demonstrate the capabilities of our method, we investigate the noise tolerance of our GEWs in detecting the Schmidt number of the aforementioned state $\varrho$ under depolarizing noise $\mathcal{N}_p$. As shown in Fig.~\ref{fig:high_dimensional_entanglement}, numerical and experimental results clearly demonstrate the transition from separable states to $2$-entangled states and subsequently to $3$-entangled states as the visibility $p$ increases. Our GEWs exhibit significantly better noise tolerance compared to the CCNR witness. Specifically, our GEWs can detect $3$-entangled states at a much lower visibility threshold ($p_{\text{GEW}} \approx 0.75$) than the CCNR witness ($p_{\text{CCNR}} \approx 0.83$), as reported in \cite{PhysRevLett.133.040203}. This enhanced robustness highlights the practical advantages of our method in detecting high-dimensional entanglement under realistic noise conditions.

\section{Multipartite entanglement}
\label{sec:multipartite}
GEWs can also be applied to detect multipartite entanglement. Here, we focus on the detection of genuinely entangled states that cannot be expressed as a convex combination of biseparable states across different bipartitions. To detect genuine entanglement, we can choose the convex set $\mathcal{C}$ as the set of all biseparable states. The parametrization can be done similarly by replacing the product pure states with biproduct pure states under various bipartitions (see Appendix~\ref{app:dream-parameterization}).

We consider the following $n$-qubit states under depolarizing noise defined as $\mathcal{N}_p(|\psi\rangle) = p |\psi\rangle\langle\psi| + (1-p)\frac{I}{2^n}$, where $p$ denotes the visibility and $I$ is the identity operator in the $2^n$-dimensional Hilbert space. The states investigated include the three-qubit W state $|W_3\rangle$ and the four-qubit W state $|W_4\rangle$. As shown in Fig.~\ref{fig:genuine_entanglement}, for each state, the visibility threshold given by GEWs is lower than that of the fidelity-based witnesses and OSD witnesses \cite{PhysRevLett.133.040203}. Moreover, when the depolarizing noise model is explicitly incorporated into the optimization described in Appendix~\ref{app:newton}, the approximate numerical estimate for the three-qubit W state threshold changes from $0.541$ to about $0.478$, numerically approaching the known benchmark reported in Ref.~\cite{PhysRevLett.106.190502}.

\begin{figure}[H]
    \centering
    \includegraphics[width=0.45\textwidth]{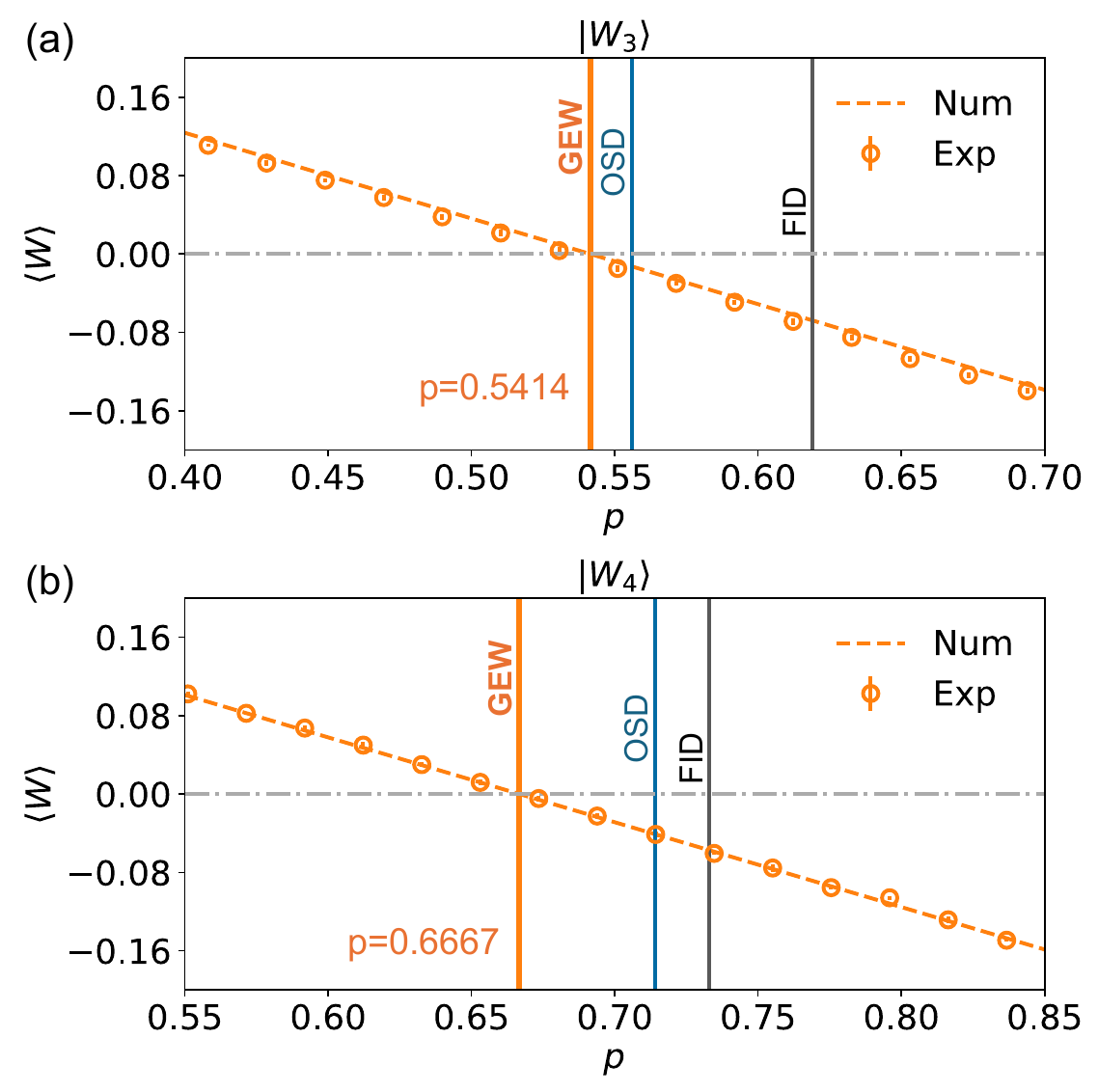}
    \caption{Numerical and experimental results for the noisy state \(\mathcal{N}_p(|\psi\rangle) = p |\psi\rangle\langle\psi| + (1-p)\frac{I}{2^n}\) for (a) three-qubit W state and (b) four-qubit W state. A negative expectation value of the witness indicates the presence of genuine entanglement. The dashed lines represent the numerical results, while dots denote the experimental results from repeated experiments. The solid black and blue lines are the visibility thresholds given by fidelity-based witnesses and OSD witnesses, while the orange ones represent the thresholds from GEWs.}
    \label{fig:genuine_entanglement}
\end{figure}

Besides the bipartition of multipartite states, GEWs can also be applied to finer partitions, such as the detection of completely entangled states. In this scenario, the convex set $\mathcal{C}$ is chosen as the set of fully separable states under the given partition. Further results are provided in Appendix~\ref{app:additional-results}.

\section{Conclusion}
We propose a universal method for constructing GEWs that can detect arbitrary entanglement structures of target states while achieving maximal distinguishability among all nontrivial witnesses. This enhanced distinguishability improves the noise tolerance of our detection method, as demonstrated for depolarizing noise, compared to existing approaches.

In addition to its theoretical strengths, our approach demonstrates practical feasibility in quantum experiments, providing a versatile tool for entanglement detection in real-world scenarios. Furthermore, the proposed method can be extended to detect other quantum resources \cite{RevModPhys.91.025001,zhu2025unified}, highlighting its broad applicability in quantum information science.

\begin{acknowledgments}
X.N. and D.L. are supported by the National Natural Science Foundation of China
(Grants No. 12574543 and No. 12575020), the Guangdong Provincial Quantum Science Strategic Initiative (Grants No. GDZX2203001, No. GDZX2303001, No. GDZX2403004, No. GDZX2503001, and No. GDZX2506002),
and the Guangdong Basic and Applied Basic Research Foundation (Grants No. 2025A1515011599 and No. 2026B1515020012).
\end{acknowledgments}

\section*{Data Availability}
The code and data for this work are open source and available at \cite{repo}.

\appendix
\counterwithin*{theorem}{section}
\renewcommand{\thetheorem}{\thesection\arabic{theorem}}

\setlength{\textfloatsep}{10pt plus 2pt minus 2pt}
\setlength{\floatsep}{10pt plus 2pt minus 2pt}
\setlength{\dbltextfloatsep}{12pt plus 2pt minus 2pt}
\setlength{\dblfloatsep}{10pt plus 2pt minus 2pt}

\section{Proof of Theorem~\ref{theo:threshold}}
\label{app:threshold-proof}

\begin{proof}
For the target state $\rho$, the corresponding GEW has the following form:
\[
W = \frac{(\sigma^{\star} - \rho) - \mathrm{Tr}[\sigma^{\star}(\sigma^{\star} - \rho)] I}{\|\sigma^{\star} - \rho\|},
\]
where $\sigma^\star$ is the nearest separable state to $\rho$ under the HS measure. The visibility threshold $p_{\text{GEW}}$ satisfies
\[
\mathrm{Tr}\left[W\left(p_{\text{GEW}}\rho+(1-p_{\text{GEW}})\frac{I}{d}\right)\right]=0.
\]
From the above equations, we can derive that
\[
p_{\text{GEW}}=\frac{\mathrm{Tr}[\frac{W}{d}]}{\mathrm{Tr}[\frac{W}{d}]-\mathrm{Tr}[W\rho]}=\frac{\mathrm{Tr}[\sigma^\star(\rho-\sigma^\star)]}{\mathrm{Tr}[\rho(\rho-\sigma^\star)]}.
\]
For simplicity, we denote the inner product of two matrices as $\langle A,B\rangle=\mathrm{Tr}[A^\dagger B]$, so that
\[
p_{\text{GEW}}=1-\frac{\mathrm{Tr}[(\rho-\sigma^\star)^2]}{\mathrm{Tr}[\rho(\rho-\sigma^\star)]}=1-\frac{\|\rho-\sigma^\star\|^2}{\langle\rho,\rho-\sigma^\star\rangle}.
\]
Since $\langle A,B\rangle\leq \|A\|\cdot\|B\|$, we obtain that
\[
p_{\text{GEW}}\leq 1 -\frac{\|\rho-\sigma^\star\|}{\|\rho\|}=1-\frac{D(\rho)}{\sqrt{\gamma(\rho)}},
\]
where $D(\rho)$ is the HS measure of $\rho$ and $\gamma(\rho)=\mathrm{Tr}[\rho^2]$ is the purity of $\rho$.
\end{proof}
\section{Proof of Corollary~\ref{coro:compare-fidelity-witness}}
\label{app:corollary-proof}
\begin{proof}
    For any bipartite pure state $|\psi\rangle$ in $d_A\otimes d_B$, the resulting state after depolarizing channel is $\mathcal{N}_p(|\psi\rangle)=p|\psi\rangle\langle\psi|+(1-p)\frac{I}{d}$, where $d=d_Ad_B$ is the total dimension. The corresponding fidelity-based witness is given by
    \[
    W = \mu_1(|\psi\rangle)^2I-|\psi\rangle\langle\psi|,
    \]
    where $\mu_1$ denotes the largest Schmidt coefficient of $|\psi\rangle$. Then it is easy to derive that the visibility threshold is
    \[
    p_{\text{FID}}=\frac{\mu_1^2-\frac{1}{d}}{1-\frac{1}{d}}=1-\frac{1-\mu_1^2}{1-\frac{1}{d}}.
    \]
    Next, we aim to prove that $D(|\psi\rangle)>\frac{1-\mu_1^2}{1-\frac{1}{d}}$, such that
    \[
    p_{\text{GEW}}\leq 1-D(|\psi\rangle)<p_{\text{FID}}.
    \]
    First, it is easy to prove that $\max_{\sigma\in\mathrm{SEP}}\mathrm{Tr}[\rho_\psi\sigma]=\max_{\sigma\in\mathrm{SEP}}\langle\psi|\sigma|\psi\rangle=\mu_1^2$. Then for any $\sigma \in \mathrm{SEP}$, we have
    \[
    \mathrm{Tr}[(\rho_\psi-\frac{I}{d})(\rho_\psi-\sigma)]=1-\mathrm{Tr}[\rho_\psi\sigma]\geq 1-\mu_1^2.
    \]
    On the other hand, by Cauchy-Schwarz inequality for the Hilbert-Schmidt inner product, we have
    \[
    \begin{aligned}
        \mathrm{Tr}[(\rho_\psi-\frac{I}{d})(\rho_\psi-\sigma)] &\leq \|\rho_\psi-\frac{I}{d}\|\cdot\|\rho_\psi-\sigma\|\\
        &=\sqrt{1-\frac{1}{d}}\|\rho_\psi-\sigma\|.
    \end{aligned}
    \]
    Finally, we obtain that
    \[
    D(|\psi\rangle)=\min_{\sigma\in \mathrm{SEP}}\|\rho_\psi-\sigma\|\geq \frac{1-\mu_1^2}{\sqrt{1-\frac{1}{d}}}>\frac{1-\mu_1^2}{1-\frac{1}{d}}.
    \]
\end{proof}
\section{Proof of Theorem~\ref{theo:isotropic}}
\label{app:isotropic-proof}
\begin{proof}
For simplicity, here, we denote the set of states with Schmidt number no larger than $k$ as $S_k$ and the inner product of two operators $A$ and $B$ as $\langle A, B \rangle = \mathrm{Tr}[A^{\dagger} B]$.

From Ref.~\cite{PhysRevA.63.050301}, $\mathrm{SN}(\rho_I(F))= k$ for $\frac{k-1}{d}<F\leq\frac{k}{d}$. This suggests that the nearest state in $S_{k-1}$ to $\rho_I(F)$ under the HS distance is $\rho_I(\frac{k-1}{d})$. It is therefore sufficient to prove that the following witness constructed from Eq.~\eqref{eq:optimal_witness}
\begin{equation}
    W_k=\frac{(\rho_I(\frac{k-1}{d})-\rho_I(F))-\langle\rho_I(\frac{k-1}{d}),\rho_I(\frac{k-1}{d})-\rho_I(F)\rangle I}{\|\rho_I(\frac{k-1}{d})-\rho_I(F)\|}
    \label{eq:witness-isotropic}
\end{equation}
is effective \cite{PhysRevA.72.052331} for $F>\frac{k-1}{d}$. Using the following results:
\[
    \begin{aligned}
        \langle\rho_I(F_1),\rho_I(F_2)\rangle=F_1F_2+\frac{(1-F_1)(1-F_2)}{d^2-1}, \\
        \|\rho_I(F_1)-\rho_I(F_2)\|=\frac{d}{\sqrt{d^2-1}}|F_1-F_2|,
    \end{aligned}
\]
we can simplify the witness in Eq.~(\ref{eq:witness-isotropic}) to the following form:
\[
    W_k=\frac{d}{\sqrt{d^2-1}}\left(\frac{k-1}{d}I-|\Psi_d^+\rangle\langle\Psi^+_d|\right).
\]

For isotropic states $\rho_I(F)$ with $F>\frac{k-1}{d}$, we can obtain that
\[
    \mathrm{Tr}[W_k\rho_I(F)]=-\frac{d}{\sqrt{d^2-1}}(F-\frac{k-1}{d})<0,
\]
while for quantum states $\sigma \in S_{k-1}$, according to Lemma 1 from Ref.~\cite{PhysRevA.63.050301}, we have
\[
    \mathrm{Tr}[W_k\sigma]=\frac{k-1}{\sqrt{d^2-1}}(1-\frac{d}{k-1}\langle\Psi^+_d|\sigma|\Psi^+_d\rangle)\geq 0.
\]
Thus, we prove that the corresponding witness is effective.
\end{proof}
\section{Detection of arbitrary entanglement structures}
\label{app:general-proof}
Here, we define the HS measure for a quantum state $\rho$ with respect to a convex and compact set $\mathcal{C}$ as
\begin{equation}
    D_{\mathcal{C}}(\rho) = \min_{\sigma \in \mathcal{C}} \|\rho - \sigma\|.
\end{equation}
In this framework, a witness $W_{\mathcal{C}}$ for detecting states outside the set $\mathcal{C}$ is defined as a Hermitian operator satisfying $\langle W_{\mathcal{C}} \rangle_{\sigma} \geq 0$ for all states $\sigma \in \mathcal{C}$, with the additional property that there exists at least one state $\rho \notin \mathcal{C}$ such that $\langle W_{\mathcal{C}} \rangle_{\rho} < 0$.

For any state $\rho$ outside the set $\mathcal{C}$, the Hahn-Banach separation theorem ensures the existence of a witness \cite{reed1980methods}:
\begin{theorem}
    Given a convex and compact subset $\mathcal{C}$ of the quantum state space, for any state $\rho \notin \mathcal{C}$, there exists a valid witness $W_{\mathcal{C}}$ for detecting it.
\end{theorem}

The maximal distinguishability of a target state $\rho$ from the set $\mathcal{C}$ is defined as
\begin{equation}
    B_{\mathcal{C}}(\rho) = \max_{\|W_{\mathcal{C}} - \mathrm{Tr}[W_{\mathcal{C}}]I/d\| \leq 1} (\min_{\sigma \in \mathcal{C}} \langle W_{\mathcal{C}} \rangle_{\sigma} - \langle W_{\mathcal{C}} \rangle_{\rho}),
\end{equation}
where $d$ is the Hilbert-space dimension. This quantity is the maximum gap between the minimum expectation value of $W_{\mathcal{C}}$ over all states in $\mathcal{C}$ and the expectation value of $W_{\mathcal{C}}$ for the state $\rho$. The norm constraint on the trace-free part of $W_{\mathcal{C}}$ prevents trivial scaling. When the witness is shifted so that its minimum expectation value on $\mathcal{C}$ is zero, the expression reduces to
\begin{equation}
    B_{\mathcal{C}}(\rho) = \max_{\|W_{\mathcal{C}} - \mathrm{Tr}[W_{\mathcal{C}}]I/d\| \leq 1} -\langle W_{\mathcal{C}} \rangle_{\rho}.
\end{equation}
The following theorem establishes the alignment between the HS measure and the maximal distinguishability. Furthermore, it provides a geometric approach to constructing a geometrically optimal witness in the precise sense defined above.
\begin{theorem}
    For any convex compact set $\mathcal{C}$ and quantum state $\rho \notin \mathcal{C}$, the following relation holds:
    \begin{equation}
        D_{\mathcal{C}}(\rho) = B_{\mathcal{C}}(\rho),
    \end{equation}
    and the geometrically optimal witness $W_{\mathcal{C}}$ that achieves the maximum in $B_{\mathcal{C}}(\rho)$ is given by
    \begin{equation}
        W_{\mathcal{C}}(\sigma^{\star}) = \frac{(\sigma
^{\star} - \rho) - \mathrm{Tr}[\sigma^{\star}(\sigma^{\star} - \rho)] I}{\|\sigma^{\star} - \rho\|},
    \end{equation}
    where $\sigma^{\star} = \arg\min_{\sigma \in \mathcal{C}} \|\rho - \sigma\|$ is the nearest state in $\mathcal{C}$ to $\rho$ under the HS distance.
\end{theorem}
\begin{proof}
For simplicity, we denote the inner product of two operators $A$ and $B$ as $\langle A, B \rangle = \mathrm{Tr}[A^{\dagger} B]$ and ignore the subscript $\mathcal{C}$ in $D_{\mathcal{C}}(\rho)$, $W_{\mathcal{C}}$ and $B_{\mathcal{C}}(\rho)$ in the following proof.

From the maximization, we know that for any $\rho'\in \mathcal{C}$:
    \[
        B(\rho)\geq\min_{\sigma\in \mathcal{C}}\langle \frac{\rho'-\rho}{\|\rho'-\rho\|},\sigma-\rho\rangle,
    \]
    where we choose $W=\frac{\rho'-\rho}{\|\rho'-\rho\|}$. On the other hand, from the minimization, we know that for any $\rho'\in \mathcal{C}$:
    \[
    \begin{aligned}
         B(\rho)&\leq \max_{\|W-\mathrm{Tr}[W]I/d\|\leq 1}\langle W, \rho'-\rho\rangle\\
         &= \max_{\|W-\mathrm{Tr}[W]I/d\|\leq 1}\langle W-\mathrm{Tr}[W]I/d, \rho'-\rho\rangle\\
         &\leq \|\rho'-\rho\|.
    \end{aligned}
    \]
    Therefore, for any $\rho'\in \mathcal{C}$,
   \[
        \min_{\sigma\in \mathcal{C}}\langle \frac{\rho'-\rho}{\|\rho'-\rho\|},\sigma-\rho\rangle\leq B(\rho)\leq \|\rho'-\rho\|.
    \]
    Here, we denote the nearest state in $\mathcal{C}$ under the HS distance as $\sigma^\star$, then let $\rho'=\sigma^\star$ so that $B(\rho)\leq \|\sigma^\star-\rho\|=D(\rho)$ and
    \[
            B(\rho) \geq \min_{\sigma\in \mathcal{C}}\langle \frac{\sigma^\star-\rho}{\|\sigma^\star-\rho\|},\sigma-\rho\rangle
                    =\|\sigma^\star-\rho\|=D(\rho),
    \]
    where the first equality follows from the variational condition for the nearest point in the convex compact set $\mathcal{C}$. The above inequalities are saturated when $D(\rho)=B(\rho)$, with the corresponding geometrically optimal witness
    \[
        W_{\text{opt}}=\frac{(\sigma^\star-\rho)-\langle\sigma^\star,\sigma^\star-\rho\rangle I}{\|\sigma^\star-\rho\|},
    \]
    where the part $\langle\sigma^\star,\sigma^\star-\rho\rangle I$ places the supporting hyperplane through $\sigma^\star$, while geometrical optimality follows from the maximal-distinguishability relation above.
\end{proof}

\section{Dynamic refinement of extremes approximation model (DREAM)}
\label{app:dream}
\subsection{Overview of the optimization procedure}
Here, we consider the general HS measure of a target state $\rho$ with respect to a convex, compact set $\mathcal{C}$:
\[
D_{\mathcal{C}}(\rho) = \min_{\sigma \in \mathcal{C}} \|\rho - \sigma\| = \min_{\sigma \in \mathcal{C}} \sqrt{\mathrm{Tr}[(\rho - \sigma)^2]}.
\]
The key step in this optimization is the parametrization of the convex set $\mathcal{C}$. Due to the convexity of $\mathcal{C}$, it can be expressed as the convex hull of all its extreme points. These extreme points, corresponding to pure states $\{|\psi_i\rangle\}$ with specific properties, allow us to approximate $\mathcal{C}$ using a convex hull with limited extreme points. A traditional convex hull approximation involves sampling a large number of extreme points and constructing the convex hull \cite{PhysRevA.98.012315, hou2020upper}. However, this method can be highly inefficient and inaccurate, as most of the sampled extreme points may be irrelevant to the computation of $D_{\mathcal{C}}$. While one could attempt to improve the approximation by sampling additional extreme points and discarding those with low probability weights, such a randomness-based updating scheme provides limited improvement.

\begin{figure}[h]
    \centering
    \includegraphics[width=\linewidth]{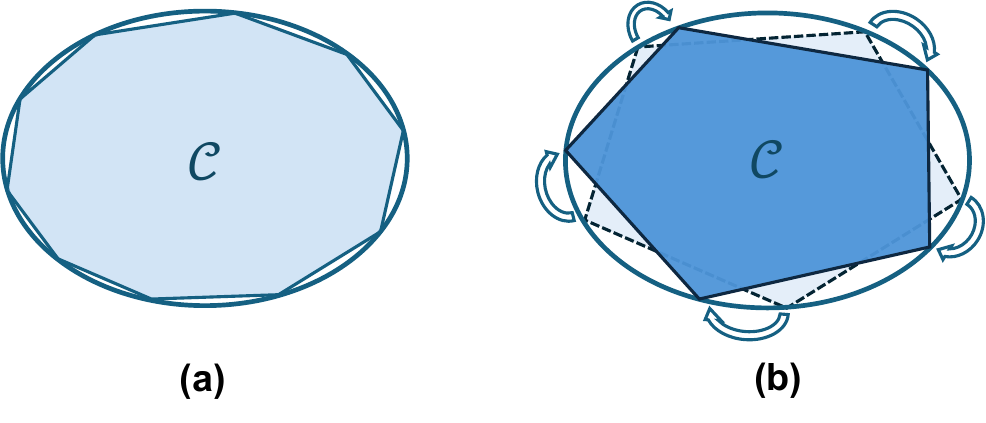}
    \caption{Schematic diagram of convex hull approximation (CHA): (a) Traditional CHA: Random sampling of a large number of extreme points. (b) Dynamic CHA: A small number of extreme points characterized by variational parameters.}
    \label{fig:cha}
\end{figure}

Alternatively, Carathéodory's theorem \cite{caratheodory1911variabilitatsbereich} guarantees that any state within a convex subset of the quantum state space can be expressed as a finite convex combination of extreme points, with the number of terms upper-bounded by $D^2$, where $D$ is the total dimension of the quantum system. This insight enables the use of parametrized extreme points to construct a dynamic convex hull approximation, which requires significantly fewer extreme points compared to random sampling, see Fig.~\ref{fig:cha}. Consequently, the computation of the HS measure can be reformulated as:
\begin{equation}
    D_{\mathcal{C}}(\rho) = \min_{\vec{\alpha}, \vec{\beta}} \left\|\rho - \left(\sum_{i=1}^n \lambda_i\left(\vec{\alpha}\right) |\psi_i(\vec{\beta})\rangle \langle\psi_i(\vec{\beta})| \right) \right\|,
\end{equation}
where $\vec{\alpha}$ and $\vec{\beta}$ are real parameters characterizing $\{\lambda_i\}$ and $\{|\psi_i\rangle\}$, respectively. Here, the number of terms $n \leq D^2$ is a hyperparameter in the optimization procedure. Numerical experiments indicate that using several multiples of $D$ typically provides satisfactory results for most quantum states, which aligns with the previous studies \cite{wu2004convex,ha2014separable,chen2013dimensions}.

To solve the resulting unconstrained optimization problem, we employ the PyTorch deep learning framework to construct a computational graph for the objective function. This allows automatic differentiation via backpropagation to compute gradients efficiently. For optimization, we utilize the classical L-BFGS algorithm. As the problem is inherently nonconvex, convergence to the global minimum is not guaranteed. The output is therefore denoted by $\hat{\sigma}$ and treated as a numerical approximation to the exact closest state $\sigma^\star$, rather than as a globally certified solution. To improve numerical precision and assess stability, we perform multiple runs with random initializations.
\subsection{Parametrization}
\label{app:dream-parameterization}
The parametrization of the probability distribution is introduced in Sec.~\ref{sec:construction}. Here, we focus on the parametrization of extreme points $\{|\psi_i\rangle\}$ in different entanglement detection tasks.

\emph{\textbf{High-dimensional entanglement.}}---
In this scenario, $\{|\psi_i\rangle\}$ represents the set of states with Schmidt ranks no larger than $k-1$, defined as follows:
\begin{equation}
|\psi\rangle = \frac{|\tilde{\psi}\rangle}{\sqrt{\langle\tilde{\psi}|\tilde{\psi}\rangle}}, \quad
|\tilde{\psi}\rangle = \sum_{j=1}^{k-1} s_j |\phi_A^{(j)}\rangle \otimes |\phi_B^{(j)}\rangle,
\end{equation}
where $\{s_j\}$ are positive coefficients. The parametrization of the local states $\{\phi_A^{(j)}\}$ and $\{\phi_B^{(j)}\}$ is detailed in Sec.~\ref{sec:construction}. For the coefficients $\{s_j\}$, we employ the SoftPlus function, a commonly used activation function in machine learning:
\begin{equation}
s_j = \ln(1 + e^{\theta_j}) \in \mathbb{R}^+,
\end{equation}
where $\{\theta_j\}$ are real-valued parameters.

\emph{\textbf{Genuine entanglement.}}---
For the detection of genuine entanglement, $\{|\psi_i\rangle\}$ is chosen as the set of product pure states under different bipartitions. For example, consider a tripartite quantum state in the composite Hilbert space $\mathcal{H}_A \otimes \mathcal{H}_B \otimes \mathcal{H}_C$. The possible bipartitions are $A|BC$, $B|CA$, and $C|AB$. Consequently, the set of extreme points for the set of biseparable states can be expressed as:
\begin{equation}
\begin{aligned}
    \{|\psi_i\rangle\}={}&\{|\phi_A^{(j)}\rangle \otimes |\phi_{BC}^{(j)}\rangle\} \\
    &\cup \{|\phi_B^{(k)}\rangle \otimes |\phi_{CA}^{(k)}\rangle\} \\
    &\cup \{|\phi_C^{(l)}\rangle \otimes |\phi_{AB}^{(l)}\rangle\}.
\end{aligned}
\end{equation}
For each product pure state under a given bipartition, the parametrization follows the same procedure as in the bipartite case.

\emph{\textbf{Complete entanglement.}}---
For the detection of completely entangled states, $\{|\psi_i\rangle\}$ is chosen as the set of fully product states under the given partitions. For a partition $A_1, A_2, \dots, A_n$, the fully product state takes the form:
\begin{equation}
    |\psi\rangle = |\phi_{A_1}\rangle \otimes |\phi_{A_2}\rangle \otimes \cdots \otimes |\phi_{A_n}\rangle,
\end{equation}
where $|\phi_{A_j}\rangle$ represents the local pure state in the $j$-th subsystem. The parametrization of these local states is detailed in Sec.~\ref{sec:construction}.
\subsection{Computational efficiency and accuracy}
\label{app:dream-benchmark}
Since the HS measure coincides with maximal distinguishability, the computational error of a GEW can be quantified as
\begin{equation}
    |\delta(\mathrm{Tr}[W\rho])|=\delta(D(\rho))=\|\rho-\sigma\|-\|\rho-\sigma^\star\|,
\end{equation}
where $\sigma$ is the nearest state obtained from the algorithm, while $\sigma^\star$ denotes the exact nearest state.

Here, we consider the $d\otimes d$ isotropic states defined in Eq.~\eqref{eq:isotropic_state}:
\begin{equation}
    \rho_{I}(F) = F |\Psi_d^+\rangle\langle \Psi_d^+| + (1-F) \frac{I - |\Psi_d^+\rangle\langle \Psi_d^+|}{d^2 - 1}.
\end{equation}
The corresponding HS measure is
\begin{equation}
    D(\rho_I(F))=
    \begin{cases}
        0,\quad 0\leq F\leq \frac{1}{d}, \\
        \frac{d}{\sqrt{d^2-1}}(F-\frac{1}{d}),\quad \frac{1}{d}<F\leq 1.
    \end{cases}
\end{equation}

We uniformly sampled one hundred entangled isotropic states from the range \( F \in \left(\frac{1}{d}, 1\right] \) with varying local dimensions \( d \), and applied our proposed method to compute their HS measures. For each data point, we calculated the average computational time and error, as presented in Fig.~\ref{fig:benchmark}. The results demonstrate that, even for quantum systems of moderate size, the proposed method achieves satisfactory accuracy within a reasonable computational time.

\begin{figure}[h]
    \centering
    \includegraphics[width=\linewidth]{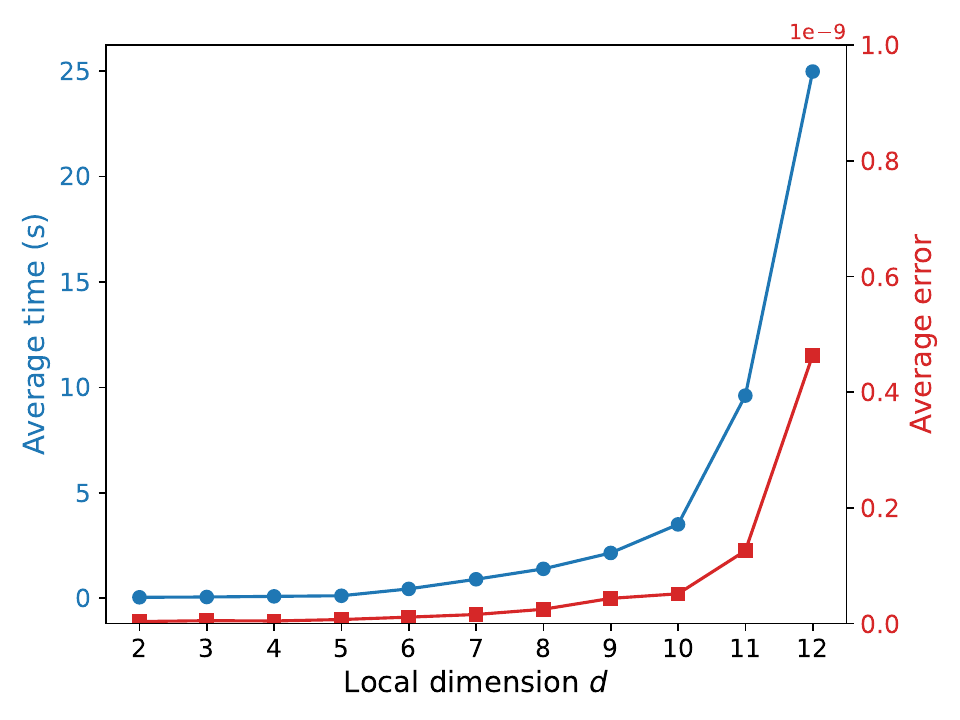}
    \caption{Average computational time and error for calculating the HS measure of isotropic states with different local dimensions. Each data point represents the average over one hundred selected entangled isotropic states. All computations were performed on a MacBook Pro with an M1 Pro chip.}
    \label{fig:benchmark}
\end{figure}

The analysis can also be extended to multipartite states. As a benchmark, we consider the HS measure of the $n$-qubit Greenberger--Horne--Zeilinger (GHZ) state,
\begin{equation}
    |\mathrm{GHZ}_n\rangle = \frac{1}{\sqrt{2}}(|0\rangle^{\otimes n} + |1\rangle^{\otimes n}).
\end{equation}
Here, the HS measure is computed with respect to the set of fully separable states. Although the GHZ and W states considered here are pure, their nearest fully separable states under the HS distance are generally \emph{mixed}: they lie on the boundary of the convex hull of fully product pure states and can involve nontrivial convex combinations of rank-deficient density matrices, as shown explicitly by $\sigma^\star$ below. The optimization is therefore not reduced to a pure-product-state overlap problem, but still involves the mixed-state separable boundary. Ref.~\cite{PhysRevA.102.012409} provides an analytical solution for the nearest fully separable state $\sigma^\star$ to $n$-qubit GHZ states:
\begin{equation}
    \sigma^\star = x\sigma_1+(1-x)\sigma_2,\quad x=\frac{\left(2^n-2\right)^2}{4+4^n-2^{n+1}},
\end{equation}
and
\begin{align}
    \sigma_{1}&=\frac{1}{2}\left(\begin{array}{ccccc}
1 & 0 & \ldots & 0 & 0 \\
0 & 0 & \ldots & 0 & 0 \\
\vdots & \vdots & \ddots & \vdots & \vdots \\
0 & 0 & \ldots & 0 & 0 \\
0 & 0 & \ldots & 0 & 1
\end{array}\right), \\
    \sigma_{2}&=\frac{1}{2^n}\left(\begin{array}{ccccc}
1 & 0 & \ldots & 0 & 1 \\
0 & 1 & \ldots & 0 & 0 \\
\vdots & \vdots & \ddots & \vdots & \vdots \\
0 & 0 & \ldots & 1 & 0 \\
1 & 0 & \ldots & 0 & 1
\end{array}\right).
\end{align}
This gives the HS measure of $n$-qubit GHZ states as
\begin{equation}
    D(|\mathrm{GHZ}_n\rangle) = \sqrt{\frac{\left(2^n-2\right)}{-4+2^{3-n}+2^{n+1}}}.
\end{equation}

Fig.~\ref{fig:ghz_hs_measure} shows the numerical results obtained from our method for $n$-qubit GHZ states. Compared with the analytical solution, the maximum relative error from $n=3$ to $n=11$ is about $0.0124\%$. Although the computational time increases with the number of qubits, it remains manageable up to $n=11$ qubits on a standard laptop. We also compute the HS measure of $n$-qubit W states,
\begin{equation}
    |W_n\rangle = \frac{1}{\sqrt{n}}(|00\ldots 001\rangle + |00\ldots 010\rangle + \cdots + |10\ldots 000\rangle).
\end{equation}
Although no closed-form expression is known for the nearest fully separable state to $n$-qubit W states, our method still converges efficiently in this regime, as shown in the right panel of Fig.~\ref{fig:ghz_hs_measure}.

\begin{figure}[!tbp]
    \centering
    \includegraphics[width=\linewidth]{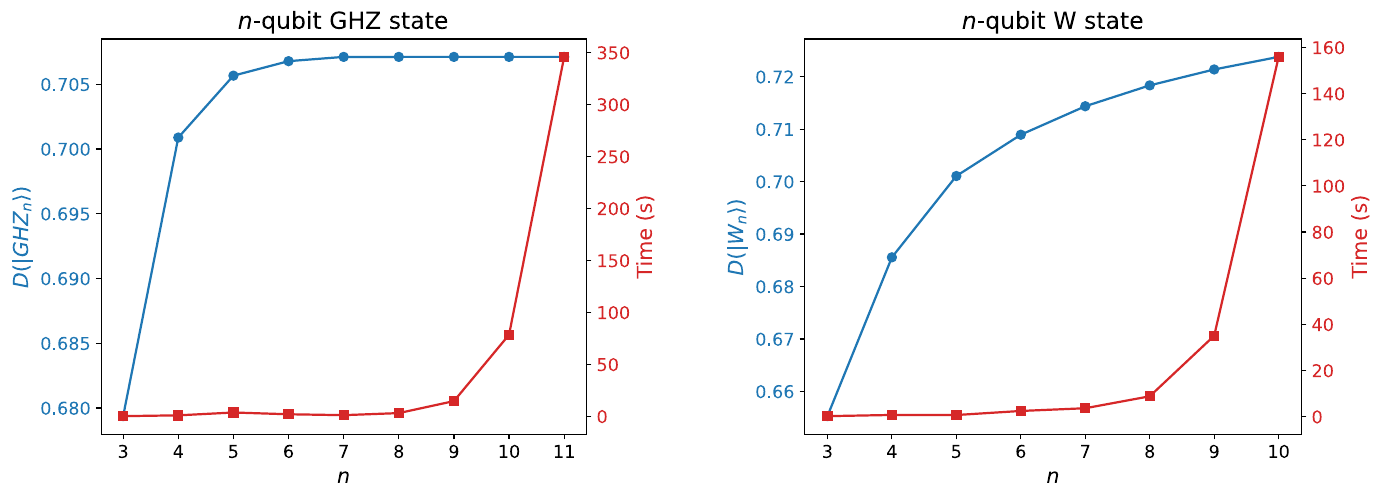}
    \caption{HS measure of $n$-qubit GHZ states (left) and $n$-qubit W states (right) with respect to fully separable states. Blue curves show the numerical HS measures, while red curves show the computational times. All computations were performed on a MacBook Pro with an M1 Pro chip.}
    \label{fig:ghz_hs_measure}
\end{figure}

\section{Enhanced noise robustness from a Newton-like method}
\label{app:newton}
In Sec.~\ref{sec:multipartite}, the visibility threshold for a target state is obtained directly using its corresponding GEW. In most cases, the performance already surpasses that of previous methods. However, if a specific noise model is incorporated into the construction of GEWs, an even better threshold can be achieved.

For instance, a Newton-like method can be employed to iteratively update our GEWs for improved noise resilience, as illustrated in Fig.~\ref{fig:newton_like_method}. The key idea is that, starting from the current GEW $W^{(n)}$, we identify the next transition state $\rho^{(n+1)}$ by solving $\mathrm{Tr}[W^{(n)}\rho^{(n+1)}]=0$, which gives
\begin{equation}
    \rho^{(n+1)}=\frac{\mathrm{Tr}[W^{(n)}]\rho-\mathrm{Tr}[W^{(n)}\rho]\cdot I}{\mathrm{Tr}[W^{(n)}]-d\cdot\mathrm{Tr}[W^{(n)}\rho]},
\end{equation}
where $\rho=\rho^{(0)}$ is the initial target state with total dimension $d$. Using our method, we then find the nearest separable state $\sigma^{(n+1)}$ to $\rho^{(n+1)}$ and construct the updated GEW $W^{(n+1)}$ accordingly. Owing to the convexity of the set $\mathcal{C}$, this iterative procedure is guaranteed to improve the visibility threshold with each iteration.
\begin{figure}[!tbp]
    \centering
    \includegraphics[width=\linewidth]{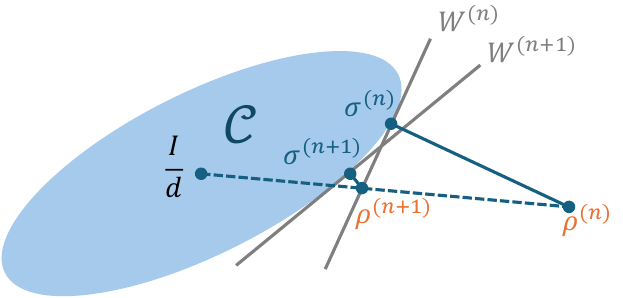}
    \caption{Schematic diagram of Newton-like method for constructing GEWs with better noise resilience against depolarizing noise.}
    \label{fig:newton_like_method}
\end{figure}

For example, the visibility thresholds $p_{\text{GEW}}$ reported in Sec.~\ref{sec:multipartite} for the three-qubit W state and the four-qubit W state are $0.5414$ and $0.6667$, respectively. After applying iterative optimization, the corresponding approximate numerical estimates become $0.4777$ and $0.4701$. As shown in Fig.~\ref{fig:w_state_threshold}, the estimates approach the known benchmark values after only a few iterations. Because the underlying optimization is nonconvex, small differences at the quoted precision should be regarded as numerical uncertainty and not as certified improvements beyond known biseparability boundaries.

\begin{figure}[!tbp]
    \centering
    \includegraphics[width=\linewidth]{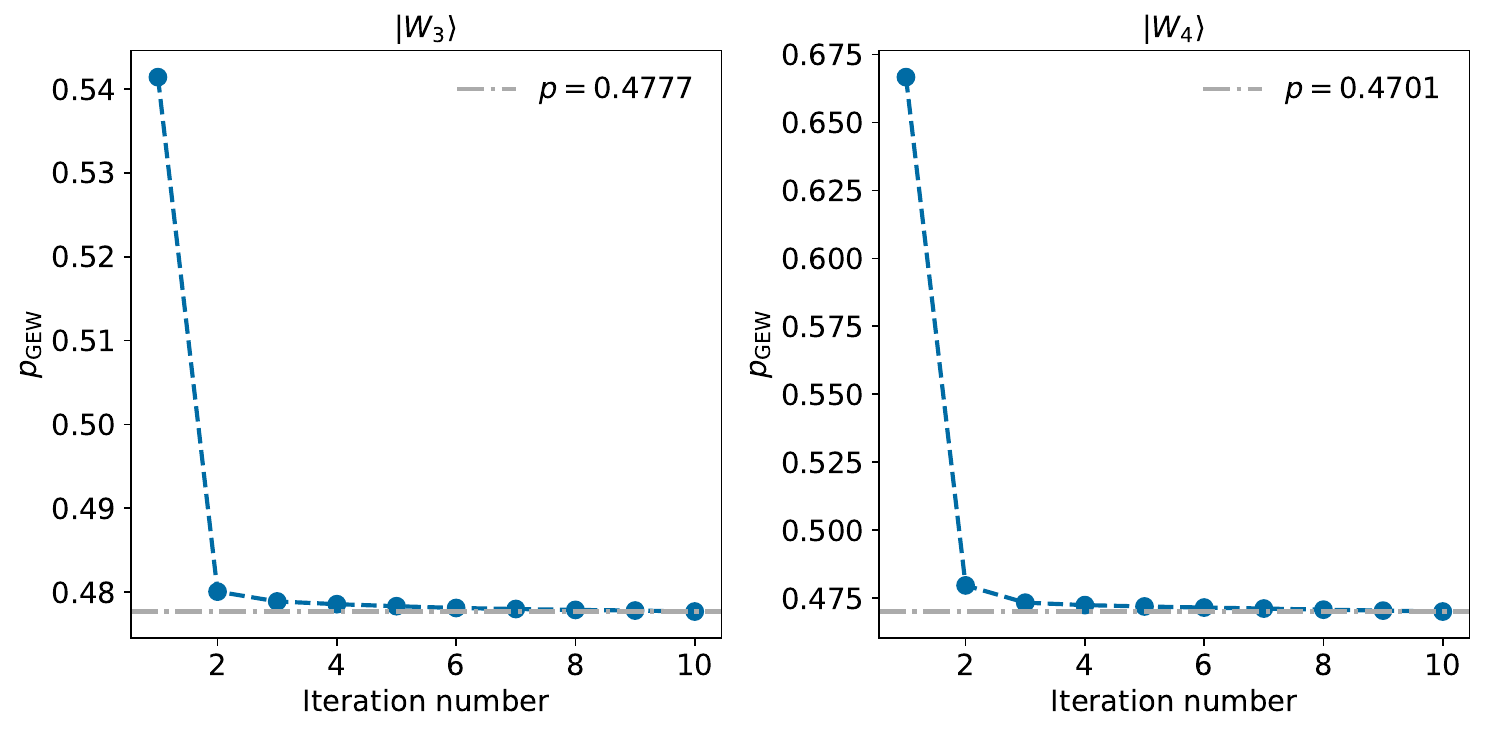}
    \caption{Numerical evolution of the approximate visibility estimates during iterative optimization.}
    \label{fig:w_state_threshold}
\end{figure}

We further compare the approximate refined visibility estimates with existing methods for the following standard genuinely multipartite entangled (GME) states:
\begin{equation}
    \begin{aligned}
        |W_3\rangle&=\frac{1}{\sqrt{3}}(|001\rangle+|010\rangle+|100\rangle),\\
        |W_4\rangle&=\frac{1}{2}(|0001\rangle+|0010\rangle+|0100\rangle+|1000\rangle),\\
        |D_4\rangle&= \frac{1}{\sqrt{6}}(|0011\rangle+|1100\rangle+|0110\rangle\\
        &+|1001\rangle+|1010\rangle+|0101\rangle),\\
        |\text{Cl}_4\rangle& = \frac{1}{2}(|0000\rangle+|0011\rangle+|1100\rangle-|1111\rangle), \\
        |\Psi_4\rangle&= \frac{1}{\sqrt{3}}(|0011\rangle+|1100\rangle\\
        &-\frac{1}{2}(|01\rangle+|10\rangle)\otimes(|01\rangle+|10\rangle))
    \end{aligned}
\end{equation}

The numerical estimates of $p_{\text{GEW}}$ are close to the fully decomposable benchmarks for the states considered above. Given the precision of the nonconvex numerical optimization, the small differences in Table~\ref{tab:multipartite_visibility} should not be interpreted as a certified performance advantage. More importantly, the two constructions differ in scope. Fully decomposable witnesses are restricted to decomposable operators and therefore cannot detect PPT-entangled bound entangled states. In contrast, a GEW based on the HS distance is constructed directly from the geometry of the target state relative to the relevant separable set and can address entanglement beyond the decomposable-witness class, as illustrated by the Horodecki and Tiles bound entangled states in Sec.~\ref{sec:bound} and Appendix~\ref{app:tiles}. The purpose of Table~\ref{tab:multipartite_visibility} is therefore to show that the numerical estimates are comparable on standard GME states, rather than to claim a certified improvement at the third decimal place.

\begin{table}[h]
\centering
\renewcommand{\arraystretch}{1.5}
\caption{Comparison of approximate numerical visibility estimates for GEWs with other methods. FID denotes fidelity-based witnesses, OSD denotes operator-Schmidt-decomposition witnesses, and FUL denotes fully decomposable witnesses. The FID and OSD values are taken from Ref.~\cite{PhysRevLett.133.040203}, while the FUL values are taken from Ref.~\cite{PhysRevLett.106.190502}. Differences at the third decimal place are within the numerical precision relevant here and are not interpreted as certified improvements beyond known biseparability boundaries.}
\label{tab:multipartite_visibility}
\resizebox{\columnwidth}{!}{%
\begin{tabular}{|c|c|c|c|c|}
\hline
State & Visibility $p_{\text{FID}}$ & Visibility $p_{\text{OSD}}$ & Visibility $p_{\text{FUL}}$ & Visibility $p_{\text{GEW}}$ \\
\hline
$\left|W_3\right\rangle$ & 0.619 & 0.556 & 0.479 & 0.478 \\
\hline
$\left|W_4\right\rangle$ & 0.733 & 0.714 & 0.474 & 0.470 \\
\hline
$\left|D_4\right\rangle$ & 0.644 & 0.540 & 0.461 & 0.459 \\
\hline
$\left|\text{Cl}_4\right\rangle$ & 0.467 & 0.463 & 0.385 & 0.383 \\
\hline
$\left|\Psi_4\right\rangle$ & 0.733 & 0.572 & 0.447 & 0.446 \\
\hline
\end{tabular}
}
\end{table}

\section{Additional results}
\label{app:additional-results}
\subsection{Tiles bound entangled state}
\label{app:tiles}
A bound entangled state (BES) can be constructed from the Tiles unextendible product basis (UPB) $\{|\psi_i\rangle\}$ as follows \cite{PhysRevLett.82.5385},
\begin{equation}
    \rho_{\mathrm{Tiles}}=\frac{1}{4}(I-\sum_{i=0}^4|\psi_i\rangle\langle\psi_i|),
\end{equation}
where the UPB states are
\begin{equation}
\begin{aligned}
&\left|\psi_0\right\rangle=\frac{1}{\sqrt{2}}|0\rangle(|0\rangle-|1\rangle),\\
& \left|\psi_1\right\rangle=\frac{1}{\sqrt{2}}(|0\rangle-|1\rangle)|2\rangle, \\
&\left|\psi_2\right\rangle=\frac{1}{\sqrt{2}}|2\rangle(|1\rangle-|2\rangle), \\
& \left|\psi_3\right\rangle=\frac{1}{\sqrt{2}}(|1\rangle-|2\rangle)|0\rangle, \\
&\left|\psi_4\right\rangle=\frac{1}{3}(|0\rangle+|1\rangle+|2\rangle)(|0\rangle+|1\rangle+|2\rangle).
\end{aligned}
\end{equation}

By computing the HS measure of Tiles BES, we find that its nearest separable state is a mixture of ten product pure states as the following form:
\begin{equation}
    \begin{aligned}
        \sigma=&p_{0}\left(\rho_{a}\otimes\rho_{b}+\rho_{b}\otimes\rho_{\tilde{a}}+\rho_{\tilde{b}}\otimes\rho_{a}+\rho_{\tilde{a}}\otimes\rho_{\tilde{b}}\right)\\
    +&p_{1}\left(\rho_{c}\otimes\rho_{d}+\rho_{d}\otimes\rho_{\tilde{c}}+\rho_{\tilde{d}}\otimes\rho_{c}+\rho_{\tilde{c}}\otimes\rho_{\tilde{d}}\right)\\
    +&p_{2}\left(\rho_{e}\otimes\rho_{f}+\rho_{f}\otimes\rho_{e}\right),
    \end{aligned}
\end{equation}
where $\rho_x:=|x\rangle\langle x|$ with $|x\rangle:=x_0|0\rangle+x_1|1\rangle+x_2|2\rangle$ and $\rho_{\tilde{x}}:=|\tilde{x}\rangle\langle\tilde{x}|$ with $|\tilde{x}\rangle:=x_2|0\rangle+x_1|1\rangle+x_0|2\rangle$. The numerical values of these vectors are provided below:
\[
\begin{aligned}
    p &= (0.0757194884, 0.1304440066, 0.08767301), \\
    a &= (0.8229256532, 0.5311820912, 0.2015910597), \\
    b &= (0.4803621975, 0.3840458715, -0.7885181849), \\
    c &= (0.7069251656, 0.7070735048, -0.0174318434), \\
    d &= (0.1899793832, 0.4370482044, -0.8791454379), \\
    e &= (0.0217550144, 0.9995266073, 0.0217550144), \\
    f &= (0.2725406432, -0.9227367965, 0.2725406432).
\end{aligned}
\]

\subsection{Four-qubit Smolin state}
Here, we consider the four-qubit Smolin state \cite{PhysRevA.63.032306}:
\begin{equation}
    \rho_S = \frac{1}{16} \left( I + \sigma_x^{\otimes 4} + \sigma_y^{\otimes 4} + \sigma_z^{\otimes 4} \right),
\end{equation}
where $\sigma_x$, $\sigma_y$, and $\sigma_z$ are the Pauli matrices. The HS measure of a state under a specific partition $\mathcal{P}$ is defined as
\begin{equation}
    D(\rho) = \min_{\sigma \in \mathrm{SEP}(\mathcal{P})} \|\rho - \sigma\|,
\end{equation}
where $\mathrm{SEP}(\mathcal{P})$ denotes the set of separable states under the partition $\mathcal{P}$. Due to the symmetry of the Smolin state, it suffices to consider the bipartitions $\{1\}:\{2,3,4\}$ and $\{1,2\}:\{3,4\}$, the tripartition $\{1\}:\{2\}:\{3,4\}$, and the four-partition $\{1\}:\{2\}:\{3\}:\{4\}$.

The Smolin state has been proven to be entangled for all partitions except $\{1,2\}:\{3,4\}$ \cite{PhysRevA.63.032306}. Furthermore, when white noise is added, the noisy state
\begin{equation}
    \mathcal{N}_p(\rho_S) = p \rho_S + (1-p) \frac{I}{16}
\end{equation}
remains entangled for $p > \frac{1}{3}$ \cite{PhysRevA.74.010305}. A natural conjecture is that
\begin{equation}
    D(\mathcal{N}_p(\rho_S)) =
    \begin{cases}
        0, & p \in [0, \frac{1}{3}], \\
        \|\mathcal{N}_p(\rho_S) - \mathcal{N}_{\frac{1}{3}}(\rho_S)\|, & p \in (\frac{1}{3}, 1],
    \end{cases}
\end{equation}
which simplifies to
\begin{equation}
        D(\mathcal{N}_p(\rho_S)) =
    \begin{cases}
        0, & p \in [0, \frac{1}{3}], \\
        \frac{\sqrt{3}}{4} \left( p - \frac{1}{3} \right), & p \in (\frac{1}{3}, 1].
    \end{cases}
\end{equation}

Our numerical results strongly support this conjecture, as illustrated in Fig.~\ref{fig:smolin_state}. The corresponding GEWs for detecting entanglement in the range $p \in (\frac{1}{3},1]$ across any partition can be constructed directly from the state $\mathcal{N}_{\frac{1}{3}}(\rho_S)$:

\begin{equation}
    W = \frac{(\mathcal{N}_{\frac{1}{3}}(\rho_S)-\rho_S)-\mathrm{Tr}[\mathcal{N}_{\frac{1}{3}}(\rho_S)(\mathcal{N}_{\frac{1}{3}}(\rho_S)-\rho_S)]I}{\|\mathcal{N}_{\frac{1}{3}}(\rho_S)-\rho_S\|}.
\end{equation}

\begin{figure}[h]
    \centering
    \includegraphics[width=0.8\linewidth]{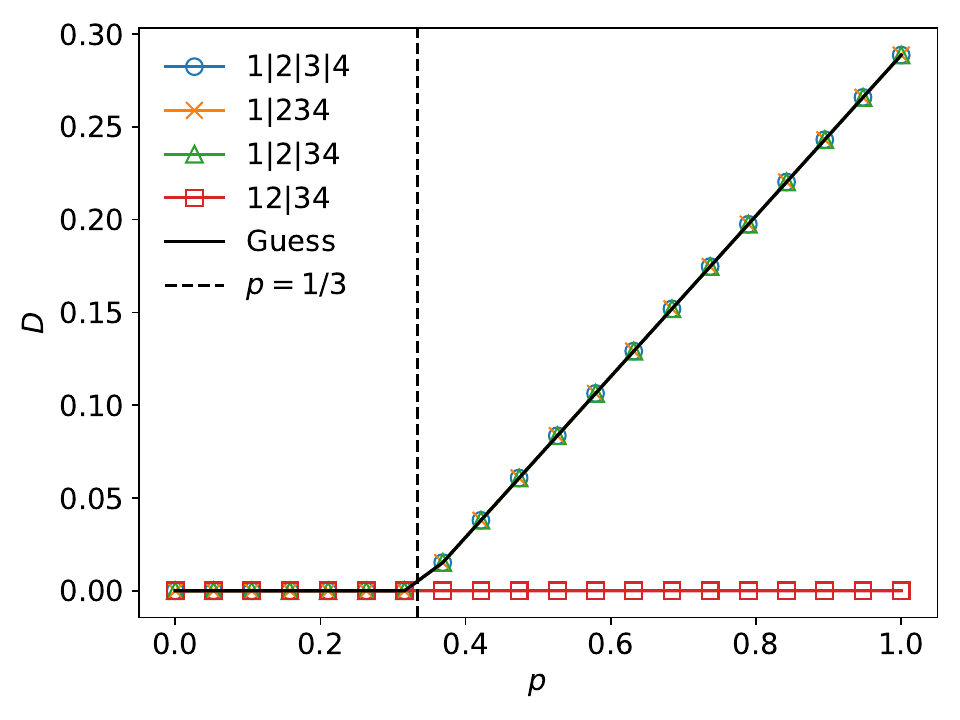}
    \caption{HS measure of noisy Smolin states $\mathcal{N}_p(\rho_S)=p\rho_S+(1-p)\frac{I}{16}$ under different partitions.}
    \label{fig:smolin_state}
\end{figure}

\section{Experimental details}
\label{app:experiment}
\subsection{Experimental system}

Our experiments are performed on a nuclear magnetic resonance (NMR) quantum processor
using a Bruker 300-MHz NMR spectrometer.
The sample is $^{13}$C-labeled crotonic acid dissolved in d$_6$-acetone.
The molecular structure is shown in Fig.~\ref{molecular}(a).
This molecule hosts four coupled $^{13}$C nuclear spins---C$_1$, C$_2$, C$_3$, and C$_4$---which together constitute a four-qubit quantum processor. All hydrogen spins, including $\mathrm{H}_1$, $\mathrm{H}_2$, and the methyl group $\mathrm{M_H}$, are fully decoupled throughout the experiments.

\begin{figure}[H]
    \centering
    \includegraphics[width=\linewidth]{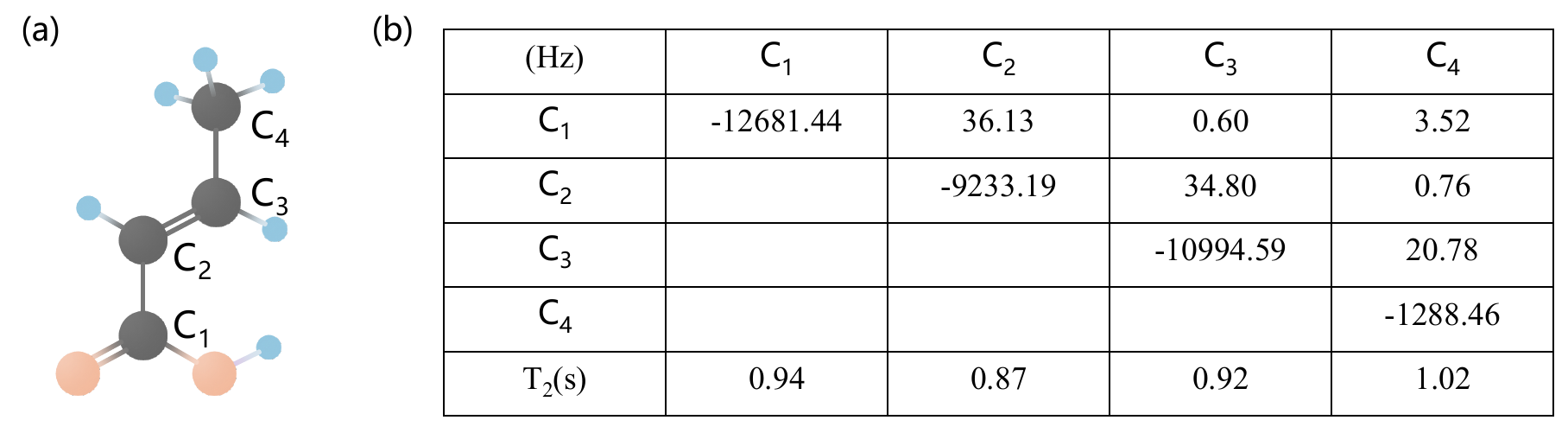}
    \caption{Molecular structure and system parameters.}
    \label{molecular}
\end{figure}

The internal dynamics of the system are governed by the Hamiltonian~\cite{PhysRevLett.129.100603}
\begin{equation}\label{Hamiltonian}
    \mathcal{H}_{\text{NMR}}=-\sum^4_{i=1}\frac{\omega_i}{2}\sigma^i_z+\frac{\pi}{2}\sum^4_{i<j}J_{ij}\sigma^i_z\sigma^j_z,
\end{equation}
where $\sigma_z$ denotes the Pauli matrix, and $\omega_i = 2\pi \nu_i$ is the Larmor frequency associated with the chemical shift $\nu_i$ of the $i$th spin. The scalar coupling strengths $J_{ij}$ quantify the interactions between spins $i$ and $j$. These parameters are extracted from the thermal-equilibrium NMR spectrum and summarized in Fig.~\ref{molecular}(b), with diagonal elements showing Larmor frequencies and off-diagonal elements representing coupling constants. Controlled two-spin interactions are implemented by letting the system evolve naturally under this Hamiltonian.

Control of each individual spin is achieved using radio-frequency (rf) pulses described by
\begin{equation}
    H_{rf}=-\sum_k\gamma_kB_1[\cos{(\omega_{rf}t+\phi)} \sigma^k_x+\sin{(\omega_{rf}t+\phi)} \sigma^k_y],
\end{equation}
where $B_1$, $\phi$, and $\omega_{rf}$ are the amplitude, phase, and frequency of the rf field, respectively, and $\gamma_k$ is the gyromagnetic ratio of spin $k$. By designing appropriate sequences of rf pulses, arbitrary single-spin rotations can be implemented. Combined with the intrinsic spin-spin couplings, these controls provide universal manipulation over the entire four-spin system.

\subsection{State preparation}

\begin{widetext}
    \centering
    \includegraphics[width=0.94\textwidth]{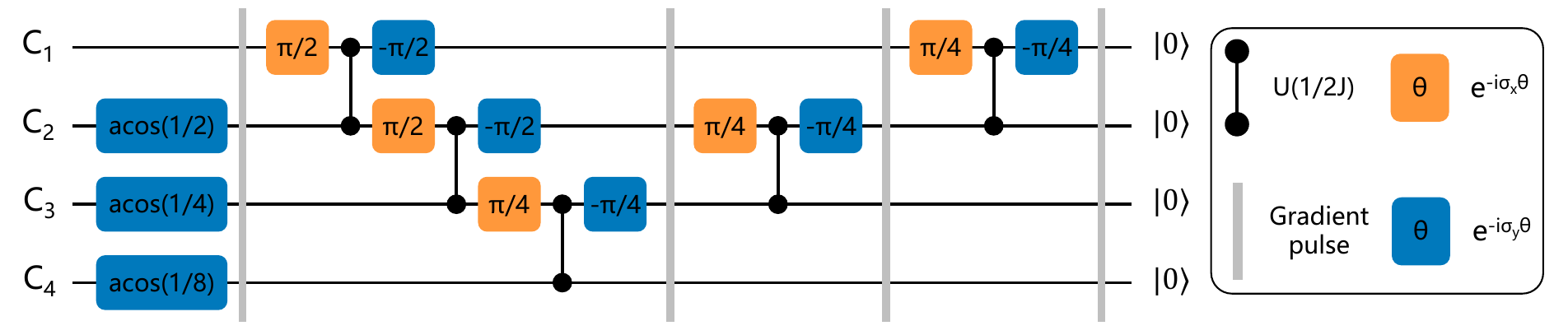}
    \captionof{figure}{Circuit for PPS preparation.}
    \label{circuit}
\end{widetext}

Because the system's natural Gibbs state does not match the initial states required for our experiments, a dedicated preparation procedure is employed.
To test the witnesses introduced in the main article, we prepare states of specific forms. This section describes how these states are generated.
The overall preparation begins with the creation of a pseudo-pure state (PPS) via spatial averaging.
The target states are then generated either by temporal averaging for predefined mixed states or by coherent unitary preparation for W states.

We begin by generating the PPS from the thermal equilibrium state.
In NMR quantum computing, a PPS is a mixed state that behaves effectively as a pure state.
For our four-spin system, the PPS is defined as
\begin{equation}
    \rho_\mathrm{PPS}=\frac{1-\epsilon}{16}I+\epsilon\left|0000\right\rangle\left\langle0000\right|,
\end{equation}
where $I$ is the $16\times 16$ identity matrix and the polarization under experimental conditions is $\epsilon \sim 10^{-5}$.
The identity term does not influence observable signals or typical NMR dynamics (rf pulses or free evolution).
Thus, up to a global scaling factor, the PPS behaves as the pure state $\left|0000\right\rangle\left\langle0000\right|$.
We employ spatial averaging to obtain this PPS from the thermal state~\cite{PhysRevLett.132.210403}.

Spatial averaging relies on applying magnetic-field gradients to suppress unwanted off-diagonal terms in the density matrix.
All experiments are carried out at room temperature.
The thermal state of the four-spin system is
\begin{equation}
    \rho_{eq}=\frac{I}{2^4}+\epsilon\sum_{k=1}^4\sigma_z^k .
\end{equation}
Using the pulse sequence shown in Fig.~\ref{circuit}, and ignoring the identity component, the deviation density operator evolves as
\begin{widetext}
\begin{equation}
    \begin{aligned}
    \sigma_z^1+\sigma_z^2+\sigma_z^3+\sigma_z^4
    &\xrightarrow{U_{p1},\mathrm{G}z}
    \sigma_z^1+\frac{1}{2}\sigma_z^2+\frac{1}{4}\sigma_z^3+\frac{1}{8}\sigma_z^4\\
    &\xrightarrow{U_{p2},\mathrm{G}z}
    \frac{1}{2}(\sigma_z^1\sigma_z^2\sigma_z^3+\sigma_z^1\sigma_z^2\sigma_z^3\sigma_z^4)
    +\frac{1}{4}(\sigma_z^2\sigma_z^3+\sigma_z^2\sigma_z^3\sigma_z^4)
    +\frac{1}{8}(\sigma_z^3\sigma_z^4+\sigma_z^4)\\
    &\xrightarrow{U_{p3},\mathrm{G}z}
    \frac{1}{4}(\sigma_z^1\sigma_z^2\sigma_z^3+\sigma_z^1\sigma_z^2+\sigma_z^1\sigma_z^2\sigma_z^3\sigma_z^4+\sigma_z^1\sigma_z^2\sigma_z^4)\\
    &\qquad\qquad
    +\frac{1}{8}(\sigma_z^2+\sigma_z^2\sigma_z^3+\sigma_z^2\sigma_z^3\sigma_z^4+\sigma_z^2\sigma_z^4+\sigma_z^3+\sigma_z^3\sigma_z^4+\sigma_z^4)\\
    &\xrightarrow{U_{p4},\mathrm{G}z}
    \frac{1}{8}(\sigma_z^1\sigma_z^2\sigma_z^3+\sigma_z^1\sigma_z^3+\sigma_z^1\sigma_z^2+\sigma_z^1
    +\sigma_z^1\sigma_z^2\sigma_z^3\sigma_z^4+\sigma_z^1\sigma_z^3\sigma_z^4\\
    &\qquad\qquad
    +\sigma_z^1\sigma_z^2\sigma_z^4+\sigma_z^1\sigma_z^4+\sigma_z^2+\sigma_z^2\sigma_z^3
    +\sigma_z^2\sigma_z^3\sigma_z^4+\sigma_z^2\sigma_z^4\\
    &\qquad\qquad
    +\sigma_z^3+\sigma_z^3\sigma_z^4+\sigma_z^4).
    \end{aligned}
\end{equation}
\end{widetext}
The resulting deviation operator is proportional to the traceless part of
\[
\bigotimes_{k=1}^4 \frac{I+\sigma_z^k}{2}
=
\left|0000\right\rangle\left\langle0000\right|,
\]
which corresponds to the desired PPS up to an irrelevant identity component and a global scaling factor.

For the high-dimensional-entanglement experiment, the bipartite $4\otimes4$ Hilbert space is encoded into the same four-qubit NMR register by grouping two physical qubits into one subsystem and the other two physical qubits into the second subsystem.
Specifically, we take $(C_1,C_2)$ as subsystem $A$ and $(C_3,C_4)$ as subsystem $B$.
The local computational basis of each four-dimensional subsystem is defined as
\[
\left|0\right\rangle \equiv \left|00\right\rangle,\quad
\left|1\right\rangle \equiv \left|01\right\rangle,\quad
\left|2\right\rangle \equiv \left|10\right\rangle,\quad
\left|3\right\rangle \equiv \left|11\right\rangle .
\]
Under this encoding, the four-qubit Hilbert space is identified with a bipartite $4\otimes4$ system used to realize the high-dimensional entangled states in the Schmidt-number-witness experiment.

For the Horodecki bound-entanglement experiment shown in Fig.~\ref{fig:bound_entanglement}, the target states are bipartite $3\otimes3$ states.
They are encoded into the same four-qubit NMR register by grouping two physical qubits as subsystem $A$ and the other two physical qubits as subsystem $B$.
For each subsystem, we use a three-dimensional subspace of the two-qubit Hilbert space as the qutrit computational basis,
\[
|0\rangle \equiv |00\rangle,\quad
|1\rangle \equiv |01\rangle,\quad
|2\rangle \equiv |10\rangle,
\]
while the remaining level $|11\rangle$ is not used in this encoding.
Under this mapping, the Horodecki states $\rho_b$ are embedded into a nine-dimensional subspace of the four-qubit Hilbert space.

We then prepare the predefined mixed states required for our experiment.
This includes the Horodecki bound-entangled states $\rho_b$ and the high-dimensional target states used in the Schmidt-number-witness experiment.
Since these target states cannot be obtained from the pure state $\left|0000\right\rangle$ through unitary evolution alone, we employ the temporal averaging technique commonly used in NMR. In this approach, we repeatedly generate the individual pure-state components that constitute the desired mixed state.
Each component can be prepared from the PPS using appropriate unitary operations, which are implemented experimentally through tailored pulse sequences~\cite{Fan2023Innovation,PhysRevLett.126.110502}.
These pulses are designed using an optimal-control algorithm, yielding simulated fidelities above 99.5\% with a total pulse duration of 20\,ms.
After preparing all components, we combine them with the proper weights to obtain the final target mixed states.

\begin{figure}[H]
    \centering
    \includegraphics[width=\linewidth]{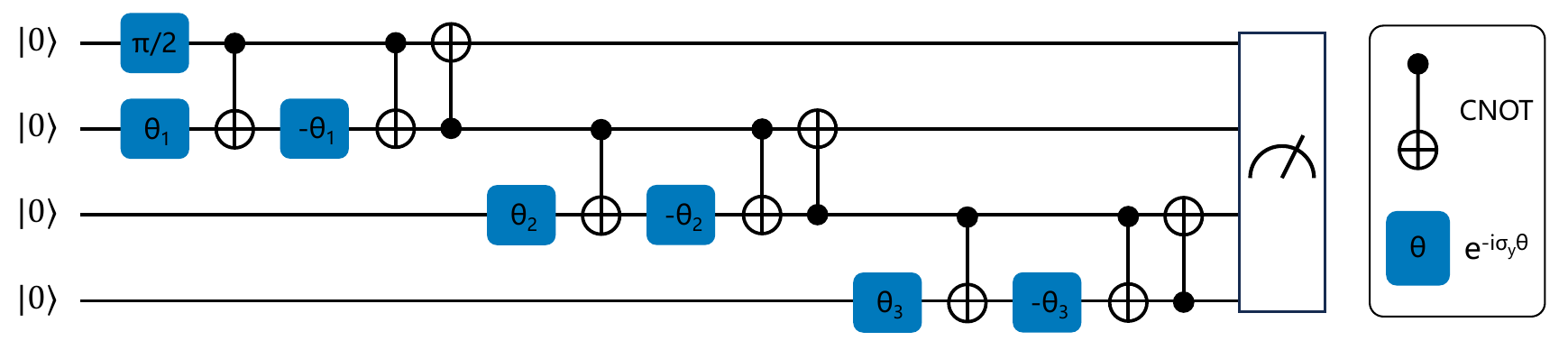}
    \caption{Logical circuit for preparing the four-qubit W state.}
    \label{fig:wstate_preparation}
\end{figure}

For the multipartite-entanglement experiments, the three- and four-qubit W states are pure states prepared coherently from the PPS by unitary operations.
The logical preparation circuit for the four-qubit W state is shown in Fig.~\ref{fig:wstate_preparation}.

The four-qubit target state is
\[
|W_4\rangle=\frac{1}{2}
\left(
|1000\rangle+|0100\rangle+|0010\rangle+|0001\rangle
\right).
\]
Starting from $|0000\rangle$, a $y$-axis rotation with angle $\pi/2$ first creates a single excitation.
The subsequent controlled rotations, implemented by the rotation pairs $\theta_j$ and $-\theta_j$ together with CNOT gates, sequentially redistribute the excitation amplitude over the four spins, following the standard controlled-rotation state-preparation strategy~\cite{Mottonen2005Transformation}.
With the convention $R_y(\theta)=e^{-i\sigma_y\theta}$, the rotation angles for preparing $|W_4\rangle$ are chosen as
\[
\begin{aligned}
\cos(2\theta^{(4)}_1)&=\frac{1}{2}, &
\cos(2\theta^{(4)}_2)&=\frac{1}{\sqrt{3}},\\
\cos(2\theta^{(4)}_3)&=\frac{1}{\sqrt{2}}.
\end{aligned}
\]
which gives equal amplitudes for the four single-excitation basis states.

The three-qubit W state is prepared in the subspace of the first three spins, while the fourth spin is kept in $|0\rangle$.
The target state is therefore
\[
|W_3\rangle\otimes|0\rangle
=
\frac{1}{\sqrt{3}}
\left(
|1000\rangle+|0100\rangle+|0010\rangle
\right).
\]
Its preparation follows the same sequential-amplitude-distribution procedure as in Fig.~\ref{fig:wstate_preparation}, but only the first two controlled-rotation blocks are used.
The final redistribution between the third and fourth spins is omitted, so that the fourth spin remains in $|0\rangle$.
The corresponding angles are chosen as
\[
\cos(2\theta^{(3)}_1)=\frac{1}{\sqrt{3}},\qquad
\cos(2\theta^{(3)}_2)=\frac{1}{\sqrt{2}}.
\]
These logical preparation circuits are compiled into shaped radio-frequency pulses using the same GRAPE-based optimal-control procedure described above.

\subsection{Readout}

The signal produced by an individual molecule is too weak to detect directly, so the NMR sample contains a macroscopic ensemble of identical molecules.
All measurements therefore produce ensemble-averaged signals, from which the quantum state of an individual molecule can be inferred.

During detection, the nuclear spins precess around the static magnetic field $B_0$ along the $z$-axis and gradually relax back to thermal equilibrium.
This precession induces an oscillating transverse (in the $x$--$y$ plane) electromagnetic signal, which is captured by the surrounding rf detection coils.
The resulting time-domain signal, known as the free induction decay (FID), can be expressed as~\cite{PhysRevLett.134.040201}
\begin{equation}
    A(t)=\mathrm{Tr}\!\left[e^{-i\mathcal{H}_{\text{NMR}} t}\,\rho\, e^{i\mathcal{H}_{\text{NMR}} t}
    \sum_k (\sigma_x^k - i\sigma_y^k)\, e^{-\lambda_k t}\right],
\end{equation}
where $\rho$ is the density matrix of the system and $\lambda_k = 1/T_2$ accounts for transverse relaxation.
The real and imaginary parts correspond to the measured signals along the $x$ and $y$ axes, respectively.
A typical real-part FID for the thermal state of crotonic acid is shown in Fig.~\ref{spec}(a).

\begin{figure}[H]
    \centering
    \includegraphics[width=\linewidth]{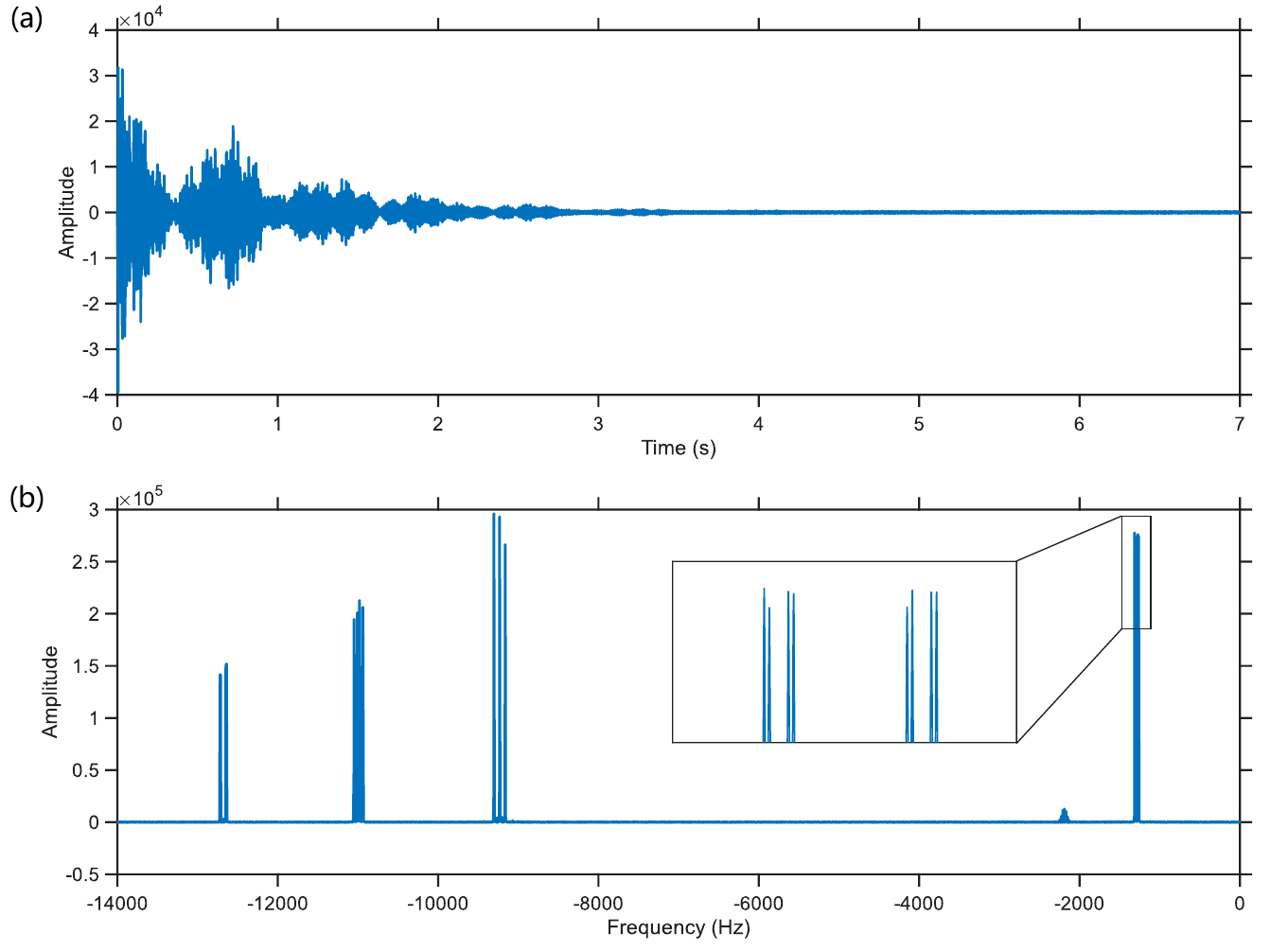}
    \caption{Direct results obtained in the NMR experiment.}
    \label{spec}
\end{figure}

Given the Hamiltonian in Eq.~\ref{Hamiltonian}, the FID can be decomposed into a sum of single-frequency components determined by the Larmor frequencies and coupling constants.
These components can be resolved through a Fourier transform.

Consider a single-spin example with Hamiltonian $H_{\mathrm{single}} = \frac{\omega_s}{2}\sigma_z$, where $\omega_s$ is the Larmor frequency.
The density matrix can be expanded in the Pauli basis as
\begin{equation}
    \rho_\mathrm{single}=\frac{I_2 + \alpha_1\sigma_x + \alpha_2\sigma_y + \alpha_3\sigma_z}{2}.
\end{equation}
The corresponding FID signal is
\begin{widetext}
\begin{equation}
    A_\mathrm{single}(t)
    = \frac{1}{2}\mathrm{Tr}\!\left[
    e^{-i\omega_s \sigma_z t/2}
    (I_2+\alpha_1\sigma_x+\alpha_2\sigma_y+\alpha_3\sigma_z)
    e^{i\omega_s \sigma_z t/2}
    (\sigma_x-i\sigma_y)e^{-\lambda_k t}
    \right].
\end{equation}
\end{widetext}
Because $\mathrm{Tr}(\sigma_x) = \mathrm{Tr}(\sigma_y) = 0$, the terms $I_2$ and $\alpha_3\sigma_z$ do not contribute.
Separating into real and imaginary parts yields
\begin{widetext}
\begin{equation}
\begin{aligned}
    \mathrm{Re}[A_\mathrm{single}(t)] &= \frac{1}{2}\mathrm{Tr}\!\left[
    e^{-i\omega_s\sigma_z t/2}(\alpha_1\sigma_x+\alpha_2\sigma_y)
    e^{i\omega_s\sigma_z t/2}\sigma_x
    \right] e^{-\lambda_k t}, \\
    \mathrm{Im}[A_\mathrm{single}(t)] &= \frac{1}{2}\mathrm{Tr}\!\left[
    e^{-i\omega_s\sigma_z t/2}(\alpha_1\sigma_x+\alpha_2\sigma_y)
    e^{i\omega_s\sigma_z t/2}\sigma_y
    \right] e^{-\lambda_k t}.
\end{aligned}
\end{equation}
\end{widetext}
Neglecting relaxation, the expressions simplify to
\begin{equation}
\begin{aligned}
    \mathrm{Re}[A_\mathrm{single}(t)] &= \frac{\alpha_1\cos\omega_s t - \alpha_2\sin\omega_s t}{2},\\
    \mathrm{Im}[A_\mathrm{single}(t)] &= \frac{\alpha_2\cos\omega_s t + \alpha_1\sin\omega_s t}{2}.
\end{aligned}
\end{equation}
After Fourier transformation,
\begin{equation}
\begin{aligned}
    \mathrm{Re}[A_\mathrm{single}(\omega)] &= \frac{\alpha_1+\alpha_2}{2}\,\delta(\omega-\omega_s),\\
    \mathrm{Im}[A_\mathrm{single}(\omega)] &= \frac{\alpha_2-\alpha_1}{2}\,\delta(\omega-\omega_s).
\end{aligned}
\end{equation}
Thus, the amplitudes of the spectral peaks encode the coefficients of the Pauli operators, allowing us to reconstruct these coefficients from the observed spectrum.

Since the NMR spectrometer directly detects only transverse components ($\sigma_x$ and $\sigma_y$), accessing longitudinal information ($\sigma_z$) requires an additional rotation.
We apply a $\pi/2$ rotation about the $y$-axis, $R_y(\pi/2)$, which maps $\sigma_z$ to $\sigma_x$, allowing $\sigma_z$ components to be measured.

This analysis generalizes to multi-spin systems by incorporating spin-spin couplings.
For the four-spin crotonic acid system, the spectrum contains 32 distinct spectral lines [Fig.~\ref{spec}(b)], enabling the extraction of a large set of Pauli coefficients from a single measurement.

The GEWs studied in this work are constructed from a geometric optimization, rather than from a preselected set of experimentally convenient observables. More precisely, a GEW is determined by the nearest state in the desired set to the target state under the HS distance, and hence corresponds to the supporting hyperplane obtained from this closest-state geometry. The number of local measurement settings is therefore not an input constraint in our optimization. Imposing a fixed set or number of measurement settings would restrict the allowed witnesses to a smaller, experimentally tailored class; such witnesses may be useful in practice, but they are not generally guaranteed to be geometrically optimal in the sense considered here.

After a GEW has been obtained, its experimental implementation is handled as a separate step. The witness is decomposed into local Pauli operators, and compatible Pauli terms are grouped into common local measurement settings. Consequently, the required number of settings depends on the Pauli structure of the resulting GEW and on the grouping strategy used in the experiment. When a noise model is explicitly incorporated to improve the visibility threshold, the optimized witness may acquire a different Pauli structure; therefore, the measurement cost need not scale monotonically with the threshold. We therefore distinguish between the construction of the geometrically optimal GEW and the subsequent Pauli decomposition and grouping used for implementation.

To experimentally detect entanglement, we evaluate the expectation value of the corresponding witness operator $W$.
The witness is first embedded into the four-qubit Hilbert space according to the encoding of the target state and then decomposed into the local Pauli basis.
Importantly, the witness evaluation itself does not require full quantum state tomography; only the Pauli components with nonzero coefficients in the witness decomposition need to be measured.

\begin{equation}
    \langle W \rangle_{\rm exp}
    = \mathrm{Tr}(W\rho)
    = \mathrm{Tr}\!\left( \sum_i w_i \sigma_{4,i}\,\rho \right)
    = \sum_i w_i \langle \sigma_{4,i} \rangle_{\rm exp},
\end{equation}
where $\sigma_{4,i}$ denotes a four-qubit Pauli operator and $w_i$ is the corresponding coefficient in the witness expansion.
Thus, the experimentally required measurements are determined by the nonzero Pauli terms appearing in $W$.

For the Horodecki bound-entanglement experiment shown in Fig.~\ref{fig:bound_entanglement}, a numerical approximation to the geometrically optimal witness $W_b$ is constructed for each target state $\rho_b$ by approximating the nearest separable state under the Hilbert--Schmidt distance.
The resulting $3\otimes3$ witness is embedded into the four-qubit Hilbert space using the qutrit encoding described in the state-preparation section.
We then decompose the embedded witness into local Pauli operators and evaluate $\langle W_b\rangle_{\rm exp}$ from the measured Pauli expectations according to the equation above.

As for multipartite entangled states, the construction of GEWs follows a similar approach. For example, the GEW for the three-qubit W state $|W_3\rangle$ can be expressed as
\begin{widetext}
\begin{equation}
\begin{aligned}
    W \approx & ~0.473\cdot III+0.053\cdot \sigma^1_z\sigma^2_z\sigma^3_z
    + 0.017\cdot(I\sigma^2_z\sigma^3_z+\sigma^1_zI\sigma^3_z+\sigma^1_z\sigma^2_zI)\\
    &-0.096\cdot(I\sigma^2_x\sigma^3_x+\sigma^1_xI\sigma^3_x+\sigma^1_x\sigma^2_xI
    +I\sigma^2_y\sigma^3_y+\sigma^1_yI\sigma^3_y+\sigma^1_y\sigma^2_yI) \\
    &-0.105\cdot(\sigma^1_z\sigma^2_x\sigma^3_x+\sigma^1_x\sigma^2_z\sigma^3_x+\sigma^1_x\sigma^2_x\sigma^3_z
    +\sigma^1_z\sigma^2_y\sigma^3_y+\sigma^1_y\sigma^2_z\sigma^3_y+\sigma^1_y\sigma^2_y\sigma^3_z).
\end{aligned}
\end{equation}
\end{widetext}

The witness measurements described above are sufficient for entanglement detection.
Separately, to assess the quality of the experimentally prepared states, we reconstruct the full density matrices via quantum state tomography.
Complete tomography of a four-qubit state requires determining $4^4-1=255$ independent parameters, where the identity component $IIII$ in the Pauli expansion is fixed by normalization and is not independently measured.

In the NMR implementation, the following strings denote the readout rotation settings used in the state-interrogation tomography protocol, rather than independently measured Pauli-basis operators.
Here, $I$ indicates that no readout rotation is applied to the corresponding spin.
Thus, the setting $IIII$ corresponds to direct acquisition without additional readout rotations.
Following the standard state-interrogation method for NMR quantum-information processing~\cite{Leskowitz2004StateInterrogation}, complete tomography of a four-qubit state can be achieved with $2^4+1=17$ readout experiments.
For each prepared state reported in this work, all 17 readout experiments were acquired. The GEW expectation value was evaluated from the Pauli components required by that witness within the same data set, while the complete data set was also used for independent state-tomography characterization. Thus, 17 is the number of readouts actually performed per prepared state in our implementation, not a claim for the minimum number theoretically required to evaluate a given witness.

We use the following 17 readout rotation settings:
\begin{widetext}
\[
IIII,\
\sigma^1_x\sigma^2_x\sigma^3_x\sigma^4_x,\
II\sigma^3_y\sigma^4_y,\
\sigma^1_y\sigma^2_y\sigma^3_x\sigma^4_x,\
III\sigma^4_y,\
\sigma^1_x\sigma^2_y\sigma^3_x\sigma^4_x,\
\sigma^1_y\sigma^2_x\sigma^3_y I,
\]
\[
I\sigma^2_x\sigma^3_y I,\
III\sigma^4_x,\
\sigma^1_x I\sigma^3_y\sigma^4_y,\
\sigma^1_y\sigma^2_x II,\
\sigma^1_y\sigma^2_y\sigma^3_x\sigma^4_y,\
\sigma^1_x\sigma^2_y\sigma^3_x I,
\]
\[
II\sigma^3_y\sigma^4_x,\
I\sigma^2_x I\sigma^4_y,\
II\sigma^3_x I,\
I\sigma^2_y I\sigma^4_y.
\]
\end{widetext}
In this protocol, each readout rotation maps a specific set of density-matrix elements onto observable single-quantum coherences in the NMR spectrum.
Therefore, settings such as $II\sigma^3_y\sigma^4_y$ and $III\sigma^4_y$ are distinct readout settings, rather than lower-weight marginals obtained from another setting.

Combining the data from these experiments yields the full reconstructed density matrix.
As an example of state-quality characterization, Fig.~\ref{densitym} shows the real and imaginary parts of the experimentally reconstructed density matrix of the bound entangled state $\rho_b$ with $b=1$.
Here, the original $3\otimes3$ state is embedded into the four-qubit Hilbert space using the qutrit encoding described in the state-preparation section, with the unused level $|11\rangle$ in each two-qubit subsystem lying outside the encoded qutrit subspace.

To quantify the state preparation accuracy, we compute the fidelity
\begin{equation}
    F(\rho_{\mathrm{the}},\rho_{\mathrm{exp}})
    = \mathrm{Tr}\!\left(\sqrt{\rho_{\mathrm{the}}^{1/2}
        \rho_{\mathrm{exp}}
        \rho_{\mathrm{the}}^{1/2}}\right),
\end{equation}
where $\rho_{\mathrm{the}}$ is the ideal theoretical state and $\rho_{\mathrm{exp}}$ is the experimentally reconstructed density matrix.
All prepared states in our experiments achieve fidelities above 98\%.

The target states in this work are specified by the preparation protocol and are therefore known theoretically before the experiment. The corresponding GEWs are computed from these known target states by finding the nearest state in the relevant convex set, such as the separable, biseparable, or low-Schmidt-number set. The experimental task is then to measure the expectation value of the precomputed witness on the actually prepared state. Thus, tomography of an unknown target state is not required for constructing the witness; in our implementation, tomography is used only as an independent characterization of the preparation quality.

\begin{figure}[H]
    \centering
    \includegraphics[width=\linewidth]{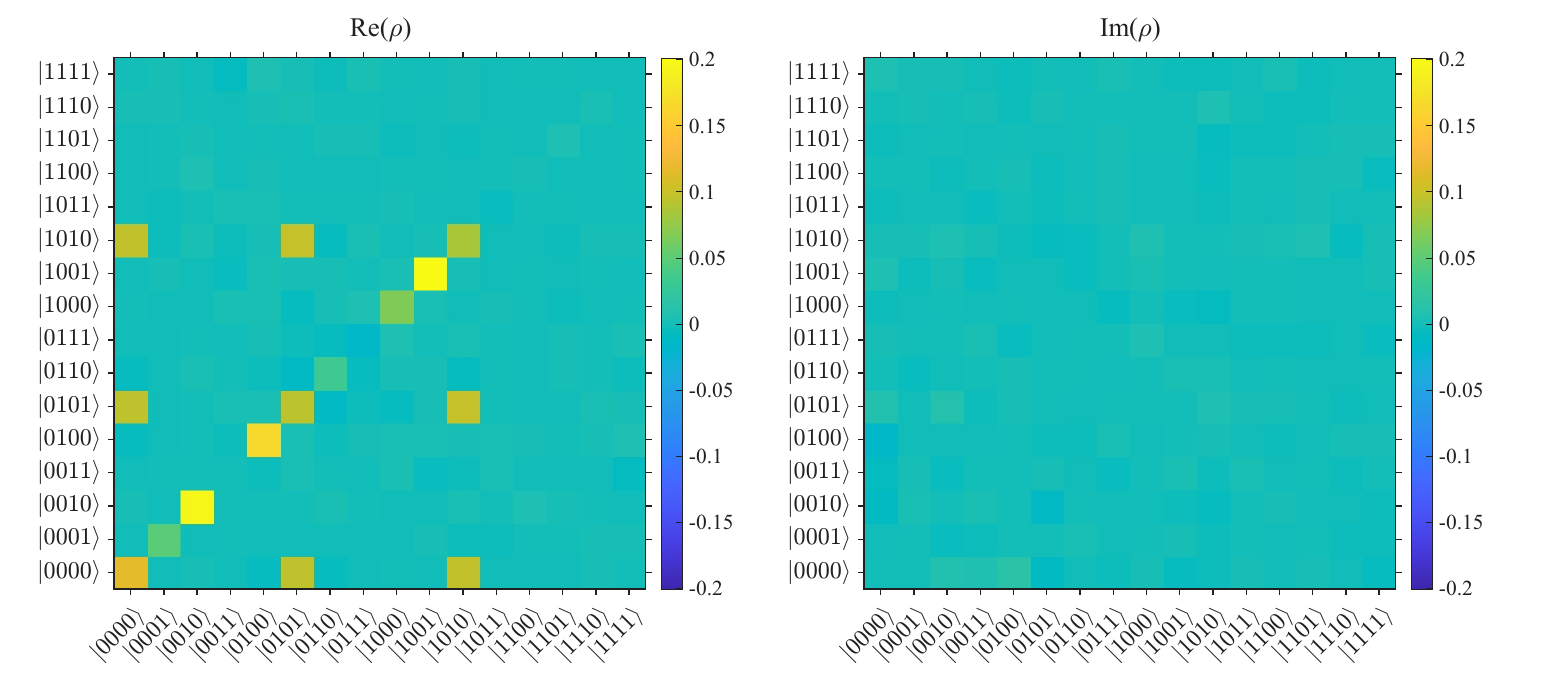}
    \caption{Four-qubit quantum state tomography result for the embedded Horodecki state $\rho_b$ with $b=1$.}
    \label{densitym}
\end{figure}

\bibliography{reference}

\begin{thebibliography}{71}%
\makeatletter
\providecommand \@ifxundefined [1]{%
 \@ifx{#1\undefined}
}%
\providecommand \@ifnum [1]{%
 \ifnum #1\expandafter \@firstoftwo
 \else \expandafter \@secondoftwo
 \fi
}%
\providecommand \@ifx [1]{%
 \ifx #1\expandafter \@firstoftwo
 \else \expandafter \@secondoftwo
 \fi
}%
\providecommand \natexlab [1]{#1}%
\providecommand \enquote  [1]{``#1''}%
\providecommand \bibnamefont  [1]{#1}%
\providecommand \bibfnamefont [1]{#1}%
\providecommand \citenamefont [1]{#1}%
\providecommand \href@noop [0]{\@secondoftwo}%
\providecommand \href [0]{\begingroup \@sanitize@url \@href}%
\providecommand \@href[1]{\@@startlink{#1}\@@href}%
\providecommand \@@href[1]{\endgroup#1\@@endlink}%
\providecommand \@sanitize@url [0]{\catcode `\\12\catcode `\$12\catcode
  `\&12\catcode `\#12\catcode `\^12\catcode `\_12\catcode `\%12\relax}%
\providecommand \@@startlink[1]{}%
\providecommand \@@endlink[0]{}%
\providecommand \url  [0]{\begingroup\@sanitize@url \@url }%
\providecommand \@url [1]{\endgroup\@href {#1}{\urlprefix }}%
\providecommand \urlprefix  [0]{URL }%
\providecommand \Eprint [0]{\href }%
\providecommand \doibase [0]{https://doi.org/}%
\providecommand \selectlanguage [0]{\@gobble}%
\providecommand \bibinfo  [0]{\@secondoftwo}%
\providecommand \bibfield  [0]{\@secondoftwo}%
\providecommand \translation [1]{[#1]}%
\providecommand \BibitemOpen [0]{}%
\providecommand \bibitemStop [0]{}%
\providecommand \bibitemNoStop [0]{.\EOS\space}%
\providecommand \EOS [0]{\spacefactor3000\relax}%
\providecommand \BibitemShut  [1]{\csname bibitem#1\endcsname}%
\let\auto@bib@innerbib\@empty
\bibitem [{\citenamefont {Bennett}\ \emph {et~al.}(1993)\citenamefont
  {Bennett}, \citenamefont {Brassard}, \citenamefont {Cr\'epeau}, \citenamefont
  {Jozsa}, \citenamefont {Peres},\ and\ \citenamefont
  {Wootters}}]{PhysRevLett.70.1895}%
  \BibitemOpen
  \bibfield  {author} {\bibinfo {author} {\bibfnamefont {C.~H.}\ \bibnamefont
  {Bennett}}, \bibinfo {author} {\bibfnamefont {G.}~\bibnamefont {Brassard}},
  \bibinfo {author} {\bibfnamefont {C.}~\bibnamefont {Cr\'epeau}}, \bibinfo
  {author} {\bibfnamefont {R.}~\bibnamefont {Jozsa}}, \bibinfo {author}
  {\bibfnamefont {A.}~\bibnamefont {Peres}},\ and\ \bibinfo {author}
  {\bibfnamefont {W.~K.}\ \bibnamefont {Wootters}},\ }\bibfield  {title}
  {\bibinfo {title} {Teleporting an unknown quantum state via dual classical
  and einstein-podolsky-rosen channels},\ }\href
  {https://doi.org/10.1103/PhysRevLett.70.1895} {\bibfield  {journal} {\bibinfo
   {journal} {Phys. Rev. Lett.}\ }\textbf {\bibinfo {volume} {70}},\ \bibinfo
  {pages} {1895} (\bibinfo {year} {1993})}\BibitemShut {NoStop}%
\bibitem [{\citenamefont {Steane}(1998)}]{steane1998quantum}%
  \BibitemOpen
  \bibfield  {author} {\bibinfo {author} {\bibfnamefont {A.}~\bibnamefont
  {Steane}},\ }\bibfield  {title} {\bibinfo {title} {Quantum computing},\
  }\href {https://doi.org/10.1088/0034-4885/61/2/002} {\bibfield  {journal}
  {\bibinfo  {journal} {Rep. Prog. Phys.}\ }\textbf {\bibinfo {volume} {61}},\
  \bibinfo {pages} {117} (\bibinfo {year} {1998})}\BibitemShut {NoStop}%
\bibitem [{\citenamefont {Zhao}\ \emph {et~al.}(2025)\citenamefont {Zhao},
  \citenamefont {Zhou},\ and\ \citenamefont {Childs}}]{zhao2025entanglement}%
  \BibitemOpen
  \bibfield  {author} {\bibinfo {author} {\bibfnamefont {Q.}~\bibnamefont
  {Zhao}}, \bibinfo {author} {\bibfnamefont {Y.}~\bibnamefont {Zhou}},\ and\
  \bibinfo {author} {\bibfnamefont {A.~M.}\ \bibnamefont {Childs}},\ }\bibfield
   {title} {\bibinfo {title} {Entanglement accelerates quantum simulation},\
  }\href {https://doi.org/https://doi.org/10.1038/s41567-025-02945-2}
  {\bibfield  {journal} {\bibinfo  {journal} {Nat. Phys.}\ }\textbf {\bibinfo
  {volume} {21}},\ \bibinfo {pages} {1338} (\bibinfo {year}
  {2025})}\BibitemShut {NoStop}%
\bibitem [{\citenamefont {Giovannetti}\ \emph {et~al.}(2006)\citenamefont
  {Giovannetti}, \citenamefont {Lloyd},\ and\ \citenamefont
  {Maccone}}]{PhysRevLett.96.010401}%
  \BibitemOpen
  \bibfield  {author} {\bibinfo {author} {\bibfnamefont {V.}~\bibnamefont
  {Giovannetti}}, \bibinfo {author} {\bibfnamefont {S.}~\bibnamefont {Lloyd}},\
  and\ \bibinfo {author} {\bibfnamefont {L.}~\bibnamefont {Maccone}},\
  }\bibfield  {title} {\bibinfo {title} {Quantum metrology},\ }\href
  {https://doi.org/10.1103/PhysRevLett.96.010401} {\bibfield  {journal}
  {\bibinfo  {journal} {Phys. Rev. Lett.}\ }\textbf {\bibinfo {volume} {96}},\
  \bibinfo {pages} {010401} (\bibinfo {year} {2006})}\BibitemShut {NoStop}%
\bibitem [{\citenamefont {Friis}\ \emph {et~al.}(2019)\citenamefont {Friis},
  \citenamefont {Vitagliano}, \citenamefont {Malik},\ and\ \citenamefont
  {Huber}}]{friis2019entanglement}%
  \BibitemOpen
  \bibfield  {author} {\bibinfo {author} {\bibfnamefont {N.}~\bibnamefont
  {Friis}}, \bibinfo {author} {\bibfnamefont {G.}~\bibnamefont {Vitagliano}},
  \bibinfo {author} {\bibfnamefont {M.}~\bibnamefont {Malik}},\ and\ \bibinfo
  {author} {\bibfnamefont {M.}~\bibnamefont {Huber}},\ }\bibfield  {title}
  {\bibinfo {title} {Entanglement certification from theory to experiment},\
  }\href {https://doi.org/10.1038/s42254-018-0003-5} {\bibfield  {journal}
  {\bibinfo  {journal} {Nat Rev Phys}\ }\textbf {\bibinfo {volume} {1}},\
  \bibinfo {pages} {72} (\bibinfo {year} {2019})}\BibitemShut {NoStop}%
\bibitem [{\citenamefont {G{\"u}hne}\ and\ \citenamefont
  {T{\'o}th}(2009)}]{guhne2009entanglement}%
  \BibitemOpen
  \bibfield  {author} {\bibinfo {author} {\bibfnamefont {O.}~\bibnamefont
  {G{\"u}hne}}\ and\ \bibinfo {author} {\bibfnamefont {G.}~\bibnamefont
  {T{\'o}th}},\ }\bibfield  {title} {\bibinfo {title} {Entanglement
  detection},\ }\href {https://doi.org/10.1016/j.physrep.2009.02.004}
  {\bibfield  {journal} {\bibinfo  {journal} {Physics Reports}\ }\textbf
  {\bibinfo {volume} {474}},\ \bibinfo {pages} {1} (\bibinfo {year}
  {2009})}\BibitemShut {NoStop}%
\bibitem [{\citenamefont {Horodecki}\ \emph {et~al.}(2009)\citenamefont
  {Horodecki}, \citenamefont {Horodecki}, \citenamefont {Horodecki},\ and\
  \citenamefont {Horodecki}}]{RevModPhys.81.865}%
  \BibitemOpen
  \bibfield  {author} {\bibinfo {author} {\bibfnamefont {R.}~\bibnamefont
  {Horodecki}}, \bibinfo {author} {\bibfnamefont {P.}~\bibnamefont
  {Horodecki}}, \bibinfo {author} {\bibfnamefont {M.}~\bibnamefont
  {Horodecki}},\ and\ \bibinfo {author} {\bibfnamefont {K.}~\bibnamefont
  {Horodecki}},\ }\bibfield  {title} {\bibinfo {title} {Quantum entanglement},\
  }\href {https://doi.org/10.1103/RevModPhys.81.865} {\bibfield  {journal}
  {\bibinfo  {journal} {Rev. Mod. Phys.}\ }\textbf {\bibinfo {volume} {81}},\
  \bibinfo {pages} {865} (\bibinfo {year} {2009})}\BibitemShut {NoStop}%
\bibitem [{\citenamefont {Terhal}(2000)}]{terhal2000bell}%
  \BibitemOpen
  \bibfield  {author} {\bibinfo {author} {\bibfnamefont {B.~M.}\ \bibnamefont
  {Terhal}},\ }\bibfield  {title} {\bibinfo {title} {Bell inequalities and the
  separability criterion},\ }\href
  {https://doi.org/10.1016/S0375-9601(00)00401-1} {\bibfield  {journal}
  {\bibinfo  {journal} {Phys. Lett. A}\ }\textbf {\bibinfo {volume} {271}},\
  \bibinfo {pages} {319} (\bibinfo {year} {2000})}\BibitemShut {NoStop}%
\bibitem [{\citenamefont {Lewenstein}\ \emph {et~al.}(2000)\citenamefont
  {Lewenstein}, \citenamefont {Kraus}, \citenamefont {Cirac},\ and\
  \citenamefont {Horodecki}}]{PhysRevA.62.052310}%
  \BibitemOpen
  \bibfield  {author} {\bibinfo {author} {\bibfnamefont {M.}~\bibnamefont
  {Lewenstein}}, \bibinfo {author} {\bibfnamefont {B.}~\bibnamefont {Kraus}},
  \bibinfo {author} {\bibfnamefont {J.~I.}\ \bibnamefont {Cirac}},\ and\
  \bibinfo {author} {\bibfnamefont {P.}~\bibnamefont {Horodecki}},\ }\bibfield
  {title} {\bibinfo {title} {Optimization of entanglement witnesses},\ }\href
  {https://doi.org/10.1103/PhysRevA.62.052310} {\bibfield  {journal} {\bibinfo
  {journal} {Phys. Rev. A}\ }\textbf {\bibinfo {volume} {62}},\ \bibinfo
  {pages} {052310} (\bibinfo {year} {2000})}\BibitemShut {NoStop}%
\bibitem [{\citenamefont {Cramer}\ \emph {et~al.}(2010)\citenamefont {Cramer},
  \citenamefont {Plenio}, \citenamefont {Flammia}, \citenamefont {Somma},
  \citenamefont {Gross}, \citenamefont {Bartlett}, \citenamefont
  {Landon-Cardinal}, \citenamefont {Poulin},\ and\ \citenamefont
  {Liu}}]{cramer2010efficient}%
  \BibitemOpen
  \bibfield  {author} {\bibinfo {author} {\bibfnamefont {M.}~\bibnamefont
  {Cramer}}, \bibinfo {author} {\bibfnamefont {M.~B.}\ \bibnamefont {Plenio}},
  \bibinfo {author} {\bibfnamefont {S.~T.}\ \bibnamefont {Flammia}}, \bibinfo
  {author} {\bibfnamefont {R.}~\bibnamefont {Somma}}, \bibinfo {author}
  {\bibfnamefont {D.}~\bibnamefont {Gross}}, \bibinfo {author} {\bibfnamefont
  {S.~D.}\ \bibnamefont {Bartlett}}, \bibinfo {author} {\bibfnamefont
  {O.}~\bibnamefont {Landon-Cardinal}}, \bibinfo {author} {\bibfnamefont
  {D.}~\bibnamefont {Poulin}},\ and\ \bibinfo {author} {\bibfnamefont {Y.-K.}\
  \bibnamefont {Liu}},\ }\bibfield  {title} {\bibinfo {title} {Efficient
  quantum state tomography},\ }\href {https://doi.org/10.1038/ncomms1147}
  {\bibfield  {journal} {\bibinfo  {journal} {Nat Commun}\ }\textbf {\bibinfo
  {volume} {1}},\ \bibinfo {pages} {149} (\bibinfo {year} {2010})}\BibitemShut
  {NoStop}%
\bibitem [{\citenamefont {Peres}(1996)}]{PhysRevLett.77.1413}%
  \BibitemOpen
  \bibfield  {author} {\bibinfo {author} {\bibfnamefont {A.}~\bibnamefont
  {Peres}},\ }\bibfield  {title} {\bibinfo {title} {Separability criterion for
  density matrices},\ }\href {https://doi.org/10.1103/PhysRevLett.77.1413}
  {\bibfield  {journal} {\bibinfo  {journal} {Phys. Rev. Lett.}\ }\textbf
  {\bibinfo {volume} {77}},\ \bibinfo {pages} {1413} (\bibinfo {year}
  {1996})}\BibitemShut {NoStop}%
\bibitem [{\citenamefont {Rudolph}(2005)}]{rudolph2005further}%
  \BibitemOpen
  \bibfield  {author} {\bibinfo {author} {\bibfnamefont {O.}~\bibnamefont
  {Rudolph}},\ }\bibfield  {title} {\bibinfo {title} {Further results on the
  cross norm criterion for separability},\ }\href
  {https://doi.org/10.1007/s11128-005-5664-1} {\bibfield  {journal} {\bibinfo
  {journal} {Quantum Inf Process}\ }\textbf {\bibinfo {volume} {4}},\ \bibinfo
  {pages} {219} (\bibinfo {year} {2005})}\BibitemShut {NoStop}%
\bibitem [{\citenamefont {Chen}\ and\ \citenamefont
  {Wu}(2003)}]{chen2003matrix}%
  \BibitemOpen
  \bibfield  {author} {\bibinfo {author} {\bibfnamefont {K.}~\bibnamefont
  {Chen}}\ and\ \bibinfo {author} {\bibfnamefont {L.-A.}\ \bibnamefont {Wu}},\
  }\bibfield  {title} {\bibinfo {title} {A matrix realignment method for
  recognizing entanglement},\ }\href
  {https://dl.acm.org/doi/abs/10.5555/2011534.2011535} {\bibfield  {journal}
  {\bibinfo  {journal} {Quantum Info. Comput.}\ }\textbf {\bibinfo {volume}
  {3}},\ \bibinfo {pages} {193–202} (\bibinfo {year} {2003})}\BibitemShut
  {NoStop}%
\bibitem [{\citenamefont {Bourennane}\ \emph {et~al.}(2004)\citenamefont
  {Bourennane}, \citenamefont {Eibl}, \citenamefont {Kurtsiefer}, \citenamefont
  {Gaertner}, \citenamefont {Weinfurter}, \citenamefont {G\"uhne},
  \citenamefont {Hyllus}, \citenamefont {Bru\ss{}}, \citenamefont
  {Lewenstein},\ and\ \citenamefont {Sanpera}}]{PhysRevLett.92.087902}%
  \BibitemOpen
  \bibfield  {author} {\bibinfo {author} {\bibfnamefont {M.}~\bibnamefont
  {Bourennane}}, \bibinfo {author} {\bibfnamefont {M.}~\bibnamefont {Eibl}},
  \bibinfo {author} {\bibfnamefont {C.}~\bibnamefont {Kurtsiefer}}, \bibinfo
  {author} {\bibfnamefont {S.}~\bibnamefont {Gaertner}}, \bibinfo {author}
  {\bibfnamefont {H.}~\bibnamefont {Weinfurter}}, \bibinfo {author}
  {\bibfnamefont {O.}~\bibnamefont {G\"uhne}}, \bibinfo {author} {\bibfnamefont
  {P.}~\bibnamefont {Hyllus}}, \bibinfo {author} {\bibfnamefont
  {D.}~\bibnamefont {Bru\ss{}}}, \bibinfo {author} {\bibfnamefont
  {M.}~\bibnamefont {Lewenstein}},\ and\ \bibinfo {author} {\bibfnamefont
  {A.}~\bibnamefont {Sanpera}},\ }\bibfield  {title} {\bibinfo {title}
  {Experimental detection of multipartite entanglement using witness
  operators},\ }\href {https://doi.org/10.1103/PhysRevLett.92.087902}
  {\bibfield  {journal} {\bibinfo  {journal} {Phys. Rev. Lett.}\ }\textbf
  {\bibinfo {volume} {92}},\ \bibinfo {pages} {087902} (\bibinfo {year}
  {2004})}\BibitemShut {NoStop}%
\bibitem [{\citenamefont {T\'oth}\ and\ \citenamefont
  {G\"uhne}(2005)}]{PhysRevLett.94.060501}%
  \BibitemOpen
  \bibfield  {author} {\bibinfo {author} {\bibfnamefont {G.}~\bibnamefont
  {T\'oth}}\ and\ \bibinfo {author} {\bibfnamefont {O.}~\bibnamefont
  {G\"uhne}},\ }\bibfield  {title} {\bibinfo {title} {Detecting genuine
  multipartite entanglement with two local measurements},\ }\href
  {https://doi.org/10.1103/PhysRevLett.94.060501} {\bibfield  {journal}
  {\bibinfo  {journal} {Phys. Rev. Lett.}\ }\textbf {\bibinfo {volume} {94}},\
  \bibinfo {pages} {060501} (\bibinfo {year} {2005})}\BibitemShut {NoStop}%
\bibitem [{\citenamefont {G{\"u}hne}\ \emph {et~al.}(2003)\citenamefont
  {G{\"u}hne}, \citenamefont {Hyllus}, \citenamefont {Bru{\ss}}, \citenamefont
  {Ekert}, \citenamefont {Lewenstein}, \citenamefont {Macchiavello},\ and\
  \citenamefont {Sanpera}}]{guhne2003experimental}%
  \BibitemOpen
  \bibfield  {author} {\bibinfo {author} {\bibfnamefont {O.}~\bibnamefont
  {G{\"u}hne}}, \bibinfo {author} {\bibfnamefont {P.}~\bibnamefont {Hyllus}},
  \bibinfo {author} {\bibfnamefont {D.}~\bibnamefont {Bru{\ss}}}, \bibinfo
  {author} {\bibfnamefont {A.}~\bibnamefont {Ekert}}, \bibinfo {author}
  {\bibfnamefont {M.}~\bibnamefont {Lewenstein}}, \bibinfo {author}
  {\bibfnamefont {C.}~\bibnamefont {Macchiavello}},\ and\ \bibinfo {author}
  {\bibfnamefont {A.}~\bibnamefont {Sanpera}},\ }\bibfield  {title} {\bibinfo
  {title} {Experimental detection of entanglement via witness operators and
  local measurements},\ }\href {https://doi.org/10.1080/09500340308234554}
  {\bibfield  {journal} {\bibinfo  {journal} {Journal of Modern Optics}\
  }\textbf {\bibinfo {volume} {50}},\ \bibinfo {pages} {1079} (\bibinfo {year}
  {2003})}\BibitemShut {NoStop}%
\bibitem [{\citenamefont {Gerke}\ \emph {et~al.}(2018)\citenamefont {Gerke},
  \citenamefont {Vogel},\ and\ \citenamefont {Sperling}}]{PhysRevX.8.031047}%
  \BibitemOpen
  \bibfield  {author} {\bibinfo {author} {\bibfnamefont {S.}~\bibnamefont
  {Gerke}}, \bibinfo {author} {\bibfnamefont {W.}~\bibnamefont {Vogel}},\ and\
  \bibinfo {author} {\bibfnamefont {J.}~\bibnamefont {Sperling}},\ }\bibfield
  {title} {\bibinfo {title} {Numerical construction of multipartite
  entanglement witnesses},\ }\href {https://doi.org/10.1103/PhysRevX.8.031047}
  {\bibfield  {journal} {\bibinfo  {journal} {Phys. Rev. X}\ }\textbf {\bibinfo
  {volume} {8}},\ \bibinfo {pages} {031047} (\bibinfo {year}
  {2018})}\BibitemShut {NoStop}%
\bibitem [{\citenamefont {Sperling}\ and\ \citenamefont
  {Vogel}(2013)}]{PhysRevLett.111.110503}%
  \BibitemOpen
  \bibfield  {author} {\bibinfo {author} {\bibfnamefont {J.}~\bibnamefont
  {Sperling}}\ and\ \bibinfo {author} {\bibfnamefont {W.}~\bibnamefont
  {Vogel}},\ }\bibfield  {title} {\bibinfo {title} {Multipartite entanglement
  witnesses},\ }\href {https://doi.org/10.1103/PhysRevLett.111.110503}
  {\bibfield  {journal} {\bibinfo  {journal} {Phys. Rev. Lett.}\ }\textbf
  {\bibinfo {volume} {111}},\ \bibinfo {pages} {110503} (\bibinfo {year}
  {2013})}\BibitemShut {NoStop}%
\bibitem [{\citenamefont {Sanpera}\ \emph {et~al.}(2001)\citenamefont
  {Sanpera}, \citenamefont {Bru\ss{}},\ and\ \citenamefont
  {Lewenstein}}]{PhysRevA.63.050301}%
  \BibitemOpen
  \bibfield  {author} {\bibinfo {author} {\bibfnamefont {A.}~\bibnamefont
  {Sanpera}}, \bibinfo {author} {\bibfnamefont {D.}~\bibnamefont {Bru\ss{}}},\
  and\ \bibinfo {author} {\bibfnamefont {M.}~\bibnamefont {Lewenstein}},\
  }\bibfield  {title} {\bibinfo {title} {Schmidt-number witnesses and bound
  entanglement},\ }\href {https://doi.org/10.1103/PhysRevA.63.050301}
  {\bibfield  {journal} {\bibinfo  {journal} {Phys. Rev. A}\ }\textbf {\bibinfo
  {volume} {63}},\ \bibinfo {pages} {050301} (\bibinfo {year}
  {2001})}\BibitemShut {NoStop}%
\bibitem [{\citenamefont {Wyderka}\ \emph {et~al.}(2023)\citenamefont
  {Wyderka}, \citenamefont {Chesi}, \citenamefont {Kampermann}, \citenamefont
  {Macchiavello},\ and\ \citenamefont {Bru\ss{}}}]{PhysRevA.107.022431}%
  \BibitemOpen
  \bibfield  {author} {\bibinfo {author} {\bibfnamefont {N.}~\bibnamefont
  {Wyderka}}, \bibinfo {author} {\bibfnamefont {G.}~\bibnamefont {Chesi}},
  \bibinfo {author} {\bibfnamefont {H.}~\bibnamefont {Kampermann}}, \bibinfo
  {author} {\bibfnamefont {C.}~\bibnamefont {Macchiavello}},\ and\ \bibinfo
  {author} {\bibfnamefont {D.}~\bibnamefont {Bru\ss{}}},\ }\bibfield  {title}
  {\bibinfo {title} {Construction of efficient schmidt-number witnesses for
  high-dimensional quantum states},\ }\href
  {https://doi.org/10.1103/PhysRevA.107.022431} {\bibfield  {journal} {\bibinfo
   {journal} {Phys. Rev. A}\ }\textbf {\bibinfo {volume} {107}},\ \bibinfo
  {pages} {022431} (\bibinfo {year} {2023})}\BibitemShut {NoStop}%
\bibitem [{\citenamefont {Terhal}\ and\ \citenamefont
  {Horodecki}(2000)}]{PhysRevA.61.040301}%
  \BibitemOpen
  \bibfield  {author} {\bibinfo {author} {\bibfnamefont {B.~M.}\ \bibnamefont
  {Terhal}}\ and\ \bibinfo {author} {\bibfnamefont {P.}~\bibnamefont
  {Horodecki}},\ }\bibfield  {title} {\bibinfo {title} {Schmidt number for
  density matrices},\ }\href {https://doi.org/10.1103/PhysRevA.61.040301}
  {\bibfield  {journal} {\bibinfo  {journal} {Phys. Rev. A}\ }\textbf {\bibinfo
  {volume} {61}},\ \bibinfo {pages} {040301} (\bibinfo {year}
  {2000})}\BibitemShut {NoStop}%
\bibitem [{\citenamefont {Gurvits}(2003)}]{10.1145/780542.780545}%
  \BibitemOpen
  \bibfield  {author} {\bibinfo {author} {\bibfnamefont {L.}~\bibnamefont
  {Gurvits}},\ }\bibfield  {title} {\bibinfo {title} {Classical deterministic
  complexity of edmonds' problem and quantum entanglement},\ }in\ \href
  {https://doi.org/10.1145/780542.780545} {\emph {\bibinfo {booktitle}
  {Proceedings of the Thirty-Fifth Annual ACM Symposium on Theory of
  Computing}}},\ \bibinfo {series and number} {STOC '03}\ (\bibinfo
  {publisher} {Association for Computing Machinery},\ \bibinfo {address} {New
  York, NY, USA},\ \bibinfo {year} {2003})\ pp.\ \bibinfo {pages}
  {10--19}\BibitemShut {NoStop}%
\bibitem [{\citenamefont {Ioannou}(2007)}]{ioannou2006computational}%
  \BibitemOpen
  \bibfield  {author} {\bibinfo {author} {\bibfnamefont {L.~M.}\ \bibnamefont
  {Ioannou}},\ }\bibfield  {title} {\bibinfo {title} {Computational complexity
  of the quantum separability problem},\ }\href@noop {} {\bibfield  {journal}
  {\bibinfo  {journal} {Quantum Inf. Comput.}\ }\textbf {\bibinfo {volume}
  {7}},\ \bibinfo {pages} {335} (\bibinfo {year} {2007})}\BibitemShut {NoStop}%
\bibitem [{\citenamefont {Krammer}(2009)}]{krammer2009characterizing}%
  \BibitemOpen
  \bibfield  {author} {\bibinfo {author} {\bibfnamefont {P.}~\bibnamefont
  {Krammer}},\ }\bibfield  {title} {\bibinfo {title} {Characterizing
  entanglement with geometric entanglement witnesses},\ }\href
  {https://doi.org/10.1088/1751-8113/42/6/065305} {\bibfield  {journal}
  {\bibinfo  {journal} {J. Phys. A: Math. Theor.}\ }\textbf {\bibinfo {volume}
  {42}},\ \bibinfo {pages} {065305} (\bibinfo {year} {2009})}\BibitemShut
  {NoStop}%
\bibitem [{\citenamefont {Bertlmann}\ \emph {et~al.}(2005)\citenamefont
  {Bertlmann}, \citenamefont {Durstberger}, \citenamefont {Hiesmayr},\ and\
  \citenamefont {Krammer}}]{PhysRevA.72.052331}%
  \BibitemOpen
  \bibfield  {author} {\bibinfo {author} {\bibfnamefont {R.~A.}\ \bibnamefont
  {Bertlmann}}, \bibinfo {author} {\bibfnamefont {K.}~\bibnamefont
  {Durstberger}}, \bibinfo {author} {\bibfnamefont {B.~C.}\ \bibnamefont
  {Hiesmayr}},\ and\ \bibinfo {author} {\bibfnamefont {P.}~\bibnamefont
  {Krammer}},\ }\bibfield  {title} {\bibinfo {title} {Optimal entanglement
  witnesses for qubits and qutrits},\ }\href
  {https://doi.org/10.1103/PhysRevA.72.052331} {\bibfield  {journal} {\bibinfo
  {journal} {Phys. Rev. A}\ }\textbf {\bibinfo {volume} {72}},\ \bibinfo
  {pages} {052331} (\bibinfo {year} {2005})}\BibitemShut {NoStop}%
\bibitem [{\citenamefont {Bertlmann}\ and\ \citenamefont
  {Krammer}(2008)}]{PhysRevA.77.024303}%
  \BibitemOpen
  \bibfield  {author} {\bibinfo {author} {\bibfnamefont {R.~A.}\ \bibnamefont
  {Bertlmann}}\ and\ \bibinfo {author} {\bibfnamefont {P.}~\bibnamefont
  {Krammer}},\ }\bibfield  {title} {\bibinfo {title} {Geometric entanglement
  witnesses and bound entanglement},\ }\href
  {https://doi.org/10.1103/PhysRevA.77.024303} {\bibfield  {journal} {\bibinfo
  {journal} {Phys. Rev. A}\ }\textbf {\bibinfo {volume} {77}},\ \bibinfo
  {pages} {024303} (\bibinfo {year} {2008})}\BibitemShut {NoStop}%
\bibitem [{\citenamefont {Bertlmann}\ \emph {et~al.}(2002)\citenamefont
  {Bertlmann}, \citenamefont {Narnhofer},\ and\ \citenamefont
  {Thirring}}]{PhysRevA.66.032319}%
  \BibitemOpen
  \bibfield  {author} {\bibinfo {author} {\bibfnamefont {R.~A.}\ \bibnamefont
  {Bertlmann}}, \bibinfo {author} {\bibfnamefont {H.}~\bibnamefont
  {Narnhofer}},\ and\ \bibinfo {author} {\bibfnamefont {W.}~\bibnamefont
  {Thirring}},\ }\bibfield  {title} {\bibinfo {title} {Geometric picture of
  entanglement and bell inequalities},\ }\href
  {https://doi.org/10.1103/PhysRevA.66.032319} {\bibfield  {journal} {\bibinfo
  {journal} {Phys. Rev. A}\ }\textbf {\bibinfo {volume} {66}},\ \bibinfo
  {pages} {032319} (\bibinfo {year} {2002})}\BibitemShut {NoStop}%
\bibitem [{\citenamefont {Reed}\ and\ \citenamefont
  {Simon}(1980)}]{reed1980methods}%
  \BibitemOpen
  \bibfield  {author} {\bibinfo {author} {\bibfnamefont {M.}~\bibnamefont
  {Reed}}\ and\ \bibinfo {author} {\bibfnamefont {B.}~\bibnamefont {Simon}},\
  }\href {https://doi.org/10.1016/B978-0-12-585001-8.X5001-6} {\emph {\bibinfo
  {title} {Methods of modern mathematical physics: Functional analysis}}},\
  Vol.~\bibinfo {volume} {1}\ (\bibinfo  {publisher} {Academic Press},\
  \bibinfo {address} {New York},\ \bibinfo {year} {1980})\BibitemShut {NoStop}%
\bibitem [{\citenamefont {Lu}\ \emph {et~al.}(2018{\natexlab{a}})\citenamefont
  {Lu}, \citenamefont {Zhao}, \citenamefont {Li}, \citenamefont {Yin},
  \citenamefont {Yuan}, \citenamefont {Hung}, \citenamefont {Chen},
  \citenamefont {Li}, \citenamefont {Liu}, \citenamefont {Peng}, \citenamefont
  {Liang}, \citenamefont {Ma}, \citenamefont {Chen},\ and\ \citenamefont
  {Pan}}]{PhysRevX.8.021072}%
  \BibitemOpen
  \bibfield  {author} {\bibinfo {author} {\bibfnamefont {H.}~\bibnamefont
  {Lu}}, \bibinfo {author} {\bibfnamefont {Q.}~\bibnamefont {Zhao}}, \bibinfo
  {author} {\bibfnamefont {Z.-D.}\ \bibnamefont {Li}}, \bibinfo {author}
  {\bibfnamefont {X.-F.}\ \bibnamefont {Yin}}, \bibinfo {author} {\bibfnamefont
  {X.}~\bibnamefont {Yuan}}, \bibinfo {author} {\bibfnamefont {J.-C.}\
  \bibnamefont {Hung}}, \bibinfo {author} {\bibfnamefont {L.-K.}\ \bibnamefont
  {Chen}}, \bibinfo {author} {\bibfnamefont {L.}~\bibnamefont {Li}}, \bibinfo
  {author} {\bibfnamefont {N.-L.}\ \bibnamefont {Liu}}, \bibinfo {author}
  {\bibfnamefont {C.-Z.}\ \bibnamefont {Peng}}, \bibinfo {author}
  {\bibfnamefont {Y.-C.}\ \bibnamefont {Liang}}, \bibinfo {author}
  {\bibfnamefont {X.}~\bibnamefont {Ma}}, \bibinfo {author} {\bibfnamefont
  {Y.-A.}\ \bibnamefont {Chen}},\ and\ \bibinfo {author} {\bibfnamefont
  {J.-W.}\ \bibnamefont {Pan}},\ }\bibfield  {title} {\bibinfo {title}
  {Entanglement structure: Entanglement partitioning in multipartite systems
  and its experimental detection using optimizable witnesses},\ }\href
  {https://doi.org/10.1103/PhysRevX.8.021072} {\bibfield  {journal} {\bibinfo
  {journal} {Phys. Rev. X}\ }\textbf {\bibinfo {volume} {8}},\ \bibinfo {pages}
  {021072} (\bibinfo {year} {2018}{\natexlab{a}})}\BibitemShut {NoStop}%
\bibitem [{\citenamefont {Zhou}\ \emph {et~al.}(2019)\citenamefont {Zhou},
  \citenamefont {Zhao}, \citenamefont {Yuan},\ and\ \citenamefont
  {Ma}}]{zhou2019detecting}%
  \BibitemOpen
  \bibfield  {author} {\bibinfo {author} {\bibfnamefont {Y.}~\bibnamefont
  {Zhou}}, \bibinfo {author} {\bibfnamefont {Q.}~\bibnamefont {Zhao}}, \bibinfo
  {author} {\bibfnamefont {X.}~\bibnamefont {Yuan}},\ and\ \bibinfo {author}
  {\bibfnamefont {X.}~\bibnamefont {Ma}},\ }\bibfield  {title} {\bibinfo
  {title} {Detecting multipartite entanglement structure with minimal
  resources},\ }\href {https://doi.org/10.1038/s41534-019-0200-9} {\bibfield
  {journal} {\bibinfo  {journal} {npj Quantum Inf}\ }\textbf {\bibinfo {volume}
  {5}},\ \bibinfo {pages} {83} (\bibinfo {year} {2019})}\BibitemShut {NoStop}%
\bibitem [{\citenamefont {Shahandeh}\ \emph {et~al.}(2014)\citenamefont
  {Shahandeh}, \citenamefont {Sperling},\ and\ \citenamefont
  {Vogel}}]{PhysRevLett.113.260502}%
  \BibitemOpen
  \bibfield  {author} {\bibinfo {author} {\bibfnamefont {F.}~\bibnamefont
  {Shahandeh}}, \bibinfo {author} {\bibfnamefont {J.}~\bibnamefont
  {Sperling}},\ and\ \bibinfo {author} {\bibfnamefont {W.}~\bibnamefont
  {Vogel}},\ }\bibfield  {title} {\bibinfo {title} {Structural quantification
  of entanglement},\ }\href {https://doi.org/10.1103/PhysRevLett.113.260502}
  {\bibfield  {journal} {\bibinfo  {journal} {Phys. Rev. Lett.}\ }\textbf
  {\bibinfo {volume} {113}},\ \bibinfo {pages} {260502} (\bibinfo {year}
  {2014})}\BibitemShut {NoStop}%
\bibitem [{\citenamefont {Huber}\ and\ \citenamefont
  {de~Vicente}(2013)}]{PhysRevLett.110.030501}%
  \BibitemOpen
  \bibfield  {author} {\bibinfo {author} {\bibfnamefont {M.}~\bibnamefont
  {Huber}}\ and\ \bibinfo {author} {\bibfnamefont {J.~I.}\ \bibnamefont
  {de~Vicente}},\ }\bibfield  {title} {\bibinfo {title} {Structure of
  multidimensional entanglement in multipartite systems},\ }\href
  {https://doi.org/10.1103/PhysRevLett.110.030501} {\bibfield  {journal}
  {\bibinfo  {journal} {Phys. Rev. Lett.}\ }\textbf {\bibinfo {volume} {110}},\
  \bibinfo {pages} {030501} (\bibinfo {year} {2013})}\BibitemShut {NoStop}%
\bibitem [{\citenamefont {Horodecki}\ \emph {et~al.}(1998)\citenamefont
  {Horodecki}, \citenamefont {Horodecki},\ and\ \citenamefont
  {Horodecki}}]{PhysRevLett.80.5239}%
  \BibitemOpen
  \bibfield  {author} {\bibinfo {author} {\bibfnamefont {M.}~\bibnamefont
  {Horodecki}}, \bibinfo {author} {\bibfnamefont {P.}~\bibnamefont
  {Horodecki}},\ and\ \bibinfo {author} {\bibfnamefont {R.}~\bibnamefont
  {Horodecki}},\ }\bibfield  {title} {\bibinfo {title} {Mixed-state
  entanglement and distillation: Is there a ``bound'' entanglement in
  nature?},\ }\href {https://doi.org/10.1103/PhysRevLett.80.5239} {\bibfield
  {journal} {\bibinfo  {journal} {Phys. Rev. Lett.}\ }\textbf {\bibinfo
  {volume} {80}},\ \bibinfo {pages} {5239} (\bibinfo {year}
  {1998})}\BibitemShut {NoStop}%
\bibitem [{\citenamefont {Walter}\ \emph {et~al.}(2016)\citenamefont {Walter},
  \citenamefont {Gross},\ and\ \citenamefont
  {Eisert}}]{walter2016multipartite}%
  \BibitemOpen
  \bibfield  {author} {\bibinfo {author} {\bibfnamefont {M.}~\bibnamefont
  {Walter}}, \bibinfo {author} {\bibfnamefont {D.}~\bibnamefont {Gross}},\ and\
  \bibinfo {author} {\bibfnamefont {J.}~\bibnamefont {Eisert}},\ }\bibfield
  {title} {\bibinfo {title} {Multipartite entanglement},\ }in\ \href
  {https://doi.org/10.1002/9783527805785.ch14} {\emph {\bibinfo {booktitle}
  {Quantum Information: From Foundations to Quantum Technology
  Applications}}},\ \bibinfo {editor} {edited by\ \bibinfo {editor}
  {\bibfnamefont {D.}~\bibnamefont {Bru{\ss}}}\ and\ \bibinfo {editor}
  {\bibfnamefont {G.}~\bibnamefont {Leuchs}}}\ (\bibinfo  {publisher} {Wiley
  Online Library},\ \bibinfo {address} {Weinheim},\ \bibinfo {year} {2016})\
  Chap.~\bibinfo {chapter} {14}, pp.\ \bibinfo {pages} {293--330},\ \bibinfo
  {note} {{Part III}}\BibitemShut {NoStop}%
\bibitem [{\citenamefont {Hu}\ \emph {et~al.}(2021)\citenamefont {Hu},
  \citenamefont {Xing}, \citenamefont {Guo}, \citenamefont {Weilenmann},
  \citenamefont {Aguilar}, \citenamefont {Gao}, \citenamefont {Liu},
  \citenamefont {Huang}, \citenamefont {Li}, \citenamefont {Guo}, \citenamefont
  {Wang},\ and\ \citenamefont {Navascu\'es}}]{PhysRevLett.127.220501}%
  \BibitemOpen
  \bibfield  {author} {\bibinfo {author} {\bibfnamefont {X.-M.}\ \bibnamefont
  {Hu}}, \bibinfo {author} {\bibfnamefont {W.-B.}\ \bibnamefont {Xing}},
  \bibinfo {author} {\bibfnamefont {Y.}~\bibnamefont {Guo}}, \bibinfo {author}
  {\bibfnamefont {M.}~\bibnamefont {Weilenmann}}, \bibinfo {author}
  {\bibfnamefont {E.~A.}\ \bibnamefont {Aguilar}}, \bibinfo {author}
  {\bibfnamefont {X.}~\bibnamefont {Gao}}, \bibinfo {author} {\bibfnamefont
  {B.-H.}\ \bibnamefont {Liu}}, \bibinfo {author} {\bibfnamefont {Y.-F.}\
  \bibnamefont {Huang}}, \bibinfo {author} {\bibfnamefont {C.-F.}\ \bibnamefont
  {Li}}, \bibinfo {author} {\bibfnamefont {G.-C.}\ \bibnamefont {Guo}},
  \bibinfo {author} {\bibfnamefont {Z.}~\bibnamefont {Wang}},\ and\ \bibinfo
  {author} {\bibfnamefont {M.}~\bibnamefont {Navascu\'es}},\ }\bibfield
  {title} {\bibinfo {title} {Optimized detection of high-dimensional
  entanglement},\ }\href {https://doi.org/10.1103/PhysRevLett.127.220501}
  {\bibfield  {journal} {\bibinfo  {journal} {Phys. Rev. Lett.}\ }\textbf
  {\bibinfo {volume} {127}},\ \bibinfo {pages} {220501} (\bibinfo {year}
  {2021})}\BibitemShut {NoStop}%
\bibitem [{\citenamefont {Erhard}\ \emph {et~al.}(2020)\citenamefont {Erhard},
  \citenamefont {Krenn},\ and\ \citenamefont {Zeilinger}}]{erhard2020advances}%
  \BibitemOpen
  \bibfield  {author} {\bibinfo {author} {\bibfnamefont {M.}~\bibnamefont
  {Erhard}}, \bibinfo {author} {\bibfnamefont {M.}~\bibnamefont {Krenn}},\ and\
  \bibinfo {author} {\bibfnamefont {A.}~\bibnamefont {Zeilinger}},\ }\bibfield
  {title} {\bibinfo {title} {Advances in high-dimensional quantum
  entanglement},\ }\href {https://doi.org/10.1038/s42254-020-0193-5} {\bibfield
   {journal} {\bibinfo  {journal} {Nat Rev Phys}\ }\textbf {\bibinfo {volume}
  {2}},\ \bibinfo {pages} {365} (\bibinfo {year} {2020})}\BibitemShut {NoStop}%
\bibitem [{\citenamefont {Bavaresco}\ \emph {et~al.}(2018)\citenamefont
  {Bavaresco}, \citenamefont {Herrera~Valencia}, \citenamefont {Kl{\"o}ckl},
  \citenamefont {Pivoluska}, \citenamefont {Erker}, \citenamefont {Friis},
  \citenamefont {Malik},\ and\ \citenamefont
  {Huber}}]{bavaresco2018measurements}%
  \BibitemOpen
  \bibfield  {author} {\bibinfo {author} {\bibfnamefont {J.}~\bibnamefont
  {Bavaresco}}, \bibinfo {author} {\bibfnamefont {N.}~\bibnamefont
  {Herrera~Valencia}}, \bibinfo {author} {\bibfnamefont {C.}~\bibnamefont
  {Kl{\"o}ckl}}, \bibinfo {author} {\bibfnamefont {M.}~\bibnamefont
  {Pivoluska}}, \bibinfo {author} {\bibfnamefont {P.}~\bibnamefont {Erker}},
  \bibinfo {author} {\bibfnamefont {N.}~\bibnamefont {Friis}}, \bibinfo
  {author} {\bibfnamefont {M.}~\bibnamefont {Malik}},\ and\ \bibinfo {author}
  {\bibfnamefont {M.}~\bibnamefont {Huber}},\ }\bibfield  {title} {\bibinfo
  {title} {Measurements in two bases are sufficient for certifying
  high-dimensional entanglement},\ }\href
  {https://doi.org/10.1038/s41567-018-0203-z} {\bibfield  {journal} {\bibinfo
  {journal} {Nat. Phys.}\ }\textbf {\bibinfo {volume} {14}},\ \bibinfo {pages}
  {1032} (\bibinfo {year} {2018})}\BibitemShut {NoStop}%
\bibitem [{\citenamefont {Goodfellow}\ \emph {et~al.}(2016)\citenamefont
  {Goodfellow}, \citenamefont {Bengio},\ and\ \citenamefont
  {Courville}}]{goodfellow2016deep}%
  \BibitemOpen
  \bibfield  {author} {\bibinfo {author} {\bibfnamefont {I.}~\bibnamefont
  {Goodfellow}}, \bibinfo {author} {\bibfnamefont {Y.}~\bibnamefont {Bengio}},\
  and\ \bibinfo {author} {\bibfnamefont {A.}~\bibnamefont {Courville}},\
  }\href@noop {} {\emph {\bibinfo {title} {Deep Learning}}}\ (\bibinfo
  {publisher} {MIT Press},\ \bibinfo {year} {2016})\ \bibinfo {note}
  {\url{http://www.deeplearningbook.org}}\BibitemShut {NoStop}%
\bibitem [{\citenamefont
  {Carath{\'e}odory}(1911)}]{caratheodory1911variabilitatsbereich}%
  \BibitemOpen
  \bibfield  {author} {\bibinfo {author} {\bibfnamefont {C.}~\bibnamefont
  {Carath{\'e}odory}},\ }\bibfield  {title} {\bibinfo {title} {{\"U}ber den
  variabilit{\"a}tsbereich der fourier’schen konstanten von positiven
  harmonischen funktionen},\ }\href@noop {} {\bibfield  {journal} {\bibinfo
  {journal} {Rendiconti Del Circolo Matematico di Palermo (1884-1940)}\
  }\textbf {\bibinfo {volume} {32}},\ \bibinfo {pages} {193} (\bibinfo {year}
  {1911})}\BibitemShut {NoStop}%
\bibitem [{\citenamefont {Lu}\ \emph {et~al.}(2018{\natexlab{b}})\citenamefont
  {Lu}, \citenamefont {Huang}, \citenamefont {Li}, \citenamefont {Li},
  \citenamefont {Chen}, \citenamefont {Lu}, \citenamefont {Ji}, \citenamefont
  {Shen}, \citenamefont {Zhou},\ and\ \citenamefont
  {Zeng}}]{PhysRevA.98.012315}%
  \BibitemOpen
  \bibfield  {author} {\bibinfo {author} {\bibfnamefont {S.}~\bibnamefont
  {Lu}}, \bibinfo {author} {\bibfnamefont {S.}~\bibnamefont {Huang}}, \bibinfo
  {author} {\bibfnamefont {K.}~\bibnamefont {Li}}, \bibinfo {author}
  {\bibfnamefont {J.}~\bibnamefont {Li}}, \bibinfo {author} {\bibfnamefont
  {J.}~\bibnamefont {Chen}}, \bibinfo {author} {\bibfnamefont {D.}~\bibnamefont
  {Lu}}, \bibinfo {author} {\bibfnamefont {Z.}~\bibnamefont {Ji}}, \bibinfo
  {author} {\bibfnamefont {Y.}~\bibnamefont {Shen}}, \bibinfo {author}
  {\bibfnamefont {D.}~\bibnamefont {Zhou}},\ and\ \bibinfo {author}
  {\bibfnamefont {B.}~\bibnamefont {Zeng}},\ }\bibfield  {title} {\bibinfo
  {title} {Separability-entanglement classifier via machine learning},\ }\href
  {https://doi.org/10.1103/PhysRevA.98.012315} {\bibfield  {journal} {\bibinfo
  {journal} {Phys. Rev. A}\ }\textbf {\bibinfo {volume} {98}},\ \bibinfo
  {pages} {012315} (\bibinfo {year} {2018}{\natexlab{b}})}\BibitemShut
  {NoStop}%
\bibitem [{\citenamefont {Hou}\ \emph {et~al.}(2020)\citenamefont {Hou},
  \citenamefont {Cao}, \citenamefont {Zhou},\ and\ \citenamefont
  {Zeng}}]{hou2020upper}%
  \BibitemOpen
  \bibfield  {author} {\bibinfo {author} {\bibfnamefont {S.-Y.}\ \bibnamefont
  {Hou}}, \bibinfo {author} {\bibfnamefont {C.}~\bibnamefont {Cao}}, \bibinfo
  {author} {\bibfnamefont {D.}~\bibnamefont {Zhou}},\ and\ \bibinfo {author}
  {\bibfnamefont {B.}~\bibnamefont {Zeng}},\ }\bibfield  {title} {\bibinfo
  {title} {Upper bounds for relative entropy of entanglement based on active
  learning},\ }\href {https://doi.org/10.1088/2058-9565/abb412} {\bibfield
  {journal} {\bibinfo  {journal} {Quantum Science and Technology}\ }\textbf
  {\bibinfo {volume} {5}},\ \bibinfo {pages} {045019} (\bibinfo {year}
  {2020})}\BibitemShut {NoStop}%
\bibitem [{\citenamefont {Long}\ \emph {et~al.}(2022)\citenamefont {Long},
  \citenamefont {He}, \citenamefont {Zhang}, \citenamefont {Tang},
  \citenamefont {Lin}, \citenamefont {Liu}, \citenamefont {Nie}, \citenamefont
  {Feng}, \citenamefont {Li}, \citenamefont {Xin}, \citenamefont {Ai},\ and\
  \citenamefont {Lu}}]{PhysRevLett.129.070502}%
  \BibitemOpen
  \bibfield  {author} {\bibinfo {author} {\bibfnamefont {X.}~\bibnamefont
  {Long}}, \bibinfo {author} {\bibfnamefont {W.-T.}\ \bibnamefont {He}},
  \bibinfo {author} {\bibfnamefont {N.-N.}\ \bibnamefont {Zhang}}, \bibinfo
  {author} {\bibfnamefont {K.}~\bibnamefont {Tang}}, \bibinfo {author}
  {\bibfnamefont {Z.}~\bibnamefont {Lin}}, \bibinfo {author} {\bibfnamefont
  {H.}~\bibnamefont {Liu}}, \bibinfo {author} {\bibfnamefont {X.}~\bibnamefont
  {Nie}}, \bibinfo {author} {\bibfnamefont {G.}~\bibnamefont {Feng}}, \bibinfo
  {author} {\bibfnamefont {J.}~\bibnamefont {Li}}, \bibinfo {author}
  {\bibfnamefont {T.}~\bibnamefont {Xin}}, \bibinfo {author} {\bibfnamefont
  {Q.}~\bibnamefont {Ai}},\ and\ \bibinfo {author} {\bibfnamefont
  {D.}~\bibnamefont {Lu}},\ }\bibfield  {title} {\bibinfo {title}
  {Entanglement-enhanced quantum metrology in colored noise by quantum {Z}eno
  effect},\ }\href {https://doi.org/10.1103/PhysRevLett.129.070502} {\bibfield
  {journal} {\bibinfo  {journal} {Phys. Rev. Lett.}\ }\textbf {\bibinfo
  {volume} {129}},\ \bibinfo {pages} {070502} (\bibinfo {year}
  {2022})}\BibitemShut {NoStop}%
\bibitem [{\citenamefont {Huang}\ \emph {et~al.}(2024)\citenamefont {Huang},
  \citenamefont {Xi}, \citenamefont {Long}, \citenamefont {Liu}, \citenamefont
  {Fan}, \citenamefont {Wang}, \citenamefont {Zheng}, \citenamefont {Feng},
  \citenamefont {Nie},\ and\ \citenamefont {Lu}}]{PhysRevLett.132.210403}%
  \BibitemOpen
  \bibfield  {author} {\bibinfo {author} {\bibfnamefont {K.}~\bibnamefont
  {Huang}}, \bibinfo {author} {\bibfnamefont {C.}~\bibnamefont {Xi}}, \bibinfo
  {author} {\bibfnamefont {X.}~\bibnamefont {Long}}, \bibinfo {author}
  {\bibfnamefont {H.}~\bibnamefont {Liu}}, \bibinfo {author} {\bibfnamefont
  {Y.-a.}\ \bibnamefont {Fan}}, \bibinfo {author} {\bibfnamefont
  {X.}~\bibnamefont {Wang}}, \bibinfo {author} {\bibfnamefont {Y.}~\bibnamefont
  {Zheng}}, \bibinfo {author} {\bibfnamefont {Y.}~\bibnamefont {Feng}},
  \bibinfo {author} {\bibfnamefont {X.}~\bibnamefont {Nie}},\ and\ \bibinfo
  {author} {\bibfnamefont {D.}~\bibnamefont {Lu}},\ }\bibfield  {title}
  {\bibinfo {title} {Experimental realization of self-contained quantum
  refrigeration},\ }\href {https://doi.org/10.1103/PhysRevLett.132.210403}
  {\bibfield  {journal} {\bibinfo  {journal} {Phys. Rev. Lett.}\ }\textbf
  {\bibinfo {volume} {132}},\ \bibinfo {pages} {210403} (\bibinfo {year}
  {2024})}\BibitemShut {NoStop}%
\bibitem [{\citenamefont {Xi}\ \emph {et~al.}(2024)\citenamefont {Xi},
  \citenamefont {Liu}, \citenamefont {Liu}, \citenamefont {Huang},
  \citenamefont {Long}, \citenamefont {Ebler}, \citenamefont {Nie},
  \citenamefont {Dahlsten},\ and\ \citenamefont {Lu}}]{PhysRevLett.133.040401}%
  \BibitemOpen
  \bibfield  {author} {\bibinfo {author} {\bibfnamefont {C.}~\bibnamefont
  {Xi}}, \bibinfo {author} {\bibfnamefont {X.}~\bibnamefont {Liu}}, \bibinfo
  {author} {\bibfnamefont {H.}~\bibnamefont {Liu}}, \bibinfo {author}
  {\bibfnamefont {K.}~\bibnamefont {Huang}}, \bibinfo {author} {\bibfnamefont
  {X.}~\bibnamefont {Long}}, \bibinfo {author} {\bibfnamefont {D.}~\bibnamefont
  {Ebler}}, \bibinfo {author} {\bibfnamefont {X.}~\bibnamefont {Nie}}, \bibinfo
  {author} {\bibfnamefont {O.}~\bibnamefont {Dahlsten}},\ and\ \bibinfo
  {author} {\bibfnamefont {D.}~\bibnamefont {Lu}},\ }\bibfield  {title}
  {\bibinfo {title} {Experimental validation of enhanced information capacity
  by quantum switch in accordance with thermodynamic laws},\ }\href
  {https://doi.org/10.1103/PhysRevLett.133.040401} {\bibfield  {journal}
  {\bibinfo  {journal} {Phys. Rev. Lett.}\ }\textbf {\bibinfo {volume} {133}},\
  \bibinfo {pages} {040401} (\bibinfo {year} {2024})}\BibitemShut {NoStop}%
\bibitem [{\citenamefont {Liu}\ \emph {et~al.}(2025{\natexlab{a}})\citenamefont
  {Liu}, \citenamefont {Hur}, \citenamefont {Zhang}, \citenamefont {Che},
  \citenamefont {Long}, \citenamefont {Wang}, \citenamefont {Huang},
  \citenamefont {Fan}, \citenamefont {Zheng}, \citenamefont {Feng},
  \citenamefont {Zhou}, \citenamefont {Ng}, \citenamefont {Nie}, \citenamefont
  {Park},\ and\ \citenamefont {Lu}}]{y8wr-yml4}%
  \BibitemOpen
  \bibfield  {author} {\bibinfo {author} {\bibfnamefont {H.}~\bibnamefont
  {Liu}}, \bibinfo {author} {\bibfnamefont {T.}~\bibnamefont {Hur}}, \bibinfo
  {author} {\bibfnamefont {S.}~\bibnamefont {Zhang}}, \bibinfo {author}
  {\bibfnamefont {L.}~\bibnamefont {Che}}, \bibinfo {author} {\bibfnamefont
  {X.}~\bibnamefont {Long}}, \bibinfo {author} {\bibfnamefont {X.}~\bibnamefont
  {Wang}}, \bibinfo {author} {\bibfnamefont {K.}~\bibnamefont {Huang}},
  \bibinfo {author} {\bibfnamefont {Y.-a.}\ \bibnamefont {Fan}}, \bibinfo
  {author} {\bibfnamefont {Y.}~\bibnamefont {Zheng}}, \bibinfo {author}
  {\bibfnamefont {Y.}~\bibnamefont {Feng}}, \bibinfo {author} {\bibfnamefont
  {Y.}~\bibnamefont {Zhou}}, \bibinfo {author} {\bibfnamefont {J.}~\bibnamefont
  {Ng}}, \bibinfo {author} {\bibfnamefont {X.}~\bibnamefont {Nie}}, \bibinfo
  {author} {\bibfnamefont {D.~K.}\ \bibnamefont {Park}},\ and\ \bibinfo
  {author} {\bibfnamefont {D.}~\bibnamefont {Lu}},\ }\bibfield  {title}
  {\bibinfo {title} {Neural quantum embedding via deterministic quantum
  computation with one qubit},\ }\href {https://doi.org/10.1103/y8wr-yml4}
  {\bibfield  {journal} {\bibinfo  {journal} {Phys. Rev. Lett.}\ }\textbf
  {\bibinfo {volume} {135}},\ \bibinfo {pages} {080603} (\bibinfo {year}
  {2025}{\natexlab{a}})}\BibitemShut {NoStop}%
\bibitem [{\citenamefont {Huang}\ \emph {et~al.}(2025)\citenamefont {Huang},
  \citenamefont {Che}, \citenamefont {Wei}, \citenamefont {Xu}, \citenamefont
  {Nie}, \citenamefont {Li}, \citenamefont {Lu},\ and\ \citenamefont
  {Xin}}]{huang2025direct}%
  \BibitemOpen
  \bibfield  {author} {\bibinfo {author} {\bibfnamefont {Y.}~\bibnamefont
  {Huang}}, \bibinfo {author} {\bibfnamefont {L.}~\bibnamefont {Che}}, \bibinfo
  {author} {\bibfnamefont {C.}~\bibnamefont {Wei}}, \bibinfo {author}
  {\bibfnamefont {F.}~\bibnamefont {Xu}}, \bibinfo {author} {\bibfnamefont
  {X.}~\bibnamefont {Nie}}, \bibinfo {author} {\bibfnamefont {J.}~\bibnamefont
  {Li}}, \bibinfo {author} {\bibfnamefont {D.}~\bibnamefont {Lu}},\ and\
  \bibinfo {author} {\bibfnamefont {T.}~\bibnamefont {Xin}},\ }\bibfield
  {title} {\bibinfo {title} {Direct entanglement detection of quantum systems
  using machine learning},\ }\href {https://doi.org/10.1038/s41534-025-00970-w}
  {\bibfield  {journal} {\bibinfo  {journal} {npj Quantum Inf}\ }\textbf
  {\bibinfo {volume} {11}},\ \bibinfo {pages} {29} (\bibinfo {year}
  {2025})}\BibitemShut {NoStop}%
\bibitem [{\citenamefont {Khaneja}\ \emph {et~al.}(2005)\citenamefont
  {Khaneja}, \citenamefont {Reiss}, \citenamefont {Kehlet}, \citenamefont
  {Schulte-Herbrüggen},\ and\ \citenamefont {Glaser}}]{KHANEJA2005296}%
  \BibitemOpen
  \bibfield  {author} {\bibinfo {author} {\bibfnamefont {N.}~\bibnamefont
  {Khaneja}}, \bibinfo {author} {\bibfnamefont {T.}~\bibnamefont {Reiss}},
  \bibinfo {author} {\bibfnamefont {C.}~\bibnamefont {Kehlet}}, \bibinfo
  {author} {\bibfnamefont {T.}~\bibnamefont {Schulte-Herbrüggen}},\ and\
  \bibinfo {author} {\bibfnamefont {S.~J.}\ \bibnamefont {Glaser}},\ }\bibfield
   {title} {\bibinfo {title} {Optimal control of coupled spin dynamics: design
  of nmr pulse sequences by gradient ascent algorithms},\ }\href
  {https://doi.org/https://doi.org/10.1016/j.jmr.2004.11.004} {\bibfield
  {journal} {\bibinfo  {journal} {Journal of Magnetic Resonance}\ }\textbf
  {\bibinfo {volume} {172}},\ \bibinfo {pages} {296} (\bibinfo {year}
  {2005})}\BibitemShut {NoStop}%
\bibitem [{\citenamefont {Horodecki}(1997)}]{horodecki1997separability}%
  \BibitemOpen
  \bibfield  {author} {\bibinfo {author} {\bibfnamefont {P.}~\bibnamefont
  {Horodecki}},\ }\bibfield  {title} {\bibinfo {title} {Separability criterion
  and inseparable mixed states with positive partial transposition},\ }\href
  {https://doi.org/10.1016/S0375-9601(97)00416-7} {\bibfield  {journal}
  {\bibinfo  {journal} {Phys. Lett. A}\ }\textbf {\bibinfo {volume} {232}},\
  \bibinfo {pages} {333} (\bibinfo {year} {1997})}\BibitemShut {NoStop}%
\bibitem [{\citenamefont {G\"uhne}\ \emph {et~al.}(2021)\citenamefont
  {G\"uhne}, \citenamefont {Mao},\ and\ \citenamefont
  {Yu}}]{PhysRevLett.126.140503}%
  \BibitemOpen
  \bibfield  {author} {\bibinfo {author} {\bibfnamefont {O.}~\bibnamefont
  {G\"uhne}}, \bibinfo {author} {\bibfnamefont {Y.}~\bibnamefont {Mao}},\ and\
  \bibinfo {author} {\bibfnamefont {X.-D.}\ \bibnamefont {Yu}},\ }\bibfield
  {title} {\bibinfo {title} {Geometry of faithful entanglement},\ }\href
  {https://doi.org/10.1103/PhysRevLett.126.140503} {\bibfield  {journal}
  {\bibinfo  {journal} {Phys. Rev. Lett.}\ }\textbf {\bibinfo {volume} {126}},\
  \bibinfo {pages} {140503} (\bibinfo {year} {2021})}\BibitemShut {NoStop}%
\bibitem [{\citenamefont {Horodecki}\ \emph {et~al.}(1999)\citenamefont
  {Horodecki}, \citenamefont {Horodecki},\ and\ \citenamefont
  {Horodecki}}]{PhysRevLett.82.1056}%
  \BibitemOpen
  \bibfield  {author} {\bibinfo {author} {\bibfnamefont {P.}~\bibnamefont
  {Horodecki}}, \bibinfo {author} {\bibfnamefont {M.}~\bibnamefont
  {Horodecki}},\ and\ \bibinfo {author} {\bibfnamefont {R.}~\bibnamefont
  {Horodecki}},\ }\bibfield  {title} {\bibinfo {title} {Bound entanglement can
  be activated},\ }\href {https://doi.org/10.1103/PhysRevLett.82.1056}
  {\bibfield  {journal} {\bibinfo  {journal} {Phys. Rev. Lett.}\ }\textbf
  {\bibinfo {volume} {82}},\ \bibinfo {pages} {1056} (\bibinfo {year}
  {1999})}\BibitemShut {NoStop}%
\bibitem [{not()}]{note1}%
  \BibitemOpen
  \href@noop {} {}\bibinfo {note} {Horodecki state $\rho_b$ refers to the
  expression given in Eq.~(4) in Ref.~\cite{PhysRevLett.82.1056}. Notice that
  $\rho_b$ is invariant under the change of $b \to 5-b$ and interchange of the
  parties. So that the original analysis for $b \in [2,5]$ in
  Ref.~\cite{PhysRevLett.82.1056} can be extended to the whole range of $b \in
  [0,5]$}\BibitemShut {NoStop}%
\bibitem [{\citenamefont {Bennett}\ \emph {et~al.}(1999)\citenamefont
  {Bennett}, \citenamefont {DiVincenzo}, \citenamefont {Mor}, \citenamefont
  {Shor}, \citenamefont {Smolin},\ and\ \citenamefont
  {Terhal}}]{PhysRevLett.82.5385}%
  \BibitemOpen
  \bibfield  {author} {\bibinfo {author} {\bibfnamefont {C.~H.}\ \bibnamefont
  {Bennett}}, \bibinfo {author} {\bibfnamefont {D.~P.}\ \bibnamefont
  {DiVincenzo}}, \bibinfo {author} {\bibfnamefont {T.}~\bibnamefont {Mor}},
  \bibinfo {author} {\bibfnamefont {P.~W.}\ \bibnamefont {Shor}}, \bibinfo
  {author} {\bibfnamefont {J.~A.}\ \bibnamefont {Smolin}},\ and\ \bibinfo
  {author} {\bibfnamefont {B.~M.}\ \bibnamefont {Terhal}},\ }\bibfield  {title}
  {\bibinfo {title} {Unextendible product bases and bound entanglement},\
  }\href {https://doi.org/10.1103/PhysRevLett.82.5385} {\bibfield  {journal}
  {\bibinfo  {journal} {Phys. Rev. Lett.}\ }\textbf {\bibinfo {volume} {82}},\
  \bibinfo {pages} {5385} (\bibinfo {year} {1999})}\BibitemShut {NoStop}%
\bibitem [{\citenamefont {Pandya}\ \emph {et~al.}(2020)\citenamefont {Pandya},
  \citenamefont {Sakarya},\ and\ \citenamefont {Wie\ifmmode~\acute{s}\else
  \'{s}\fi{}niak}}]{PhysRevA.102.012409}%
  \BibitemOpen
  \bibfield  {author} {\bibinfo {author} {\bibfnamefont {P.}~\bibnamefont
  {Pandya}}, \bibinfo {author} {\bibfnamefont {O.}~\bibnamefont {Sakarya}},\
  and\ \bibinfo {author} {\bibfnamefont {M.}~\bibnamefont
  {Wie\ifmmode~\acute{s}\else \'{s}\fi{}niak}},\ }\bibfield  {title} {\bibinfo
  {title} {Hilbert-schmidt distance and entanglement witnessing},\ }\href
  {https://doi.org/10.1103/PhysRevA.102.012409} {\bibfield  {journal} {\bibinfo
   {journal} {Phys. Rev. A}\ }\textbf {\bibinfo {volume} {102}},\ \bibinfo
  {pages} {012409} (\bibinfo {year} {2020})}\BibitemShut {NoStop}%
\bibitem [{\citenamefont {Zhang}\ \emph {et~al.}(2024)\citenamefont {Zhang},
  \citenamefont {Denker}, \citenamefont {Asadian},\ and\ \citenamefont
  {G\"uhne}}]{PhysRevLett.133.040203}%
  \BibitemOpen
  \bibfield  {author} {\bibinfo {author} {\bibfnamefont {C.}~\bibnamefont
  {Zhang}}, \bibinfo {author} {\bibfnamefont {S.}~\bibnamefont {Denker}},
  \bibinfo {author} {\bibfnamefont {A.}~\bibnamefont {Asadian}},\ and\ \bibinfo
  {author} {\bibfnamefont {O.}~\bibnamefont {G\"uhne}},\ }\bibfield  {title}
  {\bibinfo {title} {Analyzing quantum entanglement with the schmidt
  decomposition in operator space},\ }\href
  {https://doi.org/10.1103/PhysRevLett.133.040203} {\bibfield  {journal}
  {\bibinfo  {journal} {Phys. Rev. Lett.}\ }\textbf {\bibinfo {volume} {133}},\
  \bibinfo {pages} {040203} (\bibinfo {year} {2024})}\BibitemShut {NoStop}%
\bibitem [{\citenamefont {Horodecki}\ and\ \citenamefont
  {Horodecki}(1999)}]{PhysRevA.59.4206}%
  \BibitemOpen
  \bibfield  {author} {\bibinfo {author} {\bibfnamefont {M.}~\bibnamefont
  {Horodecki}}\ and\ \bibinfo {author} {\bibfnamefont {P.}~\bibnamefont
  {Horodecki}},\ }\bibfield  {title} {\bibinfo {title} {Reduction criterion of
  separability and limits for a class of distillation protocols},\ }\href
  {https://doi.org/10.1103/PhysRevA.59.4206} {\bibfield  {journal} {\bibinfo
  {journal} {Phys. Rev. A}\ }\textbf {\bibinfo {volume} {59}},\ \bibinfo
  {pages} {4206} (\bibinfo {year} {1999})}\BibitemShut {NoStop}%
\bibitem [{\citenamefont {Weilenmann}\ \emph {et~al.}(2020)\citenamefont
  {Weilenmann}, \citenamefont {Dive}, \citenamefont {Trillo}, \citenamefont
  {Aguilar},\ and\ \citenamefont {Navascu\'es}}]{PhysRevLett.124.200502}%
  \BibitemOpen
  \bibfield  {author} {\bibinfo {author} {\bibfnamefont {M.}~\bibnamefont
  {Weilenmann}}, \bibinfo {author} {\bibfnamefont {B.}~\bibnamefont {Dive}},
  \bibinfo {author} {\bibfnamefont {D.}~\bibnamefont {Trillo}}, \bibinfo
  {author} {\bibfnamefont {E.~A.}\ \bibnamefont {Aguilar}},\ and\ \bibinfo
  {author} {\bibfnamefont {M.}~\bibnamefont {Navascu\'es}},\ }\bibfield
  {title} {\bibinfo {title} {Entanglement detection beyond measuring
  fidelities},\ }\href {https://doi.org/10.1103/PhysRevLett.124.200502}
  {\bibfield  {journal} {\bibinfo  {journal} {Phys. Rev. Lett.}\ }\textbf
  {\bibinfo {volume} {124}},\ \bibinfo {pages} {200502} (\bibinfo {year}
  {2020})}\BibitemShut {NoStop}%
\bibitem [{\citenamefont {Jungnitsch}\ \emph {et~al.}(2011)\citenamefont
  {Jungnitsch}, \citenamefont {Moroder},\ and\ \citenamefont
  {G\"uhne}}]{PhysRevLett.106.190502}%
  \BibitemOpen
  \bibfield  {author} {\bibinfo {author} {\bibfnamefont {B.}~\bibnamefont
  {Jungnitsch}}, \bibinfo {author} {\bibfnamefont {T.}~\bibnamefont
  {Moroder}},\ and\ \bibinfo {author} {\bibfnamefont {O.}~\bibnamefont
  {G\"uhne}},\ }\bibfield  {title} {\bibinfo {title} {Taming multiparticle
  entanglement},\ }\href {https://doi.org/10.1103/PhysRevLett.106.190502}
  {\bibfield  {journal} {\bibinfo  {journal} {Phys. Rev. Lett.}\ }\textbf
  {\bibinfo {volume} {106}},\ \bibinfo {pages} {190502} (\bibinfo {year}
  {2011})}\BibitemShut {NoStop}%
\bibitem [{\citenamefont {Chitambar}\ and\ \citenamefont
  {Gour}(2019)}]{RevModPhys.91.025001}%
  \BibitemOpen
  \bibfield  {author} {\bibinfo {author} {\bibfnamefont {E.}~\bibnamefont
  {Chitambar}}\ and\ \bibinfo {author} {\bibfnamefont {G.}~\bibnamefont
  {Gour}},\ }\bibfield  {title} {\bibinfo {title} {Quantum resource theories},\
  }\href {https://doi.org/10.1103/RevModPhys.91.025001} {\bibfield  {journal}
  {\bibinfo  {journal} {Rev. Mod. Phys.}\ }\textbf {\bibinfo {volume} {91}},\
  \bibinfo {pages} {025001} (\bibinfo {year} {2019})}\BibitemShut {NoStop}%
\bibitem [{\citenamefont {Zhu}\ \emph {et~al.}(2025)\citenamefont {Zhu},
  \citenamefont {Zhang}, \citenamefont {An},\ and\ \citenamefont
  {Zeng}}]{zhu2025unified}%
  \BibitemOpen
  \bibfield  {author} {\bibinfo {author} {\bibfnamefont {X.}~\bibnamefont
  {Zhu}}, \bibinfo {author} {\bibfnamefont {C.}~\bibnamefont {Zhang}}, \bibinfo
  {author} {\bibfnamefont {Z.}~\bibnamefont {An}},\ and\ \bibinfo {author}
  {\bibfnamefont {B.}~\bibnamefont {Zeng}},\ }\bibfield  {title} {\bibinfo
  {title} {Unified framework for calculating convex roof resource measures},\
  }\href {https://doi.org/10.1038/s41534-025-01012-1} {\bibfield  {journal}
  {\bibinfo  {journal} {npj Quantum Inf}\ }\textbf {\bibinfo {volume} {11}},\
  \bibinfo {pages} {56} (\bibinfo {year} {2025})}\BibitemShut {NoStop}%
\bibitem [{rep()}]{repo}%
  \BibitemOpen
  \href@noop {} {}\bibinfo {note} {The code and data for this work are
  open-source and available at
  \url{https://github.com/numqi/geometric-witness.git}.}\BibitemShut {Stop}%
\bibitem [{\citenamefont {Wu}(2004)}]{wu2004convex}%
  \BibitemOpen
  \bibfield  {author} {\bibinfo {author} {\bibfnamefont {S.}~\bibnamefont
  {Wu}},\ }\bibfield  {title} {\bibinfo {title} {The convex sum of product
  states for a separable state},\ }\href
  {https://doi.org/10.1016/j.physleta.2003.12.046} {\bibfield  {journal}
  {\bibinfo  {journal} {Phys. Lett. A}\ }\textbf {\bibinfo {volume} {321}},\
  \bibinfo {pages} {301} (\bibinfo {year} {2004})}\BibitemShut {NoStop}%
\bibitem [{\citenamefont {Ha}\ and\ \citenamefont
  {Kye}(2014)}]{ha2014separable}%
  \BibitemOpen
  \bibfield  {author} {\bibinfo {author} {\bibfnamefont {K.-C.}\ \bibnamefont
  {Ha}}\ and\ \bibinfo {author} {\bibfnamefont {S.-H.}\ \bibnamefont {Kye}},\
  }\bibfield  {title} {\bibinfo {title} {Separable states with unique
  decompositions},\ }\href {https://doi.org/10.1007/s00220-014-1980-6}
  {\bibfield  {journal} {\bibinfo  {journal} {Commun. Math. Phys.}\ }\textbf
  {\bibinfo {volume} {328}},\ \bibinfo {pages} {131} (\bibinfo {year}
  {2014})}\BibitemShut {NoStop}%
\bibitem [{\citenamefont {Chen}\ and\ \citenamefont
  {{\DH}okovi{\'c}}(2013)}]{chen2013dimensions}%
  \BibitemOpen
  \bibfield  {author} {\bibinfo {author} {\bibfnamefont {L.}~\bibnamefont
  {Chen}}\ and\ \bibinfo {author} {\bibfnamefont {D.~{\v{Z}}.}\ \bibnamefont
  {{\DH}okovi{\'c}}},\ }\bibfield  {title} {\bibinfo {title} {Dimensions,
  lengths, and separability in finite-dimensional quantum systems},\
  }\href@noop {} {\bibfield  {journal} {\bibinfo  {journal} {Journal of
  Mathematical Physics}\ }\textbf {\bibinfo {volume} {54}} (\bibinfo {year}
  {2013})}\BibitemShut {NoStop}%
\bibitem [{\citenamefont {Smolin}(2001)}]{PhysRevA.63.032306}%
  \BibitemOpen
  \bibfield  {author} {\bibinfo {author} {\bibfnamefont {J.~A.}\ \bibnamefont
  {Smolin}},\ }\bibfield  {title} {\bibinfo {title} {Four-party unlockable
  bound entangled state},\ }\href {https://doi.org/10.1103/PhysRevA.63.032306}
  {\bibfield  {journal} {\bibinfo  {journal} {Phys. Rev. A}\ }\textbf {\bibinfo
  {volume} {63}},\ \bibinfo {pages} {032306} (\bibinfo {year}
  {2001})}\BibitemShut {NoStop}%
\bibitem [{\citenamefont {Augusiak}\ and\ \citenamefont
  {Horodecki}(2006)}]{PhysRevA.74.010305}%
  \BibitemOpen
  \bibfield  {author} {\bibinfo {author} {\bibfnamefont {R.}~\bibnamefont
  {Augusiak}}\ and\ \bibinfo {author} {\bibfnamefont {P.}~\bibnamefont
  {Horodecki}},\ }\bibfield  {title} {\bibinfo {title} {Bound entanglement
  maximally violating bell inequalities: Quantum entanglement is not fully
  equivalent to cryptographic security},\ }\href
  {https://doi.org/10.1103/PhysRevA.74.010305} {\bibfield  {journal} {\bibinfo
  {journal} {Phys. Rev. A}\ }\textbf {\bibinfo {volume} {74}},\ \bibinfo
  {pages} {010305} (\bibinfo {year} {2006})}\BibitemShut {NoStop}%
\bibitem [{\citenamefont {Nie}\ \emph {et~al.}(2022)\citenamefont {Nie},
  \citenamefont {Zhu}, \citenamefont {Huang}, \citenamefont {Tang},
  \citenamefont {Long}, \citenamefont {Lin}, \citenamefont {Tian},
  \citenamefont {Qiu}, \citenamefont {Xi}, \citenamefont {Yang}, \citenamefont
  {Li}, \citenamefont {Dong}, \citenamefont {Xin},\ and\ \citenamefont
  {Lu}}]{PhysRevLett.129.100603}%
  \BibitemOpen
  \bibfield  {author} {\bibinfo {author} {\bibfnamefont {X.}~\bibnamefont
  {Nie}}, \bibinfo {author} {\bibfnamefont {X.}~\bibnamefont {Zhu}}, \bibinfo
  {author} {\bibfnamefont {K.}~\bibnamefont {Huang}}, \bibinfo {author}
  {\bibfnamefont {K.}~\bibnamefont {Tang}}, \bibinfo {author} {\bibfnamefont
  {X.}~\bibnamefont {Long}}, \bibinfo {author} {\bibfnamefont {Z.}~\bibnamefont
  {Lin}}, \bibinfo {author} {\bibfnamefont {Y.}~\bibnamefont {Tian}}, \bibinfo
  {author} {\bibfnamefont {C.}~\bibnamefont {Qiu}}, \bibinfo {author}
  {\bibfnamefont {C.}~\bibnamefont {Xi}}, \bibinfo {author} {\bibfnamefont
  {X.}~\bibnamefont {Yang}}, \bibinfo {author} {\bibfnamefont {J.}~\bibnamefont
  {Li}}, \bibinfo {author} {\bibfnamefont {Y.}~\bibnamefont {Dong}}, \bibinfo
  {author} {\bibfnamefont {T.}~\bibnamefont {Xin}},\ and\ \bibinfo {author}
  {\bibfnamefont {D.}~\bibnamefont {Lu}},\ }\bibfield  {title} {\bibinfo
  {title} {Experimental realization of a quantum refrigerator driven by
  indefinite causal orders},\ }\href
  {https://doi.org/10.1103/PhysRevLett.129.100603} {\bibfield  {journal}
  {\bibinfo  {journal} {Phys. Rev. Lett.}\ }\textbf {\bibinfo {volume} {129}},\
  \bibinfo {pages} {100603} (\bibinfo {year} {2022})}\BibitemShut {NoStop}%
\bibitem [{\citenamefont {Fan}\ \emph {et~al.}(2023)\citenamefont {Fan},
  \citenamefont {Li}, \citenamefont {Hu}, \citenamefont {Li}, \citenamefont
  {Long}, \citenamefont {Liu}, \citenamefont {Yang}, \citenamefont {Nie},
  \citenamefont {Li}, \citenamefont {Xin}, \citenamefont {Lu},\ and\
  \citenamefont {Wan}}]{Fan2023Innovation}%
  \BibitemOpen
  \bibfield  {author} {\bibinfo {author} {\bibfnamefont {Y.-a.}\ \bibnamefont
  {Fan}}, \bibinfo {author} {\bibfnamefont {Y.}~\bibnamefont {Li}}, \bibinfo
  {author} {\bibfnamefont {Y.}~\bibnamefont {Hu}}, \bibinfo {author}
  {\bibfnamefont {Y.}~\bibnamefont {Li}}, \bibinfo {author} {\bibfnamefont
  {X.}~\bibnamefont {Long}}, \bibinfo {author} {\bibfnamefont {H.}~\bibnamefont
  {Liu}}, \bibinfo {author} {\bibfnamefont {X.}~\bibnamefont {Yang}}, \bibinfo
  {author} {\bibfnamefont {X.}~\bibnamefont {Nie}}, \bibinfo {author}
  {\bibfnamefont {J.}~\bibnamefont {Li}}, \bibinfo {author} {\bibfnamefont
  {T.}~\bibnamefont {Xin}}, \bibinfo {author} {\bibfnamefont {D.}~\bibnamefont
  {Lu}},\ and\ \bibinfo {author} {\bibfnamefont {Y.}~\bibnamefont {Wan}},\
  }\bibfield  {title} {\bibinfo {title} {Experimental quantum simulation of a
  topologically protected hadamard gate via braiding fibonacci anyons},\ }\href
  {https://doi.org/10.1016/j.xinn.2023.100480} {\bibfield  {journal} {\bibinfo
  {journal} {The Innovation}\ }\textbf {\bibinfo {volume} {4}},\ \bibinfo
  {pages} {100480} (\bibinfo {year} {2023})}\BibitemShut {NoStop}%
\bibitem [{\citenamefont {Xin}\ \emph {et~al.}(2021)\citenamefont {Xin},
  \citenamefont {Che}, \citenamefont {Xi}, \citenamefont {Singh}, \citenamefont
  {Nie}, \citenamefont {Li}, \citenamefont {Dong},\ and\ \citenamefont
  {Lu}}]{PhysRevLett.126.110502}%
  \BibitemOpen
  \bibfield  {author} {\bibinfo {author} {\bibfnamefont {T.}~\bibnamefont
  {Xin}}, \bibinfo {author} {\bibfnamefont {L.}~\bibnamefont {Che}}, \bibinfo
  {author} {\bibfnamefont {C.}~\bibnamefont {Xi}}, \bibinfo {author}
  {\bibfnamefont {A.}~\bibnamefont {Singh}}, \bibinfo {author} {\bibfnamefont
  {X.}~\bibnamefont {Nie}}, \bibinfo {author} {\bibfnamefont {J.}~\bibnamefont
  {Li}}, \bibinfo {author} {\bibfnamefont {Y.}~\bibnamefont {Dong}},\ and\
  \bibinfo {author} {\bibfnamefont {D.}~\bibnamefont {Lu}},\ }\bibfield
  {title} {\bibinfo {title} {Experimental quantum principal component analysis
  via parametrized quantum circuits},\ }\href
  {https://doi.org/10.1103/PhysRevLett.126.110502} {\bibfield  {journal}
  {\bibinfo  {journal} {Phys. Rev. Lett.}\ }\textbf {\bibinfo {volume} {126}},\
  \bibinfo {pages} {110502} (\bibinfo {year} {2021})}\BibitemShut {NoStop}%
\bibitem [{\citenamefont {M{\"o}tt{\"o}nen}\ \emph {et~al.}(2005)\citenamefont
  {M{\"o}tt{\"o}nen}, \citenamefont {Vartiainen}, \citenamefont {Bergholm},\
  and\ \citenamefont {Salomaa}}]{Mottonen2005Transformation}%
  \BibitemOpen
  \bibfield  {author} {\bibinfo {author} {\bibfnamefont {M.}~\bibnamefont
  {M{\"o}tt{\"o}nen}}, \bibinfo {author} {\bibfnamefont {J.~J.}\ \bibnamefont
  {Vartiainen}}, \bibinfo {author} {\bibfnamefont {V.}~\bibnamefont
  {Bergholm}},\ and\ \bibinfo {author} {\bibfnamefont {M.~M.}\ \bibnamefont
  {Salomaa}},\ }\bibfield  {title} {\bibinfo {title} {Transformation of quantum
  states using uniformly controlled rotations},\ }\href@noop {} {\bibfield
  {journal} {\bibinfo  {journal} {Quantum Inf. Comput.}\ }\textbf {\bibinfo
  {volume} {5}},\ \bibinfo {pages} {467} (\bibinfo {year} {2005})},\ \Eprint
  {https://arxiv.org/abs/quant-ph/0407010} {quant-ph/0407010} \BibitemShut
  {NoStop}%
\bibitem [{\citenamefont {Liu}\ \emph {et~al.}(2025{\natexlab{b}})\citenamefont
  {Liu}, \citenamefont {Liu}, \citenamefont {Chen}, \citenamefont {Nie},
  \citenamefont {Liu},\ and\ \citenamefont {Lu}}]{PhysRevLett.134.040201}%
  \BibitemOpen
  \bibfield  {author} {\bibinfo {author} {\bibfnamefont {H.}~\bibnamefont
  {Liu}}, \bibinfo {author} {\bibfnamefont {Z.}~\bibnamefont {Liu}}, \bibinfo
  {author} {\bibfnamefont {S.}~\bibnamefont {Chen}}, \bibinfo {author}
  {\bibfnamefont {X.}~\bibnamefont {Nie}}, \bibinfo {author} {\bibfnamefont
  {X.}~\bibnamefont {Liu}},\ and\ \bibinfo {author} {\bibfnamefont
  {D.}~\bibnamefont {Lu}},\ }\bibfield  {title} {\bibinfo {title} {Certifying
  quantum temporal correlation via randomized measurements: Theory and
  experiment},\ }\href {https://doi.org/10.1103/PhysRevLett.134.040201}
  {\bibfield  {journal} {\bibinfo  {journal} {Phys. Rev. Lett.}\ }\textbf
  {\bibinfo {volume} {134}},\ \bibinfo {pages} {040201} (\bibinfo {year}
  {2025}{\natexlab{b}})}\BibitemShut {NoStop}%
\bibitem [{\citenamefont {Leskowitz}\ and\ \citenamefont
  {Mueller}(2004)}]{Leskowitz2004StateInterrogation}%
  \BibitemOpen
  \bibfield  {author} {\bibinfo {author} {\bibfnamefont {G.~M.}\ \bibnamefont
  {Leskowitz}}\ and\ \bibinfo {author} {\bibfnamefont {L.~J.}\ \bibnamefont
  {Mueller}},\ }\bibfield  {title} {\bibinfo {title} {State interrogation in
  nuclear magnetic resonance quantum-information processing},\ }\href
  {https://doi.org/10.1103/PhysRevA.69.052302} {\bibfield  {journal} {\bibinfo
  {journal} {Phys. Rev. A}\ }\textbf {\bibinfo {volume} {69}},\ \bibinfo
  {pages} {052302} (\bibinfo {year} {2004})}\BibitemShut {NoStop}%
\end{thebibliography}%

\end{document}